\pdfoutput=1
\documentclass[]{jkpaper}
\usepackage{tikz}
\usetikzlibrary{decorations.markings}
\usetikzlibrary{decorations.pathreplacing}
\usetikzlibrary{arrows.meta}
\usetikzlibrary{patterns}
\usetikzlibrary{decorations.pathmorphing}
\usepackage{tikz-cd}
\usepackage{bigints}
\usepackage{biblatex}
\usepackage{aas_macros}
\usepackage{soul}

\allowdisplaybreaks

\makeatletter
\newcommand\footnoteref[1]{\protected@xdef\@thefnmark{\ref{#1}}\@footnotemark}
\makeatother

\theoremstyle{plain}
\newtheorem{lemma}{Lemma}[section]

\newauthornote{\jjvk}{JJVK}{blue!80}

\title{Generalised second law beyond\texorpdfstring{\\}{ }the semiclassical regime}
\author{Josh Kirklin}
\institution{Perimeter Institute for Theoretical Physics,\texorpdfstring{\\}{ }31 Caroline Street North, Waterloo, ON, N2L 2Y5, Canada}

\email{
  \emaillink{jkirklin@pitp.ca}
}

\abstr{
  We prove that the generalised second law (GSL), with an appropriate modification, holds in perturbative gravity to all orders beyond the semiclassical limit and without a UV cutoff imposed on the fields. Our proof uses algebraic techniques and builds on the recent work of Faulkner and Speranza, which combined Wall's proof of the GSL with the identification of generalised entropy as the von Neumann entropy of a boost-invariant crossed product algebra. The key additional step in our approach is to further impose invariance under null translations. Doing so requires one to describe horizon exterior regions in a relational manner, so we introduce `dynamical cuts': quantum reference frames which give the location of a cut of the horizon. We use idealised dynamical cuts, but expect that our methods can be generalised to more realistic models. The modified GSL that we prove says that the difference in generalised entropies of the regions outside two dynamical cuts is bounded below by the free energy of the degrees of freedom giving the location of the later cut. If one takes a semiclassical limit, imposes a UV cutoff, and requires the cuts to obey certain energy conditions, then our result reduces to the standard GSL.
}

\bibliography{refs}
\begin{document}
\maketitleandtoc

\section{Introduction}
\label{Section: Introduction}

The generalised second law (GSL) has taught us much about the fundamental nature of quantum gravity. It states that the generalised entropy
\begin{equation}
  S_{\text{gen}} = \frac{A}{4 G_{\text{N}}} + S_{\text{out}}
  \label{Equation: S gen}
\end{equation}
of a black hole horizon cannot decrease over time. Here, $A$ is the area of some `cut' (i.e.\ spacelike slice) of the horizon, and $S_{\text{out}}$ is the entropy of the degrees of freedom in the region exterior to the cut. The GSL was originally proposed by Bekenstein~\cite{Bekenstein} on thermodynamical grounds, but if one views $S_{\text{out}}$ as an entanglement entropy, then one arrives at the popular interpretation of $S_{\text{gen}}$ as a genuine statistical measure of the state of microscopic gravitational degrees of freedom~\cite{Susskind_1994,jacobson1994blackholeentropyinduced,FROLOV1997339,Frolov_1998}; this is the interpretation that is most relevant for this paper.

There is a long history of proofs of the GSL~\cite{TenProofs,ZurekThorne,ThorneZurekPrice,SEWELL1987309,Wald:1995yp,FrolovPage,Sorkin:1986mg,Sorkin:1997ja,Mukohyama_1997,Wall:2011hj,Faulkner:2024gst}, with varying degrees of rigour and differing sets of assumptions -- but they have to our knowledge almost exclusively applied within \emph{semiclassical, perturbative} gravity. In this paper, we will go beyond the semiclassical regime, proving that the GSL, with an appropriate and modest modification, holds in fully quantum (but still perturbative) gravity. Our proposed modified GSL states that the increase in generalised entropy between two cuts is bounded below by the (possibly negative) \emph{free energy} of a certain set of reference degrees of freedom quantifying the relationship between the cuts. These degrees of freedom are hidden in the semiclassical regime, and accordingly our modified GSL reduces to the standard GSL when one takes an appropriate semiclassical limit. But, for reasons we will explain, the free energy term is required in the full theory for consistency with gravitational gauge invariance.

It is worth making clear from the outset exactly what we mean by the `semiclassical regime'. It is a description of gravitational physics in terms of local quantum degrees of freedom evolving on a classical manifold. Broadly speaking, there are two reasons this can be only an approximation (and one which need not remain valid during the full course of the horizon's evolution -- see for example Figure~\ref{Figure: leaving and entering semiclassical regime}, which will be explained in more detail below). First, as is widely appreciated, any formulation of spacetime in terms of a smooth manifold will break down at the Planck scale. In perturbative gravity, one basically sweeps this issue under the rug: the perturbations to the metric are regarded as quantum fields and do not account for Planck scale physics. But the second issue with the semiclassical regime is relevant even in perturbation theory: spacetime diffeomorphisms are gravitational gauge symmetries, and the value of a field at any fixed spacetime point is not diffeomorphism-invariant, so there \emph{are} no precisely local physical degrees of freedom.\footnote{In perturbation theory, we do have a local description \emph{before} imposing gauge invariance. But one must impose gauge invariance to obtain an adequate physical theory.} When it comes to the GSL, this is of key importance, because the term $S_{\text{out}}$ in~\eqref{Equation: S gen} is supposed to be the entropy of a local algebra of operators.

Thankfully, there are plenty of states in gravity for which a local picture \emph{emerges}, so that the semiclassical regime becomes a valid approximation, and~\eqref{Equation: S gen} makes sense. The key feature of such states is that they contain well-behaved \emph{dynamical reference frames}, which may be used to formulate an adequate notion of \emph{relational} locality~\cite{Goeller:2022rsx}. Such reference frames can be constructed in a variety of different ways, depending on the state. The semiclassical regime applies to perturbative gravity only when one restricts to states containing reference frames whose fluctuations are sufficiently suppressed to permit the emergence of effectively local degrees of freedom.\footnote{Often, a limit of small $G_{\text{N}}$ is also associated with the semiclassical regime. In the present language this can be understood as related to assuming that the reference frame does not significantly backreact on other degrees of freedom.} In this case, $S_{\text{out}}$ in~\eqref{Equation: S gen} may be understood as the entropy of these degrees of freedom.

Dynamical reference frames are typically not fundamental objects, instead being built out of other underlying dynamical degrees of freedom. They are important for a physical interpretation of locality even in classical gravity~\cite{Carrozza:2021gju,Carrozza:2022xut,Goeller:2022rsx}. But they are fundamentally more subtle in the quantum theory, because they must be made out of \emph{quantum} degrees of freedom. Such \emph{quantum reference frames} (QRFs) are thus subject to all the usual quantum phenomena -- they obey the uncertainty principle, they can be in superpositions, they can be entangled, and so on. They have been the subject of growing interest in various contexts~\cite{Aharonov:1967zza,Aharonov:1984,angeloPhysicsQuantumReference2011a,Giacomini:2017zju,delaHamette:2021oex,Hoehn:2023ehz,Hoehn:2019fsy,Hoehn:2020epv,AliAhmad:2021adn,Vanrietvelde:2018pgb,Giacomini:2021gei,Castro-Ruiz:2019nnl,delaHamette:2020dyi,Kabel:2024lzr,Kabel:2023jve,Carette:2023wpz,Loveridge:2019phw,loveridgeSymmetryReferenceFrames2018a,Bartlett:2006tzx,Castro-Ruiz:2021vnq,Vanrietvelde:2018dit,Hoehn:2021flk,Suleymanov:2023wio}, but are perhaps most relevant in the study of quantum gravity, where their ultimate purpose is to reconcile diffeomorphism invariance and general covariance with local quantum theory, in the way described above. Various frameworks for understanding QRFs have been proposed, sharing motivations, but differing in some key details, and most naturally applicable in different sets of circumstances. Of these, the perspective-neutral formalism~\cite{delaHamette:2021oex,Hoehn:2023ehz,Hoehn:2019fsy,Hoehn:2020epv,AliAhmad:2021adn,Giacomini:2021gei,Castro-Ruiz:2019nnl,Vanrietvelde:2018dit,Vanrietvelde:2018pgb,Hoehn:2021flk,Suleymanov:2023wio,DeVuyst:2024pop,DEHKlong} seems to be the one most generally suited for quantum gravity, since (unlike some other formalisms) it involves the direct construction of physical states solving the constraints.\footnote{See~\cite{DeVuyst:2024pop,DEHKlong} for more discussion on this.} It is the one we will employ in this paper.

QRFs were the essential ingredient which recently allowed~\cite{leutheusser2023causalconnectabilityquantumsystems,leutheusser2023emergenttimesholographicduality,Witten_2022,Chandrasekaran_2023,Chandrasekaran_2023b,kudlerflam2024generalizedblackholeentropy,Jensen_2023,Faulkner:2024gst,Ali_Ahmad_2024,Klinger:2023tgi} to identify the generalised entropy of a general subregion with the von Neumann entropy of a certain associated Type II algebra~\cite{DeVuyst:2024pop,DEHKlong,Fewster:2024pur,AliAhmad:2024wja}. Indeed, that identification was based on explicitly introducing a QRF for boosts, with its own Hilbert space (building on earlier ideas in~\cite{Donnelly_2016}). Upon imposing gauge invariance, it was then found that the relevant algebra of observables is given by the so-called crossed product of a (kinematical) Type III QFT algebra for matter/gravitons by the action of the boost. The resulting algebra is Type II because the boost corresponds to modular flow~\cite{Takesaki,Connes,Connes:1994yd}. In a limit in which the boost time can be treated as an approximately classical coordinate (i.e.\ its quantum fluctuations are suppressed relative to everything else -- this is a semiclassical regime), it was found that the von Neumann entropy $S_{\text{vN}}$ for this Type II algebra agrees with $S_{\text{gen}}$.

An added bonus of the construction in~\cite{Faulkner:2024gst,leutheusser2023causalconnectabilityquantumsystems,leutheusser2023emergenttimesholographicduality,Witten_2022,Chandrasekaran_2023,Chandrasekaran_2023b,kudlerflam2024generalizedblackholeentropy,Jensen_2023,Faulkner:2024gst,Ali_Ahmad_2024,Klinger:2023tgi} is that it allows one to define gravitational entropies without imposing a UV cutoff on the fields. Until recently, such a cutoff was required by most proofs of the GSL, because the entanglement entropy $S_{\text{out}}$ is a UV-divergent quantity (due to the Type III nature of the QFT algebra). But a UV cutoff is highly undesirable in a full theory of gravity, being at odds with general covariance and the holographic principle~\cite{Weinberg:1980gg,Bousso_2002}. Moreover, $S_{\text{gen}}$ itself can be made UV-finite by an appropriate renormalisation of Newton's constant $G_{\text{N}}$~\cite{Susskind_1994,jacobson1994blackholeentropyinduced,FROLOV1997339,Frolov_1998}, which strongly suggests that the cutoff is ultimately unnecessary. Indeed, by making the algebra Type II, the inclusion of the boost QRF intrinsically regulates the entropy, and hence eliminates the need for a UV cutoff. Leveraging this, in~\cite{Faulkner:2024gst} a perturbative, semiclassical proof of the GSL without a UV cutoff was given, combining the techniques of~\cite{Wall:2011hj} (which was the most powerful of the previous proofs, but still was semiclassical and UV-regulated) with those of~\cite{leutheusser2023causalconnectabilityquantumsystems,leutheusser2023emergenttimesholographicduality,Witten_2022,Chandrasekaran_2023,Chandrasekaran_2023b,kudlerflam2024generalizedblackholeentropy,Jensen_2023,Ali_Ahmad_2024,Klinger:2023tgi}.\footnote{See also similar methods in~\cite{Ali:2024jkx}. We should additionally point out~\cite[Section 4]{Chandrasekaran_2023}, which discussed a slightly different proof of the GSL applying when a very large amount of time has elapsed between the cuts whose generalised entropies are compared, using the fact that (in a semiclassical regime) the associated algebras of operators are approximately uncorrelated. By contrast, our proof does not require a large time separation of the cuts, and works non-semiclassically.}

In~\cite{Faulkner:2024gst}, the horizon boost QRF was explicitly identified as a conjugate variable to the asymptotic future area $A_\infty$ of the horizon. This is consistent with a broader story involving null surfaces such as horizons, where spin 2 gravitons are just one part of the gravitational data, and various other, oft-neglected, gravitational degrees of freedom turn out to be natural candidates for reference frames. Indeed, there are also spin 0 and spin 1 gravitational degrees of freedom~\cite{Ciambelli:2023mir}; the reason they are sometimes ignored is that they only appear at second order and higher in perturbation theory, whereas the gravitons appear at first order. They include, for example, perturbations to the area of the horizon. But the gravitational gauge constraints mix the various orders of perturbation theory, so a proper treatment must take certain parts of the higher order fields into account~\cite{DEHKinstability}. Through the constraints, the higher order degrees of freedom generate certain symmetries of the gravitons and matter fields, and one should also account for the Goldstone modes of these symmetries. For example, the horizon area generates a boost of the gravitons/matter~\cite{Wald_1993,Wall:2011hj,Ciambelli:2023mir,Faulkner:2024gst}, and the corresponding Goldstone mode is a degree of freedom measuring a time along the boost. The Goldstone modes of the higher order fields can more generally be understood as providing sets of reference frames relative to which the gravitons and matter may be measured.

With this in mind, an alternative way to describe the semiclassical regime for the GSL invoked in~\cite{Faulkner:2024gst} is that it applies when the fluctuations of the fields outside the horizon are suppressed relative to fluctuations of $A_\infty$ (in Planck units), since $A_\infty$ is conjugate to the boost time.
\changed{More precisely, the semiclassical regime involved in the proof of the GSL in~\cite{Faulkner:2024gst} required}
\begin{equation}
  \changed{\frac{\Delta A^\infty}{4G_{\text{N}}} \gg \Delta H_{\text{boost}},}
\end{equation}
\changed{where $H_{\text{boost}}$ is the generator of a boost of the fields outside the horizon.}
But whether this \changed{inequality} holds can change over the course of the horizon's evolution. \changed{Indeed, $A^\infty$ has no time dependence, but $H_{\text{boost}}$ does, since it generates a boost around a specific cut of the horizon.}

\begin{figure}
  \centering
  \begin{tikzpicture}[scale=1.5]
    \fill[black!10] (-1,-1.75) -- (-1,-1) -- (2,2) -- (3,2) -- (3,-1.75) -- (2,-1.75);

    \draw[decorate, decoration=snake,red] (0.4,-1) -- (2.1,-0.8);
    \draw[decorate, decoration=snake,red] (0.4,-1.2) .. controls (1.25,-1) .. (2.1,-1);
    \draw[decorate, decoration=snake,red] (0.4,-0.8) .. controls (1.25,-0.8) .. (2.1,-0.6);

    \draw[green!70!black, very thick,-stealth] (0.4,-1) -- (-0.1,-0.1) -- (-0.6,0.8);
    \draw[green!70!black, very thick,-stealth] (2.1,-0.8) -- (1.3,1.3) -- (0.9,2.35);

    \draw[green!50!black,fill=green!20!white,thick] (0.4,-1) ellipse (0.4 and 0.3) node {$m$};
    \draw[green!50!black,fill=green!20!white,thick] (2.1,-0.8) ellipse (0.4 and 0.3) node {$m'$};

    \draw[thick] (-1.05,-1.05) -- (2.05,2.05);

    \fill[decorate, decoration={snake, amplitude=0.3pt}, inner color=red!20, outer color=red!40] (-0.1,-0.1) .. controls (0.3,0.7) and (0.5,0.9) .. (1.3,1.3) .. controls (0.9,0.5) and (0.7,0.3) .. (-0.1,-0.1);

    \draw[thick,decorate, decoration={snake,amplitude=0.5pt},black!70] (-0.1,-0.1) -- (1.3,1.3);

    \node[above left] at (-1,-1.1) {\large$\mathscr{H}$};

    \draw[Latex-,gray] (-0.6,-0.4) .. controls (-1,0) and (-2,0.5) .. (-2,1) node[black,above] {\small semiclassical};
    \draw[Latex-,gray] (0.4,0.85) .. controls (0,1.25) and (-0.5,1.5) .. (-0.5,2) node[black,above] {\small not semiclassical};
    \draw[Latex-,gray] (1.6,1.75) .. controls (1.35,2.2) and (1.8,2) .. (1.8,2.3) node[black,above] {\small semiclassical};
  \end{tikzpicture}
  \caption{\changed{The semiclassical regime applies when $\Delta A^\infty/4G_{\text{N}}\gg \Delta H_{\text{boost}}$, where $A^\infty$ is the area of the horizon in the asymptotic future, and $H_{\text{boost}}$ is the generator of a boost acting on the fields outside the horizon~\cite{Faulkner:2024gst}.} During the course of the horizon's evolution, the quantum state can move in and out of \changed[the]{this} semiclassical regime. An example is shown, with two clumps of matter $m$ and $m'$ which are highly entangled with each other \changed{across modes of $H_{\text{boost}}$}, and initially located outside the horizon. Suppose the horizon starts off in a semiclassical state (so \changed[the fields do]{$H_{\text{boost}}$ does} not fluctuate too much), then $m$ falls in, and some time later $m'$ falls in. After $m$ crosses the horizon, the entanglement causes \changed[the fields in the exterior region]{$H_{\text{boost}}$} to be much more highly fluctuating, \changed{while $\Delta A^\infty$ is unchanged,} causing us to leave the semiclassical regime. Then, once $m'$ falls in, the fluctuations of \changed[the fields]{$H_{\text{boost}}$} will be suppressed once more, so we return to the semiclassical regime.}
  \label{Figure: leaving and entering semiclassical regime}
\end{figure}

For instance, suppose we start off in the semiclassical regime, and then some matter $m$ falls into the horizon which is very highly entangled with some other matter $m'$ remaining outside the horizon{ but otherwise in a pure state; in particular we can consider a joint state of $m$ and $m'$ of the form}
\begin{equation}
  \changed{\ket{\psi}_{mm'} = \sum_\lambda\ket{\lambda}_m\otimes\ket{-\lambda}_{m'},}
\end{equation}
\changed{where each $\ket{\lambda}$ is a boost eigenstate, and the sum is done over a very large range of boost eigenvalues $\lambda$}.
This \changed{entanglement} would cause the fluctuations of \changed[the fields outside the horizon]{$H_{\text{boost}}$} to increase \changed{after $m$ has fallen in}; in an extreme case, \changed[they]{it} would now fluctuate more than $A_\infty\changed{/4G_{\text{N}}}$, and thus we would have left the semiclassical regime. Subsequently, the other matter could fall in, \changed{purifying the entanglement,} reducing the \changed[field]{$H_{\text{boost}}$} fluctuations back to their former value, and so returning us to the semiclassical regime. This is depicted in Figure~\ref{Figure: leaving and entering semiclassical regime}. The traditional GSL cannot be applied \changed{in a gauge-invariant way} to processes like these, because it involves the entropy $S_{\text{out}}$ of local fields, and so as discussed above only makes sense semiclassically.\footnote{\changed{To be clear, it is not necessarily the case that this process is \emph{inconsistent} with the traditional GSL. Indeed, if one assumes that the matter $m$ and $m'$ obeys the quantum null energy condition (QNEC) or quantum focussing conjecture (QFC)~\cite{Bousso:2015mna,Bousso:2015wca}, then this process \emph{does} satisfy the traditional GSL. The issue we are raising here is that a rigorously gauge-invariant derivation of the GSL has so far been lacking for this kind of process, precisely due to the violation of $\Delta A^\infty /4G_{\text{N}}\gg\Delta H_{\text{boost}}$. The QNEC and QFC are (quasi)local conditions, and hence (beyond a semiclassical regime) are subject to the same gauge ambiguities as local observables in gravity, as noted at the beginning of the paper. It would be interesting in future work to investigate gauge-invariant versions of the QNEC and QFC.}}\textsuperscript{,}\footnote{Another process worth mentioning is that of black hole evaporation. Eventually, the horizon area will go to zero, and so cannot fluctuate at all, suggesting that the semiclassical approximation will be rendered invalid. But this is naturally a somewhat more speculative scenario than the example depicted in Figure~\ref{Figure: leaving and entering semiclassical regime}, since it will be significantly influenced by non-perturbative effects, and thus goes beyond the scope of this paper.}

The point of this paper is to generalise the GSL so that it applies also to such processes. As a prerequisite there needs to be an extension of the definition of generalised entropy beyond the semiclassical regime, and there is some amount of ambiguity in how this should be done. But since the underlying von Neumann algebra is unambiguously defined outside the semiclassical regime, let us propose the following well-motivated extension (see also~\cite{kudlerflam2024generalizedblackholeentropy}): we just set $S_{\text{gen}}=S_{\text{vN}}$. This is an attractive definition because it allows us to maintain a statistical interpretation for $S_{\text{gen}}$. The key question is then whether the GSL applies to this definition -- in other words, \emph{is $S_{\text{vN}}$ non-decreasing in time, to all orders beyond the semiclassical regime, even during non-semiclassical processes like the one described in the previous paragraph?} As already mentioned, the answer turns out to be: \emph{yes, up to the free energy of a certain reference frame}. In an appropriate semiclassical regime, this free energy is positive, so one recovers the semiclassical GSL. But in more general states, the reference frame free energy can be negative, in which case $S_{\text{vN}}$ is allowed to decrease. We will shortly give some more detail on why this is nevertheless the appropriate form of the GSL beyond the semiclassical regime, in Section~\ref{Subsection: GSL with QRFs}.

Our proof will build on the techniques of~\cite{leutheusser2023causalconnectabilityquantumsystems,leutheusser2023emergenttimesholographicduality,Witten_2022,Chandrasekaran_2023,Chandrasekaran_2023b,kudlerflam2024generalizedblackholeentropy,Jensen_2023,Wall:2011hj,Faulkner:2024gst,Ali_Ahmad_2024,Klinger:2023tgi}, with a key extra ingredient, as follows. Those works only accounted for a single gauge symmetry (the boost), which was all that was required to regulate the von Neumann entropy. Here, we will also account for another gauge symmetry: null translations along the horizon. This, in a nutshell, is what allows our proof to work.

On a null surface, boosts and null translations may be performed along each null ray independently, but here, for simplicity, we will consider only transformations which act on all the null rays simultaneously. The principle of ultralocality~\cite{Wall:2011hj,Ciambelli:2024swv} suggests that it should be possible to extend our analysis beyond this restriction. We are also ignoring (again, for simplicity) various other gauge symmetries, such as diffeomorphisms which act tangentially to the cuts of the horizon. Note that the gauge symmetries which we \emph{do} account for are arguably more important than the ones we do not. As already mentioned, boosts are special because they correspond to modular flow. But null translations are also special, because they correspond to `half-sided modular translations'~\cite{Borchers:1991xk,Borchers:2000pv}. Thus, the gauge symmetries we account for are deeply tied to the modular structure of the theory.

Since we are gauging null translations, it is no longer physically meaningful to consider \emph{fixed} regions exterior to the horizon. This is because a null translation will move degrees of freedom in and out of these regions; since this is a gauge symmetry, it would not be clear which degrees of freedom ought to then be considered as part of the physical subsystem associated with the region. To solve this problem, one should instead define regions in a relational manner~\cite{DeWitt:1962cg,PhysRevLett.4.432,RevModPhys.33.510,PhysRev.124.274,Rovelli_1991}. This means that the location of the region should be specified covariantly in terms of dynamical degrees of freedom, so that the region itself is also moved around by gauge symmetries, and any interior degrees of freedom remain inside. In other words, we need a QRF specifying the location of the region. There are many ways one could construct such a QRF~\cite{Goeller:2022rsx,Carrozza:2022xut}; for example, one could imagine shooting in a particle such that the location at which it crosses the horizon specifies a cut $\mathcal{S}$ (see Figure~\ref{Figure: dynamical cuts from infalling particles}).\footnote{Beyond two spacetime dimensions, this should not be taken too seriously, since an infalling particle only really gives the location of a single point on the horizon. We are implicitly assuming that there is some additionally specified way of going from this single point to a full cut. Alternatively, the cut might be determined by a congruence of particles, or a brane, falling in from all directions. In any case, these are only particular examples of ways to construct dynamical cuts.} The exterior of the cut $\mathcal{S}$ is then a \emph{relationally} and \emph{covariantly} defined region, since null translations move the particle, and thus move $\mathcal{S}$. Cuts $\mathcal{S}$ of the horizon with this covariance property are an essential part of our proof. We will call them `dynamical cuts', and it is their free energy which contributes to the non-semiclassical GSL.

\changed{It seems like this covariance is all that `dynamical cuts' should be required to satisfy. From Section~\ref{Section: dynamical cuts} of the paper onwards, we will use a concrete but simplified construction, in which the location of each dynamical cut is represented by the position operator acting on an auxiliary $L^2(\RR)$ Hilbert space. This is similar to what was done in~\cite{Chandrasekaran_2023b}, where observers were introduced into de Sitter spacetime, each carrying clocks with auxiliary Hilbert spaces $L^2(\RR)$, whose position operators measure boost time. It is important to note that the setup of~\cite{Chandrasekaran_2023b}, as well as the one described in this paper, are only \emph{toy models}. The physical significance of these models is that they capture the essential properties required of clocks and dynamical cuts, and QRFs more generally. More realistic physical models would involve QRFs that are not auxiliary to the system under consideration, but rather constructed out of pre-existing degrees of freedom, such as the metric and matter fields. This has already been done to some extent for the observers of~\cite{Chandrasekaran_2023b}; for example instead of an auxiliary clock,~\cite{Chen:2024rpx} showed that one may use a slow-rolling inflaton field, and in the case of a horizon,~\cite{Faulkner:2024gst} demonstrated that one may use the gravitational Goldstone mode of the boost.}

\changed{We anticipate that similarly more realistic constructions may be applied to the case of dynamical cuts. As already noted, particles crossing the horizon provide an example, in which case the $L^2(\RR)$ Hilbert space of the dynamical cut can literally be interpreted as a space of position wavefunctions for the particle. But dynamical cuts could also be constructed in other ways. For example, one natural procedure is to use the metric to construct an affine time coordinate along the null horizon, and then specify a dynamical cut as being located at a particular affine time. In the quantum theory, the metric should be quantised, meaning affine time would be quantised, and thus so would be this dynamical cut -- but a full analysis of its properties would require a deeper understanding of the quantisation of the metric. It will be interesting to do this in future work, but for most of the paper we will just restrict to the simple $L^2(\RR)$ model.}

\subsection{GSL with QRFs}
\label{Subsection: GSL with QRFs}

The thermodynamical properties of a system depend non-trivially on the choice of reference frame used to make observations~\cite{Hoehn:2023ehz}, so one might not be surprised that the standard GSL requires modification before it can be applied it to the regions exterior to dynamical cuts. To get a better idea of what kind of modification needs to be done, consider the following rather heuristic derivation of the GSL, initially for fixed cuts.

Let us assume that the full system of the black hole and its exterior evolves unitarily, but that, from the point of view of the exterior region, the horizon appears as a thermal bath at the Hawking temperature $T$, and that there is a steady state of thermal equilibrium between the horizon and the exterior. Then one may show that the \emph{free energy}
\begin{equation}
  F_{\text{out}} = \expval{H_{\text{out}}} - TS_{\text{out}}
\end{equation}
of the exterior region is a non-increasing quantity~\cite{Sorkin:1986mg,Sorkin:1997ja,TenProofs,Landi:2020bsq}.\footnote{As explained in~\cite[Section 4.2]{TenProofs}, this is a special case of monotonicity of relative entropy.} Here, $H_{\text{out}}$ is the Hamiltonian with respect to which the state is thermal. In particular, for Killing horizons, $H_{\text{out}}$ is the generator of evolution along the Killing vector field generating the horizon, and one may show using the gravitational constraints that $T^{-1}H_{\text{out}}=-\frac{A}{4G_{\text{N}}} + H_\infty$, where $H_\infty$ is a contribution at spacelike infinity. Since we are assuming the full system evolves unitarily, we have that $H_\infty$ is conserved (because it is the Hamiltonian for the full system), so $T^{-1}\dv{t}\expval{H_{\text{out}}} = -\dv{t} \frac{\expval{A}}{4G_{\text{N}}}$ (this may also be argued using the first law of black hole mechanics~\cite{Sorkin:1997ja}). Combining this with $\dv{t}F_{\text{out}}\le0$, one has
\begin{equation}
  \dv{t}\qty(\frac{\expval{A}}{4G_{\text{N}}}+S_{\text{out}})\ge 0,
\end{equation}
which is the standard GSL.

Consider now instead two dynamical cuts $\mathcal{S}_a$ and $\mathcal{S}_b$ defined by infalling particles $a$ and $b$, such that $\mathcal{S}_a$ is to the future of $\mathcal{S}_b$, as depicted in Figure~\ref{Figure: dynamical cuts from infalling particles}. Let $\mathcal{N}_a,\,\mathcal{N}_b$ denote the regions exterior to $\mathcal{S}_a,\,\mathcal{S}_b$ respectively. We can again invoke that free energy is non-increasing. We shall consider the free energy of all the degrees of freedom in $\mathcal{N}_b$. By the time we evolve to $\mathcal{N}_a$, $a$ has left the exterior, crossing over the horizon into the black hole. But the free energy is non-increasing only when we evaluate it on a fixed set of degrees of freedom. Thus, we must include $a$ explicitly on the right-hand side of
\begin{equation}
  F(\mathcal{N}_b)\ge F(\mathcal{N}_a\cup a),
\end{equation}
where $F$ denotes the free energy of a given set of degrees of freedom. Next, we use that (this follows from subadditivity of entropy, and the assumption that interactions between the particle $a$ and the rest of the degrees of freedom are negligible):
\begin{equation}
  F(\mathcal{N}_a\cup a) \ge F(\mathcal{N}_a) + F(a)
  \label{Equation: free energy superadditivity}
\end{equation}
to obtain
\begin{equation}
  F(\mathcal{N}_b)-F(\mathcal{N}_a)\ge F(a).
\end{equation}
Now, using the same argument as for fixed cuts, the left-hand side of this is given by $T$ times the difference in generalised entropies of the two regions, so we have
\begin{equation}
  S_{\text{gen}}(\mathcal{N}_a)-S_{\text{gen}}(\mathcal{N}_b) \ge T^{-1}F(a) = T^{-1}\expval{H_a} - S(a),
  \label{Equation: modified GSL}
\end{equation}
where $H_a$ is the thermal Hamiltonian of $a$, and $S(a)$ is the entropy of $a$. This is a modified GSL, stating (roughly speaking) that the increase in the generalised entropy between two dynamical cuts is bounded below by the free energy (divided by $T$) of the degrees of freedom defining the later cut.

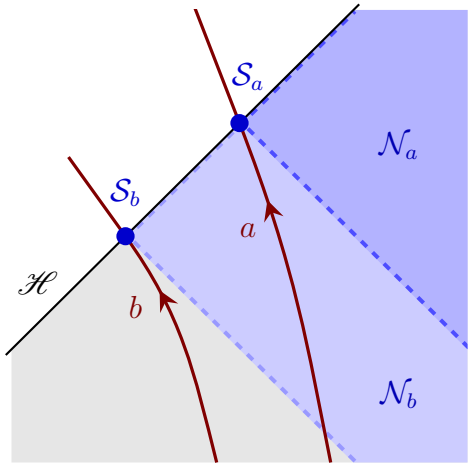
\begin{figure}
  \centering
  \begin{tikzpicture}[scale=1.5]
    \fill[black!10] (-1,-2) -- (-1,-1) -- (2,2) -- (3,2) -- (3,-2) -- (2,-2);

    \fill[blue!20] (0,0) -- (2,2) -- (3,2) -- (3,-2) -- (2,-2) -- (0,0);
    \begin{scope}
      \clip (0,0) -- (2,2) -- (3,2) -- (3,-2) -- (2,-2) -- (0,0);
      \draw[line width=3pt,blue!40,dashed] (2,-2) -- (0,0) -- (2,2);
    \end{scope}
    \fill[blue!30] (1,1) -- (2,2) -- (3,2) -- (3,-1) -- (1,1);
    \begin{scope}
      \clip (1,1) -- (2,2) -- (3,2) -- (3,-1) -- (1,1);
      \draw[line width=3pt,blue!70,dashed] (3,-1) -- (1,1) -- (2,2);
    \end{scope}
    \draw[thick] (-1.05,-1.05) -- (2.05,2.05);

    \tikzstyle particle=[red!50!black,very thick,postaction={decorate,decoration={markings,mark=at position .53 with {\arrow[scale=1.5]{stealth}}}}]
    \begin{scope}
      \clip (-1.5,-2) rectangle (3,2);
      \draw[particle] (0.8,-2) .. controls (0.6,-1.2) and (0.5,-0.7) .. (0,0) node[midway,left,shift={(-0.4,0.5)}] {\large$b$} -- (-0.5,0.7);
      \draw[particle, shift={(1,1)}] (0.8,-3) .. controls (0.6,-2) and (0.5,-1.3) .. (0,0) node[midway,left,shift={(-0.4,1)}] {\large$a$} -- (-0.5,1.3);
    \end{scope}

    \node[blue!70!black] at (2.4,-1.4) {\large$\mathcal{N}_b$};
    \node[blue!70!black] at (2.4,0.8) {\large$\mathcal{N}_a$};

    \fill[blue!80!black] (0,0) circle (0.08) node[above,shift={(0em,0.7em)}] {\large$\mathcal{S}_b$};
    \fill[blue!80!black] (1,1) circle (0.08) node[above,shift={(0.3em,0.8em)}] {\large$\mathcal{S}_a$};
    \node[above left] at (-0.5,-0.6) {\large$\mathscr{H}$};
  \end{tikzpicture}
  \caption{
    One natural example of a ``dynamical cut'' of a horizon $\mathscr{H}$ is given by the location at which it is crossed by an infalling particle. Here, the worldlines of two such particles $a$ and $b$ are shown, which intersect the horizon at the cuts $\mathcal{S}_a$ and $\mathcal{S}_b$ respectively. The associated exterior regions obey $\mathcal{N}_a\subset\mathcal{N}_b$, because $a$ falls through the horizon later than $b$.
  }
  \label{Figure: dynamical cuts from infalling particles}
\end{figure}

The argument just given makes many simplifying assumptions\footnote{A good account of many of them may be found in~\cite{TenProofs}.} and is too na\"ive to be taken with anything less than a hefty grain of salt (otherwise this paper would be somewhat shorter). Nevertheless, using much more rigorous methods,~\eqref{Equation: modified GSL} is the form of GSL that we will prove.

How should we interpret this modification of the GSL in the context of gravity and QRFs? As we will show later in the paper, one may write
\begin{equation}
  T^{-1}\expval{H_a} = \frac1{4G_{\text{N}}}\Delta_a A,
\end{equation}
where $\Delta_a A$ is the contribution to the change in area from $\mathcal{S}_b$ to $\mathcal{S}_a$ due to $a$ crossing the horizon (this is not the full change in area, because it excludes contributions from the \emph{fields} crossing the horizon). Moreover, if $a$ obeys the null energy condition, then $\Delta_a A$ is positive, from which~\eqref{Equation: modified GSL} implies
\begin{equation}
  S_{\text{gen}}(\mathcal{N}_a)\ge S_{\text{gen}}(\mathcal{N}_b) - S(a).
  \label{Equation: modified GSL NEC}
\end{equation}
The right-hand side of this inequality may be rewritten
\begin{equation}
  S_{\text{gen}}(\mathcal{N}_b) - S(a) = \frac{A(\mathcal{S}_b)}{4G_{\text{N}}} + S(\mathcal{N}_b|a),
\end{equation}
where $S(\mathcal{N}_b|a) = S(\mathcal{N}_b) - S(a)$ is `conditional entropy', which quantifies the amount of information needed to communicate the state of $\mathcal{N}_b$ to someone who already has access to the state of $a$~\cite{Horodecki:2005ehk}.\footnote{The importance of conditional entropies in extending the GSL beyond the semiclassical regime was also recently suggested in~\cite{ahmad2024emergentgeometryquantumprobability}.} In light of the fact that the state of $a$ determines the location of the later region $\mathcal{N}_a$, it makes some amount of intuitive sense that this conditional entropy would play a role in the GSL. Indeed, it is commensurate with the vague notion that one `needs to know the state of the later cut in order to compare its exterior region with that of the earlier cut'. It also perhaps motivates defining what one might call the `generalised conditional entropy'
\begin{equation}
  S_{\text{gen}}(\mathcal{N}_b|a) = \frac{1}{4G_{\text{N}}}\qty\Big(A(\mathcal{S}_b)+\Delta_a A) + S(\mathcal{N}_b|a),
\end{equation}
so that~\eqref{Equation: modified GSL} (which is strictly stronger than~\eqref{Equation: modified GSL NEC} under the present conditions, but also applies under general conditions) may be written
\begin{equation}
  S_{\text{gen}}(\mathcal{N}_a)\ge S_{\text{gen}}(\mathcal{N}_b|a).
\end{equation}
This holds without requiring $a$ to obey the null energy condition.

If the particle $a$ obeys the null energy condition \emph{and} is unentangled from other degrees of freedom (so $S(a)=0$), then~\eqref{Equation: modified GSL NEC} implies the standard GSL
\begin{equation}
  S_{\text{gen}}(\mathcal{N}_a) \ge S_{\text{gen}}(\mathcal{N}_b).
  \label{Equation: standard GSL}
\end{equation}
These conditions are appropriate for a semiclassical limit in which the location of the dynamical cut $\mathcal{S}_a$ may be treated as a classical variable. On the other hand, if neither of these conditions are obeyed, then~\eqref{Equation: modified GSL} still implies the standard GSL~\eqref{Equation: standard GSL} -- so long as the particle $a$ obeys a form of the \emph{quantum} null energy condition~\cite{Bousso:2015mna,Bousso:2015wca}. But it should be noted that in this case~\eqref{Equation: modified GSL} is still a strictly stronger inequality than the standard GSL.

One more comment on our modified GSL~\eqref{Equation: modified GSL} is in order. Suppose we rewrite it as follows:
\begin{equation}
  \frac{1}{4G_{\text{N}}}\qty(A(\mathcal{S}_b)+\Delta_a A-A(\mathcal{S}_a)) + S(\mathcal{N}_b) \le S(\mathcal{N}_a) + S(a).
  \label{Equation: modified GSL subadditivity}
\end{equation}
By the subadditivity of entropy, the right-hand side of~\eqref{Equation: modified GSL subadditivity} is also bounded below by $S(\mathcal{N}_a\cup a)$. In a sense, the modified GSL that we prove is a mixture of the ordinary GSL with this subadditivity between $\mathcal{N}_a$ and $a$. One might wonder whether it is possible to replace the right-hand side of~\eqref{Equation: modified GSL subadditivity} with $S(\mathcal{N}_a\cup a)$ to eliminate this subadditivity and get a stronger inequality, but it seems unlikely that this will work in any meaningful way. This is because the combined system $\mathcal{N}_a\cup a$ is essentially equivalent to $\mathcal{N}_b$; indeed, operators acting on $a$ can move it to $\mathcal{S}_b$, and then we have $\mathcal{N}_a=\mathcal{N}_b$. Thus, all observables in $\mathcal{N}_b$ are also observables in $\mathcal{N}_a\cup a$, so comparing the von Neumann entropy of $\mathcal{N}_b$ with the von Neumann entropy of $\mathcal{N}_a\cup a$ will not provide us with any useful insights.\footnote{This is one reason to be sceptical of the heuristic thermodynamical argument above, where we used subadditivity~\eqref{Equation: free energy superadditivity}.} For this reason, the splitting of $\mathcal{N}_a\cup a$ into $\mathcal{N}_a$ and $a$ is an essential part of the story.

We have so far mostly described the dynamical cuts $\mathcal{S}_a,\mathcal{S}_b$ as being located at the intersections of particle worldlines with the horizon. But this is far from the only way one could construct a dynamical cut (indeed, see~\cite{Goeller:2022rsx} for an extensive discussion on the large array of available dynamical frames in gravity). Another possibility could be to specify a cut as the location at which a certain field takes on a certain value. Yet another possibility could be to specify a cut in such a way that its exterior region agrees with the timelike envelope of an accelerating observer who remains just outside the horizon. The key point is that the cuts one uses should be quantum degrees of freedom which transform covariantly under gauge symmetries. We will implement this in the simplest possible way, which to some degree remains agnostic of how exactly a given dynamical cut is constructed.

More precisely, we will model each cut as having a Hilbert space $L^2(\RR)$. The position operator on this Hilbert space gives the location of the cut on the horizon, and is taken to transform appropriately under boosts and null translations. We will then formulate a von Neumann algebra $\mathcal{A}_{\mathcal{N}_a}^G$ of gauge-invariant operators acting in the region $\mathcal{N}_a$ exterior to each cut, and a von Neumann algebra $\mathcal{A}^G_{\mathcal{S}_a\in\mathcal{N}_b}$ of operators acting on a cut $\mathcal{S}_a$ located inside the exterior region $\mathcal{N}_b$ of another cut $\mathcal{S}_b$. Then, we will prove
\begin{equation}
  S_{\text{vN}}(\mathcal{A}_{\mathcal{N}_a}^G)\ge S_{\text{vN}}(\mathcal{A}_{\mathcal{N}_b}^G) + F(\mathcal{S}_a\in\mathcal{N}_b),
  \label{Equation: exact GSL}
\end{equation}
where
\begin{equation}
  F(\mathcal{S}_a\in\mathcal{N}_b) = \expval*{H_{\mathcal{S}_a|\mathcal{S}_b}} - S_{\text{vN}}(\mathcal{A}_{\mathcal{S}_a\in\mathcal{N}_b}^G),
\end{equation}
with $H_{\mathcal{S}_a|\mathcal{S}_b}$ the generator of a boost of $\mathcal{S}_a$ around $\mathcal{S}_b$ (here and in the rest of the paper we work in units such that $T=1$). The notation $S_{\text{vN}}(\mathcal{A})$ denotes the von Neumann entropy of a given state with respect to the algebra $\mathcal{A}$.

The inequality~\eqref{Equation: exact GSL} holds exactly, without any need for a UV cutoff or semiclassical limit. But if one does impose these, we will show that $S_{\text{vN}}(\mathcal{A}_{\mathcal{N}_a}^G)$ reduces to the usual semiclassical generalised entropy of $\mathcal{N}_a$, and~\eqref{Equation: exact GSL} reduces to~\eqref{Equation: modified GSL}. As explained above, we take $S_{\text{vN}}(\mathcal{A}_{\mathcal{N}_a}^G)$ to be the \emph{definition} of $S_{\text{gen}}(\mathcal{N}_a)$ beyond the semiclassical limit. Therefore,~\eqref{Equation: exact GSL} is a genuine exact extension of the GSL beyond the semiclassical/UV cutoff regimes.

Although the nature of \emph{non}-perturbative quantum gravity remains mysterious, our result gives evidence that a relational picture involving quantum reference frames might suffice for a full understanding of the statistical and thermodynamical properties of perturbative gravitational degrees of freedom.

The rest of the paper gives the details of our proof of~\eqref{Equation: exact GSL}, and proceeds as follows. In Section~\ref{Section: fixed cuts} we review relevant details of the derivation of the semiclassical GSL given in~\cite{Faulkner:2024gst} in terms of the entropies of algebras outside fixed cuts, and explain why the inability to extend that derivation beyond the semiclassical regime stems from not accounting for null translation gauge invariance. To remedy this, in Section~\ref{Section: dynamical cuts} we construct the von Neumann algebras of the regions exterior to \emph{dynamical cuts}, show that these algebras have well-defined density operators and von Neumann entropies, and describe their representation on a physical Hilbert space. Then, in Section~\ref{Section: GSL}, we prove that these algebras obey the GSL~\eqref{Equation: exact GSL}, which follows from conditioning on an ordering of the cuts, and an application of the monotonicity of relative entropy. In Section~\ref{Section: semiclassical}, we demonstrate how~\eqref{Equation: exact GSL} reduces to the ordinary semiclassical GSL, in an appropriate semiclassical limit. We end in the Conclusion with some discussion on the limitations of our approach, and brief speculation on how it may extend to the non-perturbative regime.

\section{Algebras outside fixed cuts}
\label{Section: fixed cuts}

We shall begin by recapping what happens when we do not use dynamical cuts, and in particular why failing to properly account for null translation invariance prevents one from extending the GSL beyond the semiclassical limit.

Let $v$ be a future-directed null coordinate along the horizon $\mathscr{H}$. For concreteness, we shall assume that $v$ is an affine parameter along each null ray. Other options for $v$ include the `dressing time' of~\cite{Ciambelli:2023mir}. Consider a fixed cut $\mathcal{S}(v)$ at null time $v$, let $\mathcal{N}(v)$ denote the spacetime region exterior to $\mathcal{S}(v)$, and let $\mathcal{A}_{\text{QFT}}(v)$ denote the von Neumann algebra of field operators (including those acting on perturbative gravitons) with support in $\mathcal{N}(v)$; this algebra is Type III$_1$~\cite{Witten_2018,Sorce:2023fdx,Buchholz:1986bg,Yngvason_2005}, and for simplicity we take it to be a factor (meaning it has a trivial center).

We will in this paper consider spacetimes for which the horizon $\mathscr{H}$ is a Cauchy surface, so that $\mathcal{A}_{\text{QFT}}(v)$ is generated by field operators on the horizon with support to the future of $v$. This class of spacetimes includes, for example, black hole horizons with AdS asymptotics. But it does not include other important cases such as the asymptotically flat Schwarzschild black hole, where one must additionally include in $\mathcal{A}_{\text{QFT}}(v)$ field operators with support at future null infinity $\mathscr{I}^+$ and timelike infinity $i^+$. Although we expect that our results may be extended to such cases using the techniques described in~\cite[Section 4]{Faulkner:2024gst}, we choose not to explicitly do so here, since the way in which we use dynamical cuts should be more or less independent of the inclusion of field operators at $\mathscr{I}^+\cup i^+$.

The generator of null translations $v\to v+s$ of fields on the horizon is the integrated averaged null energy
\begin{equation}
  P = \int_{\mathscr{H}} \eta T_{vv},
\end{equation}
where $T_{vv}$ are the components along $\pdv{v}$ of the field stress-energy tensor, and $\eta=\iota_{\partial_v}\epsilon$ is a volume form on $\mathscr{H}$, with $\epsilon$ the spacetime volume form. The action generated by $P$ maps the algebras outside different cuts $\mathcal{S}(v)$ into each other:
\begin{equation}
  \mathcal{A}_{\text{QFT}}(v+s) = e^{isP}\mathcal{A}_{\text{QFT}}(v)e^{-isP}.
  \label{Equation: field algebra null translation}
\end{equation}
We assume two additional conditions hold (this assumption applies in a large class of QFTs). First, that $P\ge 0$, i.e.\ $P$ is a non-negative operator; this follows from the averaged null energy condition (ANEC)~\cite{Roman,Borde,Roman2,Faulkner_2016,Hartman_2017,Kravchuk_2018}. And second, that there is a `vacuum state' $\ket{\Omega}\in\mathcal{H}$ which satisfies $P\ket{\Omega}=0$, and which is cyclic and separating for $\mathcal{A}_{\text{QFT}}(v)$ (note if it is cyclic and separating for one $v$, it is cyclic and separating for all $v$, by~\eqref{Equation: field algebra null translation}). Many more details on the scope of these assumptions (and their motivations) may be found in~\cite{Wall:2011hj,Faulkner:2024gst}. It turns out these conditions are sufficient to apply the `half-sided modular translation' machinery of Borchers~\cite{Borchers:1991xk,Borchers:2000pv} and Wiesbrock~\cite{Wiesbrock:1992mg} to find that the modular Hamiltonian of $\ket{\Omega}$ for the algebra $\mathcal{A}_{\text{QFT}}(v)$ is given by
\begin{equation}
  H_v = H - 2\pi v P,\qq{where}
  H = 2\pi \int_{\mathscr{H}} \eta v T_{vv}.
\end{equation}
This operator generates an $\mathcal{S}(v)$-preserving boost. As discussed in~\cite{Faulkner:2024gst}, the overall gauge constraint for this boost (assuming minimally coupled matter) may be written in the form
\begin{equation}
  \mathcal{C}_v = H_v + \hat{q} + \mathcal{Q}_v,
\end{equation}
where $\hat{q}=-\frac{A^{(2)}_\infty}{4G_{\text{N}}}$ is the second order fluctuation of the area of the horizon in the asymptotic future $v\to\infty$, and $\mathcal{Q}_v$ is a contribution in the asymptotic past $v\to-\infty$ whose exact form we need not worry about (since it commutes with all operators of concern in $\mathcal{N}(v)$).

As in~\cite{Faulkner:2024gst}, we model $\hat{q}$ as the position operator on a Hilbert space $\mathcal{H}_C=L^2(\RR)$. The subscript ${}_C$ is supposed to reflect the fact that the conjugate momentum operator $\hat{p}$ now plays the role of a \emph{clock}, measuring time along the boost generated by $\hat{q}$. This is an \emph{ideal quantum reference frame}, meaning it can have states of arbitrarily precise times. One could also model this clock with a \emph{non-ideal} quantum reference frame by restricting the spectrum of $\hat{q}$~\cite{DeVuyst:2024pop,DEHKlong}.

Consider the algebra of operators in $\mathcal{N}(v)$, which before imposing gauge-invariance is the tensor product of $\mathcal{A}_{\text{QFT}}(v)$ and $\mathcal{B}(\mathcal{H}_C)$ (since the clock at $v\to \infty$ is inside $\mathcal{N}(v)$).\footnote{Here we choose not to include other degrees of freedom at $v\to\infty$. See the Conclusion for discussion on this point.} Imposing invariance under the boost gauge symmetry (i.e.\ restricting to operators commuting with $\mathcal{C}_v$, denoted via the superscript ${}^{\mathcal{C}_v}$ below), one finds the algebra
\begin{equation}
  (\mathcal{A}_{\text{QFT}}(v)\otimes\mathcal{B}(\mathcal{H}_C))^{\mathcal{C}_v} = \qty{e^{iH_v\hat{p}}ae^{-iH_v\hat{p}},\hat{q}\Bigm\vert a\in \mathcal{A}_{\text{QFT}}(v)}'' = \mathcal{A}_{\text{QFT}}(v)\rtimes_\alpha \RR.
\end{equation}
The double commutant $''$ produces the von Neumann algebra generated by the operators in the braces. In the language of quantum reference frames, $e^{iH_v\hat{p}}ae^{-iH_v\hat{p}}$ is a `dressed' or `relational' observable, acting with $a$ at the time specified by the clock, while $\hat{q}$ is the generator of a `reorientation' of the clock (i.e.\ a change in the time it reads). Von Neumann algebras of this form have been well-studied since the `70s~\cite{Takesaki,Connes,Connes:1994yd}; they are called `crossed product algebras', and the right-hand side above is the notation for this, with $\alpha$ being the group action of $\RR$ on $\mathcal{A}_{\text{QFT}}(v)$ by boosts:
\begin{equation}
  \alpha_t: a\mapsto  e^{i\mathcal{C}_v t}a e^{-i\mathcal{C}_v t} =  e^{iH_v t} a e^{-iH_v t}.
\end{equation}
More generally, one can construct the crossed product of any algebra by any group of its automorphisms. But the particular crossed product described here is special, because $\alpha$ corresponds to the \emph{modular flow} of a state on $\mathcal{A}_{\text{QFT}}(v)$ (in particular, the vacuum $\ket{\Omega}$). It is a theorem that the crossed product of a Type III algebra by modular flow is Type II~\cite{TheoryOfOperatorAlgebrasI}. Thus, $(\mathcal{A}_{\text{QFT}}(v)\otimes\mathcal{B}(\mathcal{H}_C))^{\mathcal{C}_v}$ is Type II, which is physically significant because it means states on this algebra have well-defined density operators and von Neumann entropies. In a `semiclassical' limit in which the state is such that fluctuations of the clock time $\hat{p}$ are suppressed relative to fluctuations of the fields, the von Neumann entropy of such an algebra agrees\footnote{To show precise agreement, one imposes a UV cutoff on the fields. But such a UV cutoff is not required for either the von Neumann entropy or the generalised entropy to be well-defined.} with the generalised entropy of the region $\mathcal{N}(s)$~\cite{Chandrasekaran_2023b,kudlerflam2024generalizedblackholeentropy,Jensen_2023,Ali_Ahmad_2024,Klinger:2023tgi}.

\changed{There is an important technical point to be made here. For Type II factors, the von Neumann entropy is only defined up to the addition of some state-independent constant. One can fix this constant by requiring the entropy of a certain state to take a certain value. Alternatively, one does not need to fix this constant if one only discusses \emph{differences} in entropies, which are unambiguous because the constant cancels. So a more precise statement of the above is that, in the semiclassical limit, differences in the generalised entropy agree with differences in the von Neumann entropy for $(\mathcal{A}_{\text{QFT}}(v)\otimes\mathcal{B}(\mathcal{H}_C))^{\mathcal{C}_v}$.}

This was exploited in~\cite{Faulkner:2024gst} to demonstrate a form of the generalised second law (GSL) that holds semiclassically. In particular, if one considers two cuts $\mathcal{S}(v),\,\mathcal{S}(v')$ and their exterior regions $\mathcal{N}(v),\,\mathcal{N}(v')$ at different times $v<v'$, one may construct the corresponding algebras
\begin{equation}
  (\mathcal{A}_{\text{QFT}}(v)\otimes\mathcal{B}(\mathcal{H}_C))^{\mathcal{C}_v}, \qquad (\mathcal{A}_{\text{QFT}}(v')\otimes\mathcal{B}(\mathcal{H}_C))^{\mathcal{C}_{v'}}.
  \label{Equation: non-isotonic algebras}
\end{equation}
One may compare the von Neumann entropies of these two algebras, for a given fixed state. Taking a semiclassical limit,~\cite{Faulkner:2024gst} showed (using the monotonicity of QFT relative entropy, and the gravitational constraints) that the von Neumann entropy of the latter algebra is greater than that of the former, in what is essentially a statement of the GSL. \changed{Here, we are comparing the entropies of two Type II factors. These entropies are only defined up to the addition of some state-independent constants, and the constants can be different for each algebra, which means that the \emph{relative} normalisation of the entropies is something that we must fix in order for this comparison to be meaningful. In~\cite{Faulkner:2024gst} this relative normalisation was fixed by requiring that the entropies of the two algebras are equal whenever the QFT is in the vacuum state $\ket{\Omega}$.}\footnote{\changed{In any case, it is worth emphasising that the physical content of the entropy comparison is independent of this normalisation. Regardless of which normalisation is chosen, one always has the same fundamental inequality -- it is just that the various terms that appear in the inequality are labelled in different ways. After imposing a UV cutoff, one gets the same standard GSL, regardless of these choices of convention. We will comment on these points in more detail in Subsection~\ref{Subsection: normalisation ambiguities}.}}

\subsection{Beyond the semiclassical regime?}

In~\cite{Faulkner:2024gst}, some corrections to the semiclassical limit were investigated, but it was not conclusive whether the GSL would hold beyond this regime. The main obstruction to going beyond the semiclassical limit is that the latter algebra in~\eqref{Equation: non-isotonic algebras} is not a subalgebra of the former algebra, despite the corresponding spacetime region for the latter being a subset of the former: $\mathcal{N}(v')\subset\mathcal{N}(v)$. In algebraic quantum field theory, one usually assumes that for any two spacetime regions obeying $\mathcal{U}\subseteq \mathcal{V}$ the corresponding algebras obey $\mathcal{A}_{\mathcal{U}}\subseteq\mathcal{A}_{\mathcal{V}}$. This property is called `isotony'; the algebras constructed here \emph{do not} obey it. This is physically and intuitively unreasonable, since any operations or observations which can be carried out in the smaller region $\mathcal{N}(v')$ should also be capable of being carried out in the larger region $\mathcal{N}(v)$. If one did have isotony, one could exploit quantum information formulae (such as the monotonicity of relative entropy for subalgebras) in order to attempt to demonstrate a version of the GSL, without any assumption on the state (semiclassical or otherwise). Indeed, we will do so later in the paper.

The reason the algebras constructed above do not obey isotony is as follows. One does have isotony for the kinematical algebras (i.e.\ before imposing boost invariance):
\begin{equation}
  \mathcal{A}_{\text{QFT}}(v')\otimes\mathcal{B}(\mathcal{H}_C) \subseteq \mathcal{A}_{\text{QFT}}(v)\otimes\mathcal{B}(\mathcal{H}_C), \qq{since}\mathcal{A}_{\text{QFT}}(v') \subset \mathcal{A}_{\text{QFT}}(v).
  \label{Equation: kinematical isotony}
\end{equation}
In constructing~\eqref{Equation: non-isotonic algebras}, we imposed invariance under the boosts of the respective regions -- but the two boosts are different, as depicted in Figure~\ref{Figure: different boosts}. An operator invariant under one boost need not be invariant under another boost, and this is what breaks isotony.

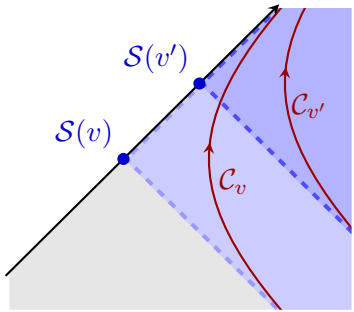
\begin{figure}
  \centering
  \begin{tikzpicture}[scale=1]
    \fill[black!10] (-1.5,-2) -- (-1.5,-1.5) -- (2,2) -- (3,2) -- (3,-2) -- (2,-2);

    \fill[blue!20] (0,0) -- (2,2) -- (3,2) -- (3,-2) -- (2,-2) -- (0,0);
    \begin{scope}
      \clip (0,0) -- (2,2) -- (3,2) -- (3,-2) -- (2,-2) -- (0,0);
      \draw[line width=3pt,blue!40,dashed] (2,-2) -- (0,0) -- (2,2);
    \end{scope}
    \fill[blue!30] (1,1) -- (2,2) -- (3,2) -- (3,-1) -- (1,1);
    \begin{scope}
      \clip (1,1) -- (2,2) -- (3,2) -- (3,-1) -- (1,1);
      \draw[line width=3pt,blue!70,dashed] (3,-1) -- (1,1) -- (2,2);
    \end{scope}
    \begin{scope}
      \clip (0,0) -- (2,2) -- (3,2) -- (3,-2) -- (2,-2) -- (0,0);
      \tikzstyle flowline=[red!60!black,thick,postaction={decorate,decoration={markings,mark=at position .53 with {\arrow{stealth}}}}]
      \draw[flowline] (3,-3) .. controls (.5,-.5) and (.5,.5) .. (3,3);
      \draw[flowline,shift={(1,1)}] (3,-3) .. controls (.5,-.5) and (.5,.5) .. (3,3);
    \end{scope}

    \node[red!60!black] at (1.45,-0.3) {$\mathcal{C}_v$};
    \node[red!60!black,shift={(1,1)}] at (1.45,-0.3) {$\mathcal{C}_{v'}$};

    \draw[thick,-stealth] (-1.55,-1.55) -- (2.05,2.05);

    \fill[blue!80!black] (0,0) circle (0.08) node[above left] {$\mathcal{S}(v)$};
    \fill[blue!80!black] (1,1) circle (0.08) node[above left] {$\mathcal{S}(v')$};

  \end{tikzpicture}
  \caption{The boosts associated with different cuts $\mathcal{S}(v)$ and $\mathcal{S}(v')$ do not preserve the corresponding exterior regions. Also, operators commuting with the generator $\mathcal{C}_v$ of one boost do not necessarily commute with the generator $\mathcal{C}_{v'}$ of the other.}
  \label{Figure: different boosts}
\end{figure}

Of course, in gravity all boosts are gauge symmetries, so we should really (at a minimum) be looking at operators which are invariant under both boosts. If we do so, then isotony manifestly follows from~\eqref{Equation: kinematical isotony}:
\begin{equation}
  (\mathcal{A}_{\text{QFT}}(v')\otimes\mathcal{B}(\mathcal{H}_C))^{\mathcal{C}_v,\mathcal{C}_{v'}} \subseteq (\mathcal{A}_{\text{QFT}}(v)\otimes\mathcal{B}(\mathcal{H}_C))^{\mathcal{C}_v,\mathcal{C}_{v'}}.
  \label{Equation: physical isotony?}
\end{equation}
One could then attempt to find entropy inequalities for these algebras. However, this is unsatisfactory for a different reason: these two algebras are in fact exactly equal to one another. The reason is that the difference $\mathcal{C}_{v'}-\mathcal{C}_v$ is proportional to $P-\Theta$, where, as above, $P$ is the integrated averaged null energy of the fields, while $\Theta$ is the expansion of the horizon in the asymptotic past $v\to \infty$; this combination $P-\Theta$ is the constraint generating null translations. Since $\Theta$ commutes with $\mathcal{A}_{\text{QFT}}(v)$, it follows that
\begin{equation}
  (\mathcal{A}_{\text{QFT}}(v)\otimes\mathcal{B}(\mathcal{H}_C))^{\mathcal{C}_v,\mathcal{C}_{v'}} = (\mathcal{A}_{\text{QFT}}(v)\otimes\mathcal{B}(\mathcal{H}_C))^{\mathcal{C}_v,P} = (\qty({\mathcal{A}_{\text{QFT}}(v)})^P\otimes\mathcal{B}(\mathcal{H}_C))^{\mathcal{C}_v}.
\end{equation}
But, as shown in~\cite{Borchers2009InstituteFM},
\begin{equation}
  \qty(\mathcal{A}_{\text{QFT}}(v))^P = \mathcal{A}_{\text{QFT}}(\infty) := \bigcap_{v''>v} \mathcal{A}_{\text{QFT}}(v''),
\end{equation}
i.e.\ the only field operators invariant under $P$ are those located in the asymptotic future $v\to\infty$. So
\begin{equation}
  (\mathcal{A}_{\text{QFT}}(v)\otimes\mathcal{B}(\mathcal{H}_C))^{\mathcal{C}_v,\mathcal{C}_{v'}} = ({\mathcal{A}_{\text{QFT}}(\infty)}\otimes\mathcal{B}(\mathcal{H}_C))^{\mathcal{C}_v},
\end{equation}
and by similar reasoning,
\begin{equation}
  (\mathcal{A}_{\text{QFT}}(v')\otimes\mathcal{B}(\mathcal{H}_C))^{\mathcal{C}_v,\mathcal{C}_{v'}} = (\mathcal{A}_{\text{QFT}}(v')\otimes\mathcal{B}(\mathcal{H}_C))^{\mathcal{C}_v,P} = ({\mathcal{A}_{\text{QFT}}(\infty)}\otimes\mathcal{B}(\mathcal{H}_C))^{\mathcal{C}_v}.
\end{equation}
Since the algebras in~\eqref{Equation: physical isotony?} are equivalent, their von Neumann entropies in any state will be the same, so we cannot hope to use them to deduce anything like the GSL.

Physically speaking, the reason this happened is that the boosts generated by $\mathcal{C}_v,\,\mathcal{C}_{v'}$ do not preserve the regions $\mathcal{N}(v'),\,\mathcal{N}(v)$, respectively. Thus, thinking of these as fixed spacetime subregions (as we have been doing in this section) is not compatible with gauge-invariance, and hence should not be expected to give physically meaningful subsystems.

An alternative (and better) approach, which we use in the rest of the paper, is relational. In particular, if one gives a \emph{dynamical} and \emph{covariant} definition of the regions outside the horizon, in terms of degrees of freedom that transform in the appropriate way, then one will get well-defined physical subsystems. This is not at all a radical thing to do -- it has been recognised many times~\cite{DeWitt:1962cg,PhysRevLett.4.432,RevModPhys.33.510,PhysRev.124.274,Rovelli_1991,Goeller:2022rsx,Carrozza:2021gju,Carrozza:2022xut} that a complete formulation of local observables in quantum gravity should be relational in nature. Indeed, the crossed product algebra construction with the boost clock $\mathcal{H}_C$ considered above is one manifestation of this, with the clock being a degree of freedom relative to which we observe other degrees of freedom. The remainder of the paper generalises this by introducing a degree of freedom (a `dynamical cut') relative to which we can define spacetime regions exterior to the horizon.

\section{Algebras outside dynamical cuts}
\label{Section: dynamical cuts}

In the previous section, we only considered one-dimensional gauge groups of boosts. We will now consider an enlarged \emph{two}-dimensional gauge group, given by the semidirect product of boosts $v\mapsto e^{2\pi t}v$ with null translations $v\mapsto v + s$ (depicted in Figure~\ref{Figure: gauge group}):
\begin{equation}
  G = \qty\Big{v\mapsto e^{2\pi t}v + s\Bigm\vert t\in\RR,\,s\in\RR}.
\end{equation}
This group contains boosts around any constant $v$ cut. It has a unitary representation on the fields given by
\begin{equation}
  U_{\text{QFT}}(s,t) = e^{-isP}e^{-itH}.
\end{equation}
Since $\mathscr{H}$ is a horizon, we can take the asymptotic future area to be invariant under null translations. Thus, the unitary representation of $G$ on $\mathcal{H}_C$ is given by
\begin{equation}
  U_C(s,t) = e^{-it\hat{q}}.
\end{equation}
As described in the previous section, to properly account for this larger gauge group in the context of spacetime regions exterior to a horizon, we need to define such regions relationally. This is the objective of this section.

\begin{figure}
  \centering
  \begin{tikzpicture}[scale=1]
    \fill[black!10] (-1.5,-2) -- (-1.5,-1.5) -- (2,2) -- (3,2) -- (3,-2) -- (2,-2);

    \begin{scope}
      \clip (-1.5,-2) -- (-1.5,-1.5) -- (2,2) -- (3,2) -- (3,-2) -- (2,-2);
      \draw[dashed,gray] (0,0) -- (3,-3);
      \tikzstyle flowline=[thick,postaction={decorate,decoration={markings,mark=at position .53 with {\arrow[scale=1.3]{stealth}}}}]
      \draw[flowline,red!70!black] (3,-3) .. controls (.5,-.5) and (.5,.5) .. (3,3);
      \draw[flowline,red!70!black] (3,-3) .. controls (.5,-.5) and (-.5,-.5) .. (-3,-3);
      \draw[flowline,blue!70!black,shift={(0,-0.4)}] (-3,-3) -- (3,3);
    \end{scope}

    \draw[thick,-stealth] (-1.55,-1.55) -- (2.05,2.05);
    \node[above left] at (-0.7,-0.8) {$\mathscr{H}$};
    \node[above left] at (0.8,0.7) {$v$};
  \end{tikzpicture}
  \caption{We will account for a gauge group consisting of both boosts {\color{red!70!black}$v\to e^{2\pi t}v$} and null translations {\color{blue!70!black}$v \to v+s$}. The dashed line represents $v=0$, but this group contains boosts around any fixed $v$.}
  \label{Figure: gauge group}
\end{figure}
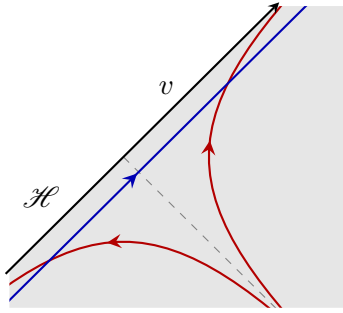

\subsection{Including a single dynamical cut}

We now introduce a \emph{dynamical cut} $\mathcal{S}$, which has Hilbert space $\mathcal{H}_{\mathcal{S}}= L^2(\RR)$. The position operator $\hat{v}$ on this Hilbert space corresponds to the location of the cut $\mathcal{S}$ in the coordinate $v$. There is a unitary representation of the gauge group $G$ on the cut, given in terms of the position eigenstates by
\begin{equation}
  U_{\mathcal{S}}(s,t) \ket{v} = e^{\pi t}\ket*{e^{2\pi t}v + s}
\end{equation}
(the factor of $e^{\pi t}$ is required for unitarity). This may be written as
\begin{equation}
  U_{\mathcal{S}}(s,t) = e^{-is\hat{k}} e^{-it\hat{d}}, \qq{where} \hat{d} = \pi\pb*{\hat{v}}{\hat{k}} = \pi(\hat{v}\hat{k}+\hat{k}\hat{v}),
\end{equation}
and $\hat{k}$ is the conjugate momentum to $\hat{v}$. The dynamical cut is an \emph{ideal} quantum reference frame; we comment on the non-ideal case in the Conclusion.

Let $\mathcal{N}$ denote the region exterior to $\mathcal{S}$. This region is not fixed in spacetime; rather its location depends on the state of $\mathcal{S}$. What von Neumann algebra of operators should we assign to the subsystem associated with this dynamical spacetime region?

To answer this, it helps to first address a class of operators we should \emph{not} include in this algebra: those which move $\mathcal{S}$. Indeed, suppose an operator $a$ moves $\mathcal{S}$ to the past. This will extend the region $\mathcal{N}$ into the past, and there will be some operators with support in the new $\mathcal{N}$ which did not have support in the old $\mathcal{N}$. Let $b$ denote such an operator. Then $a^\dagger b a$ is an operator which moves $\mathcal{S}$ to the past, acts in the enlarged $\mathcal{N}$, and then moves $\mathcal{S}$ back to its original location; overall, this operator acts \emph{outside} $\mathcal{N}$. Similarly, if $a$ moves $\mathcal{S}$ to the future, then $a b a^\dagger$ will act outside $\mathcal{N}$. We should not allow operators of this kind, since they are in conflict with our original intention to define the algebra of operators \emph{in} $\mathcal{N}$. Thus, in either case $a,a^\dagger$ must be excluded from the algebra.

So all operators in the algebra must commute with $\hat{v}$, and from this observation it is fairly clear how to proceed. For the fields, if the cut $\mathcal{S}$ is in a position eigenstate $\ket{v}$, then we should be able to act with any $a\in\mathcal{A}_{\text{QFT}}(v)$. This means that the algebra should contain any operators of the form $a \otimes \dyad{v}$, for $v\in\RR$ and $a\in\mathcal{A}_{\text{QFT}}(v)$.\footnote{More precisely, such operators are \emph{affiliated} with the algebra.} The span of all such operators may be written as a direct integral
\begin{equation}
  \mathcal{A}_{\text{QFT}}(\hat{v}) := \int^\oplus_\RR \dd{v} \mathcal{A}_{\text{QFT}}(v)\otimes\dyad{v} = \qty\Big{a\otimes\dyad{v}\Bigm\vert v\in\RR,\, a\in\mathcal{A}_{\text{QFT}}(v)}''.
  \label{Equation: dressed QFT algebras}
\end{equation}
Regardless of the location of $\mathcal{S}$, the exterior region $\mathcal{N}$ also always contains the asymptotic future of the horizon, so we should additionally include all operators acting on $\mathcal{H}_C$. Thus, the overall algebra of operators in $\mathcal{N}$ (before imposing gauge invariance) takes the form
\begin{equation}
  \mathcal{A}_{\mathcal{N}}=\mathcal{A}_{\text{QFT}}(\hat{v})\otimes \mathcal{B}(\mathcal{H}_C).
\end{equation}
It should be noted that this algebra has a non-trivial center consisting of functions of $\hat{v}$. Thus, unlike the case of a fixed cut, this algebra is not a factor. However, upon imposing gauge invariance, this center becomes trivial; indeed, only constant functions of $\hat{v}$ are null translation invariant. So the gauge-invariant algebra will be a factor.

Let us now impose gauge-invariance. One may do this by first imposing null translation invariance, and then imposing boost invariance (the other way around also works and is similar). The boost clock $\mathcal{H}_C$ is already null translation invariant, so we only need to impose null translation invariance on $\mathcal{A}_{\text{QFT}}(\hat{v})$. A general element of this algebra may be written $A=\int_{-\infty}^\infty \dd{v} a(v)\otimes\dyad{v}$ with $a(v)\in\mathcal{A}_{\text{QFT}}(v)$, and null translation invariance $A = e^{is(P+\hat{k})}A e^{-is(P+\hat{k})}$ is satisfied if and only if $a(v)= e^{ivP}a(0)e^{-ivP}$. Substituting this into the definition of $A$, one finds $A = e^{i\hat{v}P} a(0) e^{-i\hat{v}P}$. Thus, the null translation invariant subalgebra is
\begin{equation}
  e^{i\hat{v}P}\mathcal{A}_{\text{QFT}}(0)e^{-i\hat{v}P}\otimes\mathcal{B}(\mathcal{H}_C)\subset \mathcal{A}_{\mathcal{N}}.
  \label{Equation: null translation invariant single cut}
\end{equation}
Next we impose boost invariance. The overall generator of a boost (around $v=0$ -- note we can use a boost around any fixed $v$ since we already imposed null translation invariance) is
\begin{equation}
  H + \hat{q} + \mathcal{Q} + \hat{d},
\end{equation}
where $\mathcal{Q}=\mathcal{Q}_0$. The situation is now basically the same as in the previous section, where we had a fixed cut, and only boost invariance to impose. The difference is that now the fixed QFT subalgebra $\mathcal{A}_{\text{QFT}}(v)$ is replaced by the \emph{dressed} QFT subalgebra $e^{-i\hat{v}P}\mathcal{A}_{\text{QFT}}(0)e^{i\hat{v}P}$. But, for the same reasons as before, the invariant subalgebra is a crossed product
\begin{align}
  \mathcal{A}_{\mathcal{N}}^G &= e^{-i\hat{v}P}\mathcal{A}_{\text{QFT}}(0)e^{i\hat{v}P}\rtimes_\alpha\RR\\
  &= \qty{e^{iP\hat{v}}e^{iH\hat{p}}ae^{-iH\hat{p}}e^{-iP\hat{v}},\hat{q}\Bigm\vert a\in \mathcal{A}_{\text{QFT}}(0)}''
  \label{Equation: gauge invariant single cut}
\end{align}
As before, this algebra is made up of dressed operators $e^{iP\hat{v}}e^{iH\hat{p}}ae^{-iH\hat{p}}e^{-iP\hat{v}}$ and reorientations generated by $\hat{q}$.
Here, $\alpha$ is a \emph{dressed} boost, not a fixed boost as in the previous section:
\begin{equation}
  \alpha_t: \quad e^{i\hat{v}P}\mathcal{A}_{\text{QFT}}(0)e^{-i\hat{v}P} \to e^{i\hat{v}P}\mathcal{A}_{\text{QFT}}(0)e^{-i\hat{v}P},\quad
  A \mapsto e^{itH_{\hat{v}}}A e^{-itH_{\hat{v}}},
\end{equation}
where $H_{\hat{v}}=H-2\pi\hat{v} P$ is a generator of a boost in the QFT around the dynamical cut $\mathcal{S}$. Dressed operators may also be written in terms of this dressed boost:
\begin{equation}
  e^{iP\hat{v}}e^{iH\hat{p}}ae^{-iH\hat{p}}e^{-iP\hat{v}}
  =
  e^{iH_{\hat{v}}\hat{p}}e^{iP\hat{v}}ae^{-iP\hat{v}}e^{-iH_{\hat{v}}\hat{p}}.
\end{equation}
Thus, we have succeeded in writing the gauge-invariant algebra outside a dynamical cut as a crossed product. It is worth pointing out that here we are only taking a crossed product with a one-dimensional group of boosts -- not the full two-dimensional gauge group. This is because we are not allowing operators which move the cut, which would be included as reorientations associated with null translations, if we were to do the full two-dimensional crossed product.

At this point it may be observed that the algebra $\mathcal{A}_{\mathcal{N}}^G$ is isomorphic to the boost-invariant algebra of a fixed cut; the two are related by the unitary map $e^{i\hat{v}P}$, which dresses null translations. Therefore, this algebra is Type II${}_\infty$; it has traces, density operators, and von Neumann entropies. Because of this isomorphism, we have not really gained much additional structure of physical interest. But so far we have only been considering a single dynamical cut. The true power of our construction will come when there are \emph{multiple} dynamical cuts available, and indeed our eventual aim is to compare the subsystems associated with different dynamical cuts.

\subsection{Accounting for multiple dynamical cuts}

So let us now generalise the previous setup to the case with multiple dynamical cuts, $\mathcal{S}_a$, $a=1,2,\dots$, with exterior regions $\mathcal{N}_a$ respectively. Each of these has Hilbert spaces $\mathcal{H}_a=L^2(\RR)$, with position operator $\hat{v}_a$ and conjugate momentum operator $\hat{k}_a$. As described previously, the generator of a null translation on $\mathcal{S}_a$ is $\hat{k}_a$, while the generator of a boost (around $v=0$) is $\hat{d}_a=\pi\pb*{\hat{v}_a}{\hat{k}_a}$.

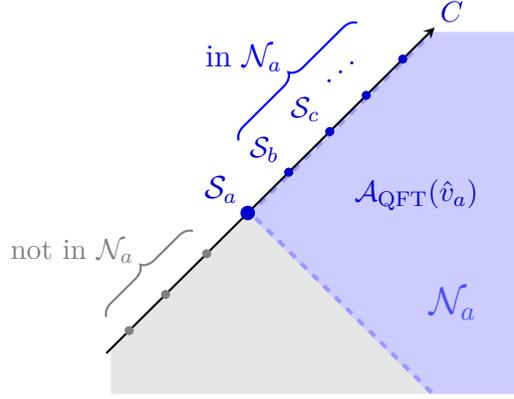
\begin{figure}
  \centering
  \begin{tikzpicture}[scale=1.2]
    \fill[black!10] (-1.5,-2) -- (-1.5,-1.5) -- (2,2) -- (2,-2);
    \fill[blue!20] (0,0) -- (2,2) -- (3,2) -- (3,-2) -- (2,-2) -- (0,0);
    \begin{scope}
      \clip (0,0) -- (2,2) -- (3,2) -- (3,-2) -- (2,-2) -- (0,0);
      \draw[line width=3pt,blue!40,dashed] (2,-2) -- (0,0) -- (2,2);
    \end{scope}

    \draw[thick,-stealth] (-1.55,-1.55) -- (2.05,2.05);

    \node[blue!70] at (2.25,-1) {\Large$\mathcal{N}_a$};
    \node[blue!80!black] at (1.85,0.2) {$\mathcal{A}_{\text{QFT}}(\hat{v}_a)$};

    \fill[blue!80!black] (0,0) circle (0.08) node[above left] {\large$\mathcal{S}_a$};

    \fill[blue!80!black] (0.45,0.45) circle (0.05) node[above left] {$\mathcal{S}_b$};
    \fill[blue!80!black] (0.9,0.9) circle (0.05) node[above left] {$\mathcal{S}_c$};
    \fill[blue!80!black] (1.3,1.3) circle (0.05) node[above left] {$\iddots$};
    \fill[blue!80!black] (1.7,1.7) circle (0.05);

    \fill[gray] (-0.45,-0.45) circle (0.05);
    \fill[gray] (-0.9,-0.9) circle (0.05);
    \fill[gray] (-1.3,-1.3) circle (0.05);

    \draw [blue,thick,decorate,decoration={brace,amplitude=5pt,mirror,raise=0ex}] (1.2,2.1) -- (-0.1,0.8) node[midway,above left]{\large in $\mathcal{N}_a$\,};
    \draw [gray,thick,decorate,decoration={brace,amplitude=5pt,mirror,raise=0ex}] (-0.6,-0.2) -- (-1.55,-1.15) node[midway,above left]{not in $\mathcal{N}_a$\,};

    \node[above right, blue!80!black] at (2,2) {$C$};
  \end{tikzpicture}
  \caption{The region $\mathcal{N}_a$ outside a dynamical cut $\mathcal{S}_a$ contains the fields in $\mathcal{N}_a$, and the asymptotic future horizon area / boost clock $C$. Additionally, it contains all the dynamical cuts which are to the future of $\mathcal{S}_a$. The algebra $\mathcal{A}_{\mathcal{N}_a}$ must account for all these degrees of freedom.}
  \label{Figure: what's in N_a}
\end{figure}

Consider the exterior region $\mathcal{N}_a$ associated with a particular choice of cut $\mathcal{S}_a$. We will now construct the algebra $\mathcal{A}_{\mathcal{N}_a}$ of this region. The story goes much as described above, but with one key difference.
Depending on the state, the region $\mathcal{N}_a$ can also contain the \emph{other} cuts $\mathcal{S}_b\ne \mathcal{S}_a$ (as depicted in Figure~\ref{Figure: what's in N_a}), and we must account for operators acting on these cuts. In particular, $\mathcal{S}_b$ is in $\mathcal{N}_a$ if and only if $\mathcal{S}_b$ is to the future of $\mathcal{S}_a$. It helps to now note two things:
\begin{itemize}
  \item Since no operators in $\mathcal{A}_{\mathcal{N}_a}$ can act on cuts to the past of $\mathcal{S}_a$, no operator can move $\mathcal{S}_b$ from the past of $\mathcal{S}_a$ to its future.
  \item Conversely, if $\mathcal{S}_b$ is located to the \emph{future} of $\mathcal{S}_a$, then no operator in $\mathcal{A}_{\mathcal{N}_a}$ can move it to the \emph{past} of $\mathcal{S}_a$ -- since if an $a\in\mathcal{A}_{\mathcal{N}_a}$ does this, then $a^\dagger$ would be an example of an operator violating the previous bullet point.
\end{itemize}
It is useful to therefore decompose the Hilbert space of each cut $\mathcal{S}_b$ as
\begin{equation}
  \mathcal{H}_b = \mathcal{H}_b^{>v}\oplus\mathcal{H}_b^{<v},
\end{equation}
where
\begin{equation}
  \mathcal{H}_b^{>v} = \theta(\hat{v}_b-v) \mathcal{H}_b = L^2((v,\infty)),\qquad \mathcal{H}_b^{<v} = \theta(v-\hat{v}_b) \mathcal{H}_b = L^2((-\infty,v))
\end{equation}
are the sectors of the Hilbert space in which $\mathcal{S}_b$ is to the future and past of $v$ respectively.\footnote{States with support at $\hat{v}_b=v$ are included once we perform a completion in the Hilbert space norm, which is part of the definition of a Hilbert space direct sum.} Here $\theta$ is a step function, so for example $\theta(\hat{v}_b-v) = \int_v^\infty\dd{v_b}\dyad{v_b}_b$ (where $\ket{v_b}_b$ are the eigenstates of $\hat{v}_b$, and we use $\dyad{v_b}_b$ as shorthand for $\ket{v_b}_b\bra{v_b}_b$). By the above bullet points, when $\mathcal{S}_a$ is located at $v$, operators in $\mathcal{N}_a$ acting on $\mathcal{S}_b$ must preserve these two sectors. Moreover, they can only act non-trivially on $\mathcal{H}_b^{>v}$. The algebra of all such operators may be written
\begin{equation}
  \mathcal{A}_{\mathcal{S}_b}(v) = \mathcal{B}(\mathcal{H}_b^{>v})\oplus \mathbb{C}\mathds{1}_{\mathcal{H}_b^{<v}}.
\end{equation}
Properly conditioning on the location of $\mathcal{S}_a$ amounts to performing a similar direct integral to~\eqref{Equation: dressed QFT algebras}. One finds that $\mathcal{A}_{\mathcal{N}_a}$ should contain all operators in
\begin{equation}
  \mathcal{A}_{\mathcal{S}_b}(\hat{v}_a) = \int^\oplus_{\RR}\dd{v_a}\mathcal{A}_{\mathcal{S}_b}(v_a)\otimes\dyad{v_a}_a.
\end{equation}
Note that we also have $\mathcal{A}_{\mathcal{S}_b}(v+s)\in\mathcal{A}_{\mathcal{S}_b}(v)$ for $s\ge 0$, and
\begin{equation}
  \mathcal{A}_{\mathcal{S}_b}(v+s) = e^{i\hat{k}_bs}\mathcal{A}_{\mathcal{S}_b}(v)e^{-i\hat{k}_bs},
  \label{Equation: cut algebra null translation}
\end{equation}
analogous to~\eqref{Equation: field algebra null translation}.

Altogether, including all operators acting on the fields and the other cuts (conditioned on the location of $\mathcal{S}_a$), as well as all operators acting on the boost clock Hilbert space $\mathcal{H}_C$, one finds that the kinematical algebra of operators associated with $\mathcal{N}_a$ is
\begin{equation}
  \mathcal{A}_{\mathcal{N}_a} = \qty\Big(\mathcal{A}_{\text{QFT}}(\hat{v}_a) \vee \bigvee_{b\ne a} \mathcal{A}_{\mathcal{S}_b}(\hat{v}_a)) \otimes \mathcal{B}(\mathcal{H}_C).
\end{equation}
Here, the notation $\mathcal{K}\vee\mathcal{L}$ denotes the `join' of von Neumann algebras $\mathcal{K},\mathcal{L}$, which is the smallest von Neumann algebra containing both $\mathcal{K}$ and $\mathcal{L}$ as subalgebras. In this case we can alternatively write
\begin{equation}
  \mathcal{A}_{\text{QFT}}(\hat{v}_a) \vee \bigvee_{b\ne a} \mathcal{A}_{\mathcal{S}_b}(\hat{v}_a)= \int^\oplus_{\RR}\dd{v_a}\mathcal{A}_{\text{QFT}}(v_a)\otimes\bigotimes_{b\ne a}\mathcal{A}_{\mathcal{S}_b}(v_a)\otimes\dyad{v_a}_a.
\end{equation}
Let us now impose gauge invariance. The overall generators of boosts and null translations are
\begin{equation}
  \mathcal{C} = H+\hat{q}+\mathcal{Q} + \sum_a \hat{d}_a, \qquad
  \mathcal{D} = P - \Theta + \sum_a\hat{k}_a,
\end{equation}
respectively (recall $\Theta$ is the horizon expansion at $v\to -\infty$). We proceed as before, by first imposing null translation invariance and then imposing boost invariance. The structure is essentially the same as previously, but now including the operators acting on the other cuts. One finds for the null translation invariant subalgebra (generalising~\eqref{Equation: null translation invariant single cut}):
\begin{equation}
  \mathcal{A}_{\mathcal{N}_a}^{\mathcal{D}} = e^{i\hat{v}_a(P + \sum_{b\ne a}\hat{k}_b)}\qty\Big(\mathcal{A}_{\text{QFT}}(0)\otimes\bigotimes_{b\ne a}\mathcal{A}_{\mathcal{S}_b}(0))e^{-i\hat{v}_a(P + \sum_{b\ne a}\hat{k}_b)}\otimes\mathcal{B}(\mathcal{H}_C).
\end{equation}
Now imposing boost invariance, one finds that the full gauge invariant subalgebra is a crossed product (generalising~\eqref{Equation: gauge invariant single cut}):
\begin{align}
  \mathcal{A}_{\mathcal{N}_a}^G &= e^{i\hat{v}_a(P + \sum_{b\ne a}\hat{k}_b)}\qty\Big(\mathcal{A}_{\text{QFT}}(0)\otimes\bigotimes_{b\ne a}\mathcal{A}_{\mathcal{S}_b}(0))e^{-i\hat{v}_a(P + \sum_{b\ne a}\hat{k}_b)}\rtimes_\alpha\RR\\
  &=
  \begin{multlined}[t]
    \Big\{e^{i\hat{v}_a(P+\sum_{b\ne a}\hat{k}_b)}e^{i\hat{p}(H+\sum_{b \ne a}\hat{d}_b)}ae^{-i\hat{p}(H+\sum_{b\ne a}\hat{d}_b)}e^{-i\hat{v}_a(P+\sum_{b\ne a}\hat{k}_b)},\hat{q}
      \Bigm\vert \\
    a\in \mathcal{A}_{\text{QFT}}(0)\otimes\bigotimes_{b\ne a}\mathcal{A}_{\mathcal{S}_b}(0)\Big\}''.
  \end{multlined}
  \label{Equation: invariant algebra}
\end{align}
As before, $\alpha$ is a dressed boost $\alpha_t:A\mapsto e^{itH_{\hat{v}_a}}Ae^{-itH_{\hat{v}_a}}$ in the region $\mathcal{N}_a$, where
\begin{equation}
  H_{\hat{v}_a} = H+\sum_{b\ne a}\hat{d}_b - 2\pi \hat{v}_a\qty\Big(P+\sum_{b\ne a}\hat{k}_b)
  \label{Equation: dressed boost generator}
\end{equation}
is the generator of this boost acting on the fields and cuts.

Unlike in the case of a single dynamical cut, the gauge-invariant algebra $\mathcal{A}_{\mathcal{N}_a}^G$ has a non-trivial center. Indeed, before imposing gauge-invariance the center of the algebra is
\begin{equation}
  Z(\mathcal{A}_{\mathcal{N}_a}) = \qty{f(\hat{v}_a),\,\theta(\hat{v}_b-\hat{v}_a)\Bigm\vert f:\RR\to \CC,\, b\ne a}''.
\end{equation}
The functions $f(\hat{v}_a)$ of the location of $\mathcal{S}_a$ are knocked out upon imposing gauge invariance, as previously. However, the other operators $\theta(\hat{v}_b-\hat{v}_a)$ survive in the center of the gauge-invariant algebra, and imposing gauge invariance does not enlarge the center (by Lemma~\ref{Lemma: crossed product center}, using that boosts are outer automorphisms of the QFT algebras). Thus, the center of the gauge-invariant algebra is given by
\begin{equation}
  Z(\mathcal{A}_{\mathcal{N}_a}^G) = \qty{\theta(\hat{v}_b-\hat{v}_a)\Bigm\vert b\ne a}''.
  \label{Equation: gauge-invariant center}
\end{equation}
Each central element measures whether other cuts $\mathcal{S}_b$ are to the future or past of $\mathcal{S}_a$.
% The gauge-invariant algebra $\mathcal{A}_{\mathcal{N}_a}^G$ thus decomposes as a direct sum of factors corresponding to different options for which set of cuts $\mathcal{S}_b$ falls within $\mathcal{N}_a$. As shown in the next subsection, each factor is Type II${}_\infty$.

\subsection{Traces on the algebras}
\label{Subsection: traces}

The algebra $\mathcal{A}_{\mathcal{N}_a}^G$ can be equipped with a family of traces parametrised by its positive central elements. Below, we will construct one such trace, which we denote $\Tr_a$. Any other trace may then be written~\cite{Sorce:2023fdx,TheoryOfOperatorAlgebrasI} as $\Tr_a^z$, where
\begin{equation}
  \Tr_a^z(A) = \Tr_a(zA),
  \label{Equation: trace ambiguity}
\end{equation}
for any positive central element $z\in Z(\mathcal{A}_{\mathcal{N}_a}^G)$. Given the form of the center~\eqref{Equation: gauge-invariant center}, it is clear that the different possible traces weight the sectors in which each cut is to the future or past of $\mathcal{S}_a$ by different amounts.

A trace allows one to define the density operator $\rho$ of any normal state $\Phi$ on $\mathcal{A}_{\mathcal{N}_a}^G$, and hence to find a von Neumann entropy for that state. But different choices of trace will lead to different density operators and entropies. Indeed, if $\rho$ is the density operator with respect to $\Tr_a$, then the density operator with respect to $\Tr_a^z$ is $\rho z^{-1}$, since
\begin{equation}
  \Phi(A) = \Tr_a(\rho A) = \Tr_a^z(\rho z^{-1} A) \qq{for all} A\in\mathcal{A}_{\mathcal{N}_a}^G.
\end{equation}
This then implies that the von Neumann entropies defined with the two traces differ by the expectation value of $\log z$:
\begin{equation}
  -\Tr_a^z\qty(\rho z^{-1}\log(\rho z^{-1})) = -\Tr_a(\rho\log\rho) + \Phi(\log z).
  \label{Equation: entropy ambiguity}
\end{equation}
This could in principle lead to issues in constructing inequalities relating entropies for different algebras, such as the GSL. However, our derivation of the GSL will follow from the monotonicity of relative entropy, and the value of the relative entropy does not depend on which trace one uses to define density operators -- indeed, the relative entropy can be written without any reference at all to traces or density operators~\cite{Araki:1976zv}. Therefore, in the rest of the paper we will just use the particular trace $\Tr_a$ which we construct below, thereby fixing this ambiguity in density operators and entropies (but one could use any of the other traces $\Tr_a^z$ and derive the same final result for the GSL -- \changed[indeed the expectation value of $\log z$ would contribute equally to the two sides of~\eqref{Equation: exact GSL}, and thus cancel out, giving the same inequality]{we will comment on this in more detail in Subsection~\ref{Subsection: normalisation ambiguities}}).

Let us now construct the trace. There are several ingredients that go into this. First, we need a normal weight on $\mathcal{A}_{\text{QFT}}(0)\otimes\bigotimes_{b\ne a}\mathcal{A}_{\mathcal{S}_b}(0)$ which satisfies the KMS condition for the boost generated by $H+\sum_{b\ne a}\hat{d}_b$. An appropriate such weight is the tensor product
\begin{equation}
  \Psi_a = \Psi_{\text{QFT}}\otimes \bigotimes_{b\ne a}\Psi_{\mathcal{S}_b},
\end{equation}
where
\begin{equation}
  \Psi_{\text{QFT}}(a) = \bra{\Omega}a\ket{\Omega}, \qquad a \in \mathcal{A}_{\text{QFT}}(0),
\end{equation}
and $\Psi_{\mathcal{S}_b}$ is a weight on $\mathcal{A}_{\mathcal{S}_b}(0)$ defined by
\begin{equation}
  \Psi_{\mathcal{S}_b}(a\oplus \alpha \mathds{1}_{\mathcal{H}_b^{<0}}) = \tr_{\mathcal{H}_b^{>0}}\qty(e^{-\hat{d}_b}a) + \alpha,\qquad a\in \mathcal{B}(\mathcal{H}_b^{>0}), \,\alpha\in\CC.
\end{equation}
Here $\tr_{\mathcal{H}_b^{>0}}$ is the ordinary Hilbert space trace associated with $\mathcal{H}_b^{>0}$, which may be written
\begin{equation}
  \tr_{\mathcal{H}_b^{>0}}\qty(a) = \int_0^\infty\dd{v_b} \bra{v_b}_b a \ket{v_b}_b.
\end{equation}
Note that boosts $e^{-it\hat{d}_a}$ preserve $\mathcal{H}_b^{>0}$ and $\mathcal{H}_b^{<0}$, so $e^{-\hat{d}}$ is a well-defined (but unbounded) operator acting on $\mathcal{H}_b^{>0}$.\footnote{It is worth commenting on a perhaps initially confusing point. A na\"ive continuation of $e^{-it\hat{d}_b}\ket{v_b}_b = e^{\pi t}\ket*{e^{2\pi t}v_b}_b$ to $t=-i$ might lead one to conclude that $e^{-\hat{d}_b}\ket{v_b}_b=-\ket{v_b}_b$, so $e^{-\hat{d}_b}=-1$. But $e^{-\hat{d}_b}$ should be a non-negative operator by construction. A similar argument might continue $e^{it\hat{d}_b}\hat{v}_be^{-it\hat{d}_b}=e^{2\pi t}\hat{v}_b$ and $e^{it\hat{d}_b}\hat{k}_be^{-it\hat{d}_b}=e^{-2\pi t}\hat{k}_b$ to $t=-i$ to conclude that $e^{-\hat{d}_b}$ commutes with $\hat{v}_b$ and $\hat{k}_b$, which would imply $e^{-\hat{d}_b}$ is a constant multiple of the identity. But this would also mean that $\hat{d}_b$ is a multiple of the identity (since it is Hermitian), which is clearly not true. The underlying cause of failure for such arguments is the unbounded nature of the operators involved. In the case at hand the action of $e^{-\hat{d}_b}$ is better understood as a rotation by $2\pi$ in the complex plane, and in general the operators and wavefunctions it acts on may have branch cuts and singularities, which these arguments do not account for. One might hope that the Paley-Wiener theorem, which says that states in $\mathcal{H}_b^{>0}$ have holomorphic momentum space wavefunctions in the upper half plane, could help -- but it is not sufficient to constrain their contributions from the lower half plane.}

$\Psi_{\text{QFT}}$ and $\Psi_{\mathcal{S}_b}$ are normal faithful functionals. Also, $\Psi_{\text{QFT}}$ is finite, so it is a state, but $\Psi_{\mathcal{S}_b}$ is only semifinite, so it is a weight. Thus, $\Psi_a$ is a normal faithful weight. It is clear that $\Psi_{\text{QFT}}$ satisfies the KMS condition with generator $H$, because $H$ is the modular Hamiltonian of $\ket{\Omega}$. Also, $\Psi_{\mathcal{S}_b}$ satisfies the KMS condition with generator $\hat{d}_b$, since for the operators acting on $\mathcal{H}_b^{>0}$ it involves the trace with the appropriate thermal operator $e^{-\hat{d}_b}$, and all relevant operators acting on $\mathcal{H}_b^{<0}$ are multiples of the identity and so commute with everything and obey the KMS condition trivially. By linearity, $\Psi_a$ therefore obeys the KMS condition for the flow generated by $H+\sum_{b\ne a}\hat{d}_b$:
\begin{equation}
  \Psi_a\qty(e^{i(H+\sum_{b\ne a}\hat{d}_b)\tau}Be^{-i(H+\sum_{b\ne a}\hat{d}_b)\tau}A)
  =
  \Psi_a\qty(Ae^{i(H+\sum_{b\ne a}\hat{d}_b)(\tau+i)}Be^{-i(H+\sum_{b\ne a}\hat{d}_b)(\tau+i)}),
  \label{Equation: Psi KMS}
\end{equation}
for all $A,B\in\mathcal{A}_{\text{QFT}}(0)\otimes\bigotimes_{b\ne a}\mathcal{A}_{\mathcal{S}_b}(0)$.

Next, define a map
\begin{equation}
  \gamma_a: e^{i\hat{v}_a(P+\sum_{b\ne a}\hat{k}_b)}\qty\Big(\mathcal{A}_{\text{QFT}}(0)\otimes\bigotimes_{b\ne a}\mathcal{A}_{\mathcal{S}_b}(0))e^{-i\hat{v}_a(P+\sum_{b\ne a}\hat{k}_b)} \to \mathcal{A}_{\text{QFT}}(0)\otimes\bigotimes_{b\ne a}\mathcal{A}_{\mathcal{S}_b}(0)
\end{equation}
by
\begin{equation}
  \gamma_a(a)\otimes\ket{0}_a = a\ket{0}_a.
\end{equation}
This is well-defined because $a\in e^{i\hat{v}_a(P+\sum_{b\ne a}\hat{k}_b)}\qty\Big(\mathcal{A}_{\text{QFT}}(0)\otimes\bigotimes_{b\ne a}\mathcal{A}_{\mathcal{S}_b}(0))e^{-i\hat{v}_a(P+\sum_{b\ne a}\hat{k}_b)}$ commutes with $\hat{v}_a$. It is normal and faithful.

Finally, for any $A\in\mathcal{A}_{\mathcal{N}_a}^G$, we have
\begin{equation}
  \bra{0}_p A \ket{p}_p \in
  e^{i\hat{v}_a(P+\sum_{b\ne a}\hat{k}_b)}\qty\Big(\mathcal{A}_{\text{QFT}}(0)\otimes\bigotimes_{b\ne a}\mathcal{A}_{\mathcal{S}_b}(0))e^{-i\hat{v}_a(P+\sum_{b\ne a}\hat{k}_b)},
  \label{Equation: general A}
\end{equation}
where $\ket{p}_p$ are eigenstates of $\hat{p}$. This may be confirmed by writing $A$ in the form
\begin{equation}
  A=\int_{-\infty}^\infty\dd{t}e^{i\hat{v}_a(P+\sum_{b\ne a}\hat{k}_b)}e^{i\hat{p}(H+\sum_{b \ne a}\hat{d}_b)}a(t)e^{-i\hat{p}(H+\sum_{b\ne a}\hat{d}_b)}e^{-i\hat{v}_a(P+\sum_{b\ne a}\hat{k}_b)}e^{-i\hat{q}t}
\end{equation}
for $a(t)\in \mathcal{A}_{\text{QFT}}(0)\otimes\bigotimes_{b\ne a}\mathcal{A}_{\mathcal{S}_b}(0)$.
One then has
\begin{equation}
  \bra{0}_p A \ket{p}_p
  = e^{i\hat{v}_a(P+\sum_{b\ne a}\hat{k}_b)}a(p)e^{-i\hat{v}_a(P+\sum_{b\ne a}\hat{k}_b)}.
\end{equation}

Having prepared these ingredients, one can define a trace functional $\Tr_a:\mathcal{A}_{\mathcal{N}_a}^G\to\CC$ by
\begin{equation}
  \Tr_a(A) = \Psi_a\qty(\gamma_a\qty(\bra{0}_pe^{-\hat{q}/2}Ae^{-\hat{q}/2}\ket{0}_p)).
\end{equation}
This is a normal weight by construction.
For $A$ in the form~\eqref{Equation: general A}, one may confirm that
\begin{equation}
  \Tr_a(e^{\hat{q}/2}A^\dagger Ae^{\hat{q}/2}) = \int_{-\infty}^\infty\dd{t}\Psi_a\qty(a(t)^\dagger a(t)),
\end{equation}
so $\Tr_a(e^{\hat{q}/2}A^\dagger Ae^{\hat{q}/2})=0$ if and only if $a(t)=0$ for all $t$, i.e.\ if and only if $Ae^{\hat{q}/2}=0$; thus, $\Tr_a$ is faithful.
Gauge invariance implies, for any $A\in\mathcal{A}_{\mathcal{N}_a}^G$,
\begin{equation}
  \bra{p}_pAe^{-\hat{q}/2}\ket{0}_p = e^{i(p+i/2)H_{\hat{v}_a}}\bra{0}_p e^{-\hat{q}/2}A\ket{-p}_pe^{-i(p+i/2)H_{\hat{v}_a}},
\end{equation}
which gives, for $A,B\in\mathcal{A}_{\mathcal{N}_A}^G$,
\begin{multline}
  \bra{0}_pe^{-\hat{q}/2}ABe^{-\hat{q}/2}\ket{0}_p = \int_{-\infty}^\infty\dd{p}\bra{0}_pe^{-\hat{q}/2}A\ket{p}_p\bra{p}_pBe^{-\hat{q}/2}\ket{0}_p \\*
  = \int_{-\infty}^\infty\dd{p}e^{i(p-i/2)H_{\hat{v}_a}}\bra{-p}_pAe^{-\hat{q}/2}\ket{0}_pe^{-H_{\hat{v}_a}}\bra{0}_pe^{-\hat{q}/2}B\ket{-p}_pe^{-i(p+i/2)H_{\hat{v}_a}}.
\end{multline}
Using that $\Psi_a$ is boost invariant, i.e.\ $\Psi_a(e^{ip(H+\sum_{b\ne a}\hat{d}_b)}(\cdot)e^{-ip(H+\sum_{b\ne a}\hat{d}_b)})=\Psi_a(\cdot)$, and that $\gamma_a(ab)=\gamma_a(a)\gamma_a(b)$, we then have (since the action of $\gamma_a$ converts each $H_{\hat{v}_a}$ that appears into an $H+\sum_{b\ne a}\hat{d}_b$)
\begin{align}
  \Tr_a(AB) &= \int_{-\infty}^\infty\dd{p}\Psi_a\qty(\gamma_a\qty(\bra{-p}_pe^{-\hat{q}/2}A\ket{0}_p)e^{-H-\sum_{b\ne a}\hat{d}_b}\gamma_a\qty(\bra{0}_pBe^{-\hat{q}/2}\ket{-p}_p)e^{H+\sum_{b\ne a}\hat{d}_b}) \\
  &= \int_{-\infty}^\infty\dd{p}\Psi_a\qty(\gamma_a\qty(\bra{0}_pe^{-\hat{q}/2}B\ket{-p}_p)\gamma_a\qty(\bra{-p}_pAe^{-\hat{q}/2}\ket{0}_p))\\
  &= \Psi_a\qty(\gamma_a\qty(\bra{0}_pe^{-\hat{q}/2}BAe^{-\hat{q}/2}\ket{0}_p ))\\
  &= \Tr_a(BA),
\end{align}
where the second line follows from the KMS property~\eqref{Equation: Psi KMS}. So $\Tr_a$ obeys the cyclic property.

In summary, $\Tr_a:\mathcal{A}_{\mathcal{N}_a}^G\to\CC$ satisfies all the properties of a trace: it is a normal faithful weight obeying the cyclic property. One may confirm that $\Tr_a(1)$ is infinite, but there are operators with finite trace (this follows from the simple to confirm fact that there are operators $a$ with finite $\Psi_a(a)$), so $\Tr_a$ is a semifinite trace. Indeed, due to its crossed product structure, $\mathcal{A}_{\mathcal{N}_a}^G$ is a direct sum of Type II${}_\infty$ factors, each corresponding to different orderings of the cuts.

\subsection{Representation on a physical Hilbert space}
\label{Subsection: physical Hilbert space}

In this subsection, we will construct a physical Hilbert space $\mathcal{H}_{\text{phys}}$ which carries a normal representation of the gauge-invariant operators which make up the dynamical cut algebras $\mathcal{A}_{\mathcal{N}_a}^G$. We will also give the density operators of states in $\mathcal{H}_{\text{phys}}$ in these algebras. It should be noted that the generalised second law that we will prove in Section~\ref{Section: GSL} is independent of the particular Hilbert space representation one employs. All that is needed is that the representation is normal; then states in this representation admit density operators with von Neumann entropies, to which the GSL then applies. This is useful, because it means that the simplifying assumptions involved in constructing the physical Hilbert space below should not affect the status of the GSL. Nevertheless, we include this subsection for completeness.

The overall gauge constraints are
\begin{equation}
  \mathcal{C} = H+\hat{q}+\mathcal{Q} + \sum_a \hat{d}_a, \qquad
  \mathcal{D} = P - \Theta + \sum_a\hat{k}_a.
\end{equation}
Let us recap the various components of the system we are considering. There are the quantum fields, with Hilbert space $\mathcal{H}_{\text{QFT}}$, on which $H,P$ act. There is also the asymptotic future area of the horizon, whose Hilbert space is $\mathcal{H}_C=L^2(\RR)$, with position operator $\hat{q}$. Additionally, there are each of the dynamical cuts $\mathcal{S}_a$, whose Hilbert spaces are $\mathcal{H}_a=L^2(\RR)$, with position and momentum operators $\hat{v}_a,\hat{k}_a$, and we have defined $\hat{d}_a=\pi\pb*{\hat{v}_a}{\hat{k}_a}$.

In the constraints $\mathcal{C}$ and $\mathcal{D}$ there are also the operators $\mathcal{Q}$ and $-\Theta$, which are the generators of boosts and null translations in the asymptotic past of the horizon respectively. These operators did not play a role in the construction of the algebras in the previous section, as they do not act inside $\mathcal{N}_a$ for any cut. Now, however, we are considering the Hilbert space of the full system, so we must take the asymptotic past into account. We will model $\mathcal{Q},-\Theta$ as the generators of boosts and null translations in the regular representation of the gauge group $G$, so they act on a Hilbert space $\mathcal{H}_{-\infty}=L^2(G,\dd{\mu_L})$, where $\dd{\mu_L}$ is the left-invariant measure on $G$.

Modelling the past asymptotic degrees of freedom in this way can be understood as making the assumption that there is an ideal quantum reference frame for the full gauge group $G$ available in the asymptotic past (similar to how the boost clock is an ideal quantum reference frame for boosts in the asymptotic future). Alternatively, $\mathcal{H}_{-\infty}$ can be viewed as a Hilbert space of edge modes that would be required to glue the horizon $\mathscr{H}$ into spacetimes with more complicated past asymptotics than what we are considering (as can $\mathcal{H}_C$ for future asymptotics).

In more detail, labelling the element $g\in G$ that maps $v\mapsto e^{2\pi t}v + s$ by $(s,t)$, we have the group eigenstates $\ket{g}=\ket{s,t}\in\mathcal{H}_{-\infty}$. The operator $e^{i\Theta s'}e^{-i\mathcal{Q}t'}$ then acts by left translation by the group element corresponding to $(s',t')$:
\begin{equation}
  e^{i\Theta s'}e^{-i\mathcal{Q}t'} \ket{s,t} = \ket*{t+t',e^{2\pi t'}s+s'}.
\end{equation}
We will disregard any other past asymptotic degrees of freedom (as we did when only using the boost clock in the asymptotic future). Of course, a more realistic physical model would involve a more complicated set of degrees of freedom in the asymptotic past and future.

Thus, the full kinematical Hilbert space is
\begin{equation}
  \mathcal{H}_{\text{kin}} = \mathcal{H}_{\text{QFT}}\otimes\bigotimes_a\mathcal{H}_a\otimes \mathcal{H}_C\otimes\mathcal{H}_{-\infty}.
\end{equation}
Let us now address how to construct the physical Hilbert space from $\mathcal{H}_{\text{kin}}$. The aim is to eliminate redundant, gauge-equivalent states from $\mathcal{H}_{\text{kin}}$. The gauge group $G$ is non-compact, but locally compact, so we can construct the physical Hilbert space via group-averaging (also sometimes called the method of coinvariants~\cite{Chandrasekaran_2023b}). One thing we have to be careful about is the fact that $G$ is not \emph{unimodular}; indeed the left- and right-invariant measures are respectively (up to constant factors)
\begin{equation}
  \dd{\mu_L(s,t)} = e^{-2\pi t}\dd{s}\dd{t}, \qquad \dd{\mu_R(s,t)} = \dd{s}\dd{t}.
\end{equation}
The factor $\Delta(s,t) = e^{2\pi t}$ by which these differ is the `modular function'\footnote{As far as the author is aware, this use of `modular' is quite distinct from that associated with Tomita-Takesaki theory, such as in `modular flow' (although there is no doubt an etymological connection). `Modular' is one of those unfortunate pieces of mathematical vocabulary with far too many disparate uses.} of the group; the fact that $\Delta(s,t)\ne 1$ is what it means for this group to be non-unimodular. Let us define a new measure $\mu$ on $G$ by taking the geometric mean of $\mu_L$ and $\mu_R$:
\begin{align}
  \dd{\mu(s,t)} &= \Delta(s,t)^{1/2}\dd{\mu_L(s,t)} = \Delta(s,t)^{-1/2}\dd{\mu_R(s,t)} \\
  &= e^{-\pi t}\dd{s}\dd{t}.
\end{align}
The construction of the physical Hilbert space $\mathcal{H}_{\text{phys}}$ then proceeds by defining a new inner product on kinematical states, which involves using this measure $\mu$ to average over the overall unitary representation $U(s,t)=e^{-is\mathcal{D}}e^{-it\mathcal{C}}$ of $G$ on $\mathcal{H}_{\text{kin}}$:
\begin{align}
  \Braket{\phi}{\phi'} &= \int_G \dd{\mu(s,t)} \bra{\phi}U(s,t) \ket{\phi'} \\
  &= \int_{-\infty}^\infty \dd{s} \int_{-\infty}^\infty \dd{t} \bra{\phi}e^{-is\mathcal{D}}e^{-it(\mathcal{C}-i\pi)}\ket{\phi'}.
  \label{Equation: physical inner product}
\end{align}
The choice of measure ensures that this inner product is Hermitian (in fact this is the unique choice that works~\cite{Giulini:1998kf}).

Kinematical states of the form
\begin{equation}
  \qty(U(s,t)\Delta(s,t)^{-1/2}-\mathds{1})\ket{\phi} = \qty(e^{-is\mathcal{D}}e^{-it(\mathcal{C}-i\pi)}-\mathds{1})\ket{\phi}
\end{equation}
are null in this inner product. In fact, physical states, which we denote $\Ket{\phi}\in\mathcal{H}_{\text{phys}}$, may be concretely constructed as equivalence classes of kinematical states, modulo these null states (the full physical Hilbert space is the completion of the space of all such equivalence classes in the norm given by the inner product above). These physical states are gauge-invariant in the sense that the equivalence classes do not change when we act with $U(s,t)\Delta(s,t)^{-1/2}$. Morally speaking, they satisfy $U(s,t) = \Delta(s,t)^{1/2}$, or in terms of the generators:
\begin{equation}
  \mathcal{C} - i\pi = 0,\qquad \mathcal{D} = 0.
  \label{Equation: non unimodular constraints}
\end{equation}
In fact, we may (loosely speaking) perform the integral in~\eqref{Equation: physical inner product} to obtain
\begin{equation}
  \Braket{\phi}{\phi'} = 4\pi^2 \bra{\phi}\delta(\mathcal{D})\delta(\mathcal{C}-i\pi)\ket{\phi'},
\end{equation}
so~\eqref{Equation: non unimodular constraints} are implemented via delta functions in the physical inner product.
The first equation in~\eqref{Equation: non unimodular constraints} may look odd, but due to~\cite{Giulini:1998kf} the $i\pi$ term is a necessary consequence of $G$ being non-unimodular (see~\cite{Duval:205249,Duval:1991jn,trunk1999algebraicconstraintquantizationpseudorigid} as well, and~\cite{Marnelius_1995} where this particular gauge group $G$ was also considered). It may be viewed as a kind of anomaly, via which physical states are charged under the transformation generated by $\mathcal{C}$. In this sense, it might be related to the anomaly observed in~\cite{Ciambelli:2024swv} (although that anomaly followed from a rather different line of reasoning).

Let us write the map from a given kinematical state $\ket{\phi}$ to the corresponding physical state $\Ket{\phi}$ (concretely, the equivalence class containing $\ket{\phi}$) as $\zeta:\ket{\phi}\mapsto\Ket{\phi}$. Suppose $a$ is a kinematical operator (so $a\in\mathcal{B}(\mathcal{H}_{\text{kin}})$) which is gauge-invariant, i.e.\ $U(s,t)aU(s,t)^\dagger=a$, or equivalently
\begin{equation}
  [a,\mathcal{C}]=[a,\mathcal{D}]=0.
\end{equation}
Then $a$ gives rise to an operator $r(a)$ acting in $\mathcal{H}_{\text{phys}}$, defined via
\begin{equation}
  r(a) \zeta\ket{\phi} = \zeta a \ket{\phi}.
\end{equation}
In this way, gauge-invariant kinematical operators are represented as physical operators:
\begin{equation}
  r: \mathcal{B}(\mathcal{H}_{\text{kin}})^G\to \mathcal{B}(\mathcal{H}_{\text{phys}}).
\end{equation}

We will show momentarily that the restriction of $r$ to $\mathcal{A}_{\mathcal{N}_a}^G$ is normal and faithful. Therefore, a trace on $r(\mathcal{A}_{\mathcal{N}_a}^G)$ may be defined via $\Tr_a^{\text{phys}} = \Tr_a \circ r^{-1}$, and the corresponding density operator $\rho$ of any physical state $\Ket{\phi}$ exists and is uniquely given by
\begin{equation}
  \Bra{\phi}A\Ket{\phi} = \Tr_a^{\text{phys}}(\rho A) \qq{for all} A\in r(\mathcal{A}_{\mathcal{N}_a}^G).
  \label{Equation: density operator definition}
\end{equation}
Moreover, the von Neumann entropy of the state may be written
\begin{equation}
  S = -\Mel{\phi}{\log\rho}{\phi} = -\Tr_a^{\text{phys}}(\rho\log \rho) = -\Tr_a\qty(r^{-1}(\rho)\log r^{-1}(\rho)),
\end{equation}
i.e.\ it is the von Neumann entropy of $r^{-1}(\rho)$. For this reason, we may work entirely at the level of the gauge-invariant kinematical algebras $\mathcal{A}_{\mathcal{N}_a}^G$, rather than their physical representations $r(\mathcal{A}_{\mathcal{N}_a}^G)$, when dealing with entropies (as in the derivation of the GSL in Section~\ref{Section: GSL}).

Let us now explain why the representation is normal and faithful. These are implied by the existence of the past asymptotic reference frame, via a generalisation of an argument given in~\cite{DEHKlong}, which briefly goes as follows. A given physical state $\Ket{\phi}=\zeta\ket{\phi}$ may always be written in the form
\begin{equation}
  \Ket{\phi} = \zeta\qty\big(\ket*{\phi_{|-\infty}}\otimes\ket{0,0}),
  \label{Equation: -infty reduction}
\end{equation}
for some unique $\ket*{\phi_{|-\infty}}\in \mathcal{H}_{\text{QFT}}\otimes\bigotimes_a\mathcal{H}_a\otimes \mathcal{H}_C$ (in the perspective-neutral formalism, this is the state `in the perspective of' the past asymptotic frame), where $\ket{0,0}\in\mathcal{H}_{-\infty}$ is the group eigenstate at the identity in $G$. It may be shown that the map $\mathcal{R}_{-\infty}:\Ket{\phi} \mapsto \ket*{\phi_{|-\infty}}$ defined by~\eqref{Equation: -infty reduction} is unitary, and furthermore that it implements a spatial isomorphism $r(a) = \mathcal{R}_{-\infty}^\dagger a \mathcal{R}_{-\infty}$ between $r(\mathcal{A}_{\mathcal{N}_a}^G)$ and $\mathcal{A}_{\mathcal{N}_a}^G$. This automatically implies that the representation is normal and faithful, as required.

It should also be noted that normalness and faithfulness will hold for a more general class of past asymptotic frames than the simple $L^2(G)$ employed here (including in particular non-ideal ones). The key object is the map $\mathcal{R}_{-\infty}$ (known in the perspective-neutral formalism as a Page-Wootters reduction map), which essentially uses the past asymptotic frame to fix the gauge. A generalisation of this map still exists if the frame is `complete', meaning its state can be used to completely fix the gauge (complete frames transform under faithful representations of $G$) -- but for non-ideal frames it is typically an isometry rather than a full unitary map. It still satisfies $r(a) = \mathcal{R}_{-\infty}^\dagger a \mathcal{R}_{-\infty}$ for $a\in\mathcal{A}_{\mathcal{N}_a}^G$, so the representation is normal. On the other hand, since $\mathcal{R}_{-\infty}$ is only an isometry, this is not automatically a spatial isomorphism, and we need to show faithfulness another way. For this we note that $r(\theta(\hat{v}_a-\hat{v}_b))\ne 0$, since there are clearly states in $\mathcal{H}_{\text{QFT}}\otimes \bigotimes_a\mathcal{H}_a\otimes\mathcal{H}_C$ for which $\mathcal{S}_a$ is located to the future of $\mathcal{S}_b$. By Lemma~\ref{Lemma: faithful normal rep}, we may then conclude that the representation of $\mathcal{A}_{N_a}^G$ on the physical Hilbert space is faithful.

\subsubsection{Density operator of a state in the physical Hilbert space}

We shall now write down explicitly the density operator of a physical state $\Ket{\phi}$. Let us first note that we can (similarly to~\eqref{Equation: -infty reduction}) always put such a state in the form (up to reordering of tensor factors):
\begin{equation}
  \Ket{\phi} = \zeta\qty\big(\ket*{\phi_{|C\mathcal{S}_a}}\otimes \ket{0}_p \otimes \ket{0}_a),
\end{equation}
where
\begin{equation}
  \ket*{\phi_{|C\mathcal{S}_a}} = \bra{0}_p\otimes\bra{0}_a\int_{-\infty}^\infty\dd{s}\int_{-\infty}^\infty\dd{t}e^{-is\mathcal{D}}e^{-it(\mathcal{C}-i\pi)}\ket{\phi}.
  \label{Equation: reduced state}
\end{equation}
In the perspective-neutral approach to QRFs~\cite{delaHamette:2021oex,Hoehn:2023ehz,Hoehn:2019fsy,Hoehn:2020epv,AliAhmad:2021adn,Giacomini:2021gei,Castro-Ruiz:2019nnl,Vanrietvelde:2018dit,Vanrietvelde:2018pgb,Hoehn:2021flk,Suleymanov:2023wio,DeVuyst:2024pop,DEHKlong}, this is the state `in the perspective of' the boost clock $C$ and the cut $\mathcal{S}_a$,\footnote{Note that the previous literature on the perspective-neutral framework has been restricted to unimodular gauge groups. Here we are generalising to the particular non-unimodular group $G$.} and
\begin{equation}
  \mathcal{R}_{C\mathcal{S}_a}: \mathcal{H}_{\text{phys}} \to \mathcal{H}_{\text{QFT}} \otimes\mathcal{H}_{-\infty}\otimes\bigotimes_{b\ne a} \mathcal{H}_b
\end{equation}
is the so-called Page-Wootters reduction map from $\Ket{\phi}$ to $\ket*{\phi_{|C\mathcal{S}_a}}$ (defined by~\eqref{Equation: reduced state}). One may show that $\mathcal{R}_{C\mathcal{S}_a}$ is a unitary operator. Let us further define
\begin{equation}
  \ket*{\phi_a(\{v_b\})} = \qty\Big(\bigotimes_{b\ne a} \bra{v_b}_b)\ket*{\phi_{|C\mathcal{S}_a}} \in \mathcal{H}_{\text{QFT}}\otimes\mathcal{H}_{-\infty},
\end{equation}
where $\{v_b\}$ denotes the collection of variables $v_b$, $b\ne a$.

Next, since $\ket{\Omega}$ is cyclic and separating for $\mathcal{A}_\varphi(0)$, we can define the relative Tomita operator $S_{\chi|\Omega}$ of this algebra to any other state $\ket{\chi}$ it acts on, via
\begin{equation}
  S_{\chi|\Omega} a \ket{\Omega} = a^\dagger \ket{\chi} \qq{for all} a \in\mathcal{A}_\varphi(0).
\end{equation}
More information on such operators may be found in the excellent review~\cite{Witten_2018}.

The density operator of $\Ket{\phi}$ is then given by
\begin{equation}
  \rho = r\Big[e^{i\hat{v}_a(P+\sum_{b\ne a}\hat{k}_b)}\tilde\rho e^{-i\hat{v}_a(P+\sum_{b\ne a}\hat{k}_b)}\Big],
  \label{Equation: density operator formula}
\end{equation}
where
\begin{multline}
  \tilde\rho = \int_{\RR^{2N}} \dd{p}\dd{p'}\qty\Big(\prod_{b\ne a} \dd{v_b}\dd{v_b'})
  e^{H/2} \qty(S_{e^{-ip'(H+\mathcal{Q})}\phi_a(\{e^{-2\pi p'}v_b\})|\Omega})^\dagger S_{e^{-ip (H+\mathcal{Q})}\phi_a(\{e^{-2\pi p}v_b'\})|\Omega} e^{H/2} \\
  \otimes \bigotimes_{b\ne a}\qty\Big(\theta(v_b)e^{\hat{d}_b/2}\ket{v_b}_b\bra*{v_b'}_be^{\hat{d}_b/2}\theta(v_b') + \theta(-v_b)\delta(v_b-v_b')\mathds{1}_b^{<0})
  \otimes e^{-\pi(N-1)(p+p')}e^{\hat{q}/2}\ket{p'}_p\bra{p}_p e^{\hat{q}/2}  ,
  \label{Equation: tilde rho}
\end{multline}
with $N$ the total number of dynamical cuts. This formula may be explicitly derived using similar techniques to those explained in~\cite{DEHKlong}. Here, let us just confirm that it is correct.

First, we note that $\tilde\rho$ is boost-invariant. Indeed, using $e^{-i(H+\mathcal{Q})t}S_{\chi|\Omega}e^{iHt} = S_{e^{-i(H+\mathcal{Q})t}\chi|\Omega}$ for $\ket{\chi}\in\mathcal{H}_{\text{QFT}}\otimes\mathcal{H}_{-\infty}$, we have
\begin{multline}
  e^{-it\mathcal{C}}\tilde\rho e^{it\mathcal{C}}\\
  = \int_{\RR^{2N}} \dd{p}\dd{p'}\qty\Big(\prod_{b\ne a} \dd{v_b}\dd{v_b'})e^{H/2} \qty(S_{e^{-i(p'+t)(H+\mathcal{Q})}\phi_a(\{e^{-2\pi p'}v_b\})|\Omega})^\dagger S_{e^{-i(p+t) (H+\mathcal{Q})}\phi_a(\{e^{-2\pi p}v_b'\})|\Omega} e^{H/2} \\
  \otimes \bigotimes_{b\ne a}\qty\Big(\theta(v_b)e^{\hat{d}_b/2}e^{\pi t}\ket*{e^{2\pi t}v_b}_b\bra*{e^{2\pi t}v_b'}_be^{\pi t}e^{\hat{d}_b/2}\theta(v_b') + \theta(-v_b)\delta(v_b-v_b')\mathds{1}_b^{<0})\\
  \otimes e^{-\pi(N-1)(p+p')}e^{\hat{q}/2}\ket{p'+t}_p\bra{p+t}_p e^{\hat{q}/2} ;
  \label{Equation: tilde rho boosted}
\end{multline}
after a change of variables $p\to p - t$, $p'\to p' - t$, $v_b\to e^{-2\pi t}v_b$, one finds that the right-hand side of~\eqref{Equation: tilde rho boosted} reduces to~\eqref{Equation: tilde rho}, so $e^{-it\mathcal{C}}\tilde\rho e^{it\mathcal{C}} = \tilde\rho$ as claimed. Therefore, $e^{i\hat{v}_a(P+\sum_{b\ne a}\hat{k}_b)}\tilde\rho e^{-i\hat{v}_a(P+\sum_{b\ne a}\hat{k}_b)}$ is invariant under the full gauge group, which means~\eqref{Equation: density operator formula} is well-defined.

Next, we note that
\begin{equation}
  \tilde\rho \in \mathcal{A}_{\text{QFT}}(0) \otimes \bigotimes_{b\ne a} \mathcal{A}_{\mathcal{S}_b}(0) \otimes \mathcal{B}(\mathcal{H}_C).
  \label{Equation: tilde rho algebra}
\end{equation}
For the parts of the integrand in~\eqref{Equation: tilde rho} which act on the cuts and $\mathcal{H}_C$, this is clear. For the QFT part, one has that, for any states $\ket{\chi},\ket{\chi'}$ in a Hilbert space acted on by $\mathcal{A}_{\text{QFT}}(0)$, the combination $S_\Omega^\dagger S_{\chi'|\Omega}^\dagger S_{\chi|\Omega}S_\Omega$ (where $S_\Omega$ is the Tomita operator of $\ket{\Omega}$) is in the commutant of $\mathcal{A}_{\text{QFT}}(0)$. This follows from the observation that
\begin{equation}
  \bra{\Omega} ab S_\Omega^\dagger S_{\chi'|\Omega}^\dagger S_{\chi|\Omega}S_\Omega c\ket{\Omega} = \bra{\chi'}abc\ket{\chi} = \bra{\Omega} a S_\Omega^\dagger S_{\chi'|\Omega}^\dagger S_{\chi|\Omega}S_\Omega bc\ket{\Omega}
\end{equation}
for all $a,b,c\in\mathcal{A}_{\text{QFT}}(0)$, and noting that $\ket{\Omega}$ is cyclic. Then, using the fact that modular conjugation $a\to J_\Omega a J_\Omega$ exchanges $\mathcal{A}_{\text{QFT}}(0)$ and its commutant, and that $H$ is the modular Hamiltonian of $\ket{\Omega}$, we have that
\begin{equation}
  e^{H/2} S_{\chi'|\Omega}^\dagger S_{\chi|\Omega} e^{H/2} = J_\Omega S_\Omega^\dagger S_{\chi'|\Omega}^\dagger S_{\chi|\Omega}S_\Omega  J_\Omega
\end{equation}
is affiliated with $\mathcal{A}_{\text{QFT}}(0)$. It follows that~\eqref{Equation: tilde rho algebra} holds, and hence that $\rho$ is an element of $r(\mathcal{A}_{\mathcal{N}_a}^G)$.

It is straightforward to then check that $\rho$ reproduces the appropriate expectation values of $\Ket{\phi}$. Indeed, given an $A=r(a)\in r(\mathcal{A}_{\mathcal{N}_a}^G)$, we have
\begin{align}
  \Tr^{\text{phys}}_a(\rho A) &= \Psi_a(\bra{0}_p e^{-\hat{q}/2}\gamma_a(r^{-1}(\rho)a)e^{-\hat{q}/2}\ket{0}_p) \\
  &= \Psi_a(\bra{0}_pe^{-\hat{q}/2}\tilde\rho \gamma_a(a) e^{-\hat{q}/2}\ket{0}_p)\\
  &
  \begin{multlined}[t]
    =\Psi_a\Big(\int_{\RR^{2N-1}} \dd{p}\qty\Big(\prod_{b\ne a} \dd{v_b}\dd{v_b'}) e^{H/2} \qty(S_{\phi_a(\{v_b\})|\Omega})^\dagger S_{e^{-ip(H+\mathcal{Q})}\phi_a(\{e^{-2\pi p}v_b'\})|\Omega}  e^{H/2} \\
      \otimes \bigotimes_{b\ne a}\qty\Big(\theta(v_b)e^{\hat{d}_b/2}\ket{v_b}_b\bra*{v_b'}_be^{\hat{d}_b/2}\theta(v_b') + \theta(-v_b)\delta(v_b-v_b')\mathds{1}_b^{<0})\\
    e^{-\pi(N-1)p}\bra{p}_p e^{\hat{q}/2}\gamma_a(a)e^{-\hat{q}/2}\ket{0}_p\Big)
  \end{multlined}\\
  &
  \begin{multlined}[t]
    =\bra{\Omega}\int_{\RR^{2N-1}} \dd{p}\qty\Big(\prod_{b\ne a} \dd{v_b}\dd{v_b'}) \qty(S_{\phi_a(\{v_b\})|\Omega})^\dagger S_{e^{-ip(H+\mathcal{Q})}\phi_a(\{e^{-2\pi p}v_b'\})|\Omega}\\
    \bigotimes_{b\ne a}\bra{v_b'}
    e^{-\pi(N-1)p}e^{(H+\sum_{b\ne a}\hat{d}_b)/2}\bra{p}_p e^{\hat{q}/2}\gamma_a(a)e^{-\hat{q}/2}\ket{0}_p \\
    e^{-(H+\sum_{b\ne a}\hat{d}_b)/2}\bigotimes_{b\ne a}\ket{v_b}\ket{\Omega}
  \end{multlined}\\
  &
  \begin{multlined}[t]
    =\int_{\RR^{2N-1}} \dd{p}\qty\Big(\prod_{b\ne a} \dd{v_b}\dd{v_b'}) \bra*{\phi_a(\{e^{-2\pi p}v_b'\})}e^{ip(H+\mathcal{Q})}\otimes
    \bigotimes_{b\ne a}\bra{v_b'}
    e^{-\pi(N-1)p}\\
    \bra{p}_p \gamma_a(a)\ket{0}_p
    \ket*{\phi_a(\{v_b\})}\otimes \bigotimes_{b\ne a}\ket{v_b}
  \end{multlined}\\
  &=\int_{\RR^{N}} \dd{p}\qty\Big(\prod_{b\ne a} \dd{v_b'})  \bra*{\phi_a(\{v_b'\})}\otimes
  \bigotimes_{b\ne a}\bra{v_b'}
  e^{ip(H+\mathcal{Q}+\sum_{b\ne a}\hat{d}_b)}\bra{p}_p \gamma_a(a)\ket{0}_p\ket*{\phi_{|C\mathcal{S}_a}}\\
  &=\int_{-\infty}^\infty \dd{p}\bra*{\phi_{|C\mathcal{S}_a}}\bra{0}_p
  e^{ip(H+\hat{q}+\mathcal{Q}+\sum_{b\ne a}\hat{d}_b)} \gamma_a(a)\ket{0}_p\ket*{\phi_{|C\mathcal{S}_a}},
  \label{Equation: density operator expectation}
\end{align}
where we have used that $\gamma_a(a)$ commutes with $H+\hat{q}+\sum_{b\ne a}\hat{d}_b$, and in the penultimate equality changed variables $v_b'\to e^{2\pi p}v_b'$. On the other hand,
\begin{align}
  \Bra{\phi} A \Ket{\phi} &= \int_{-\infty}^\infty \dd{s} \int_{-\infty}^\infty \dd{t} \bra*{\phi_{|C\mathcal{S}_a}}\otimes\bra{0}_p\otimes\bra{0}_a e^{-is\mathcal{D}}e^{-it(\mathcal{C}-i\pi)} a \ket*{\phi_{|C\mathcal{S}_a}}\otimes\ket{0}_p\otimes\ket{0}_a\\
  &= \int_{-\infty}^\infty \dd{s} \int_{-\infty}^\infty \dd{t} \bra*{\phi_{|C\mathcal{S}_a}}\otimes\bra{0}_p e^{-is(P-\Theta+\sum_{b\ne a}\hat{k}_b)}\\
  &\hspace*{10em} e^{-it(H+\hat{q}+\mathcal{Q}+\sum_{b\ne a}\hat{d}_b)} \gamma_a(a) \ket*{\phi_{|C\mathcal{S}_a}}\otimes\ket{0}_p \bra{-s}\ket{0}_a\nonumber\\
  &= \int_{-\infty}^\infty \dd{t} \bra*{\phi_{|C\mathcal{S}_a}}\otimes\bra{0}_p e^{-it(H+\hat{q}+\mathcal{Q}+\sum_{b\ne a}\hat{d}_b)} \gamma_a(a) \ket*{\phi_{|C\mathcal{S}_a}}\otimes\ket{0}_p,
  \label{Equation: phi expectation value}
\end{align}
where the second equality follows from $e^{-it\hat{d}_a}\ket{0}_a = e^{\pi t}\ket{0}_a$. Thus, comparing~\eqref{Equation: phi expectation value} with~\eqref{Equation: density operator expectation}, we see that~\eqref{Equation: density operator definition} holds. Therefore, $\rho$ is the density operator of $\Ket{\phi}$, as claimed.

\section{The generalised second law for dynamical cuts}
\label{Section: GSL}

We now prove the GSL for dynamical cuts. The argument proceeds as follows. First, we describe how the GSL includes an implicit conditioning on the ordering of the cuts, and explain how we must explicitly do this conditioning when considering dynamical cuts. Under this conditioning, the algebras associated with dynamical cuts constructed in Section~\ref{Section: dynamical cuts} obey isotony, unlike those associated with fixed cuts constructed in Section~\ref{Section: fixed cuts}. This allows us to apply the monotonicity of relative entropy. In particular, a certain application of that more general inequality lands us at the GSL:
\begin{equation}
  S_{\mathcal{N}_a}[\Phi] \ge
  S_{\mathcal{N}_b}[\Phi] + F(\mathcal{S}_a\in\mathcal{N}_b).
  \label{Equation: dynamical GSL}
\end{equation}
Here, we are assuming that the cuts are ordered such that $\mathcal{N}_a\subset\mathcal{N}_b$. As explained in the Introduction, $S_{\mathcal{N}_a}[\Phi]$, $S_{\mathcal{N}_b}[\Phi]$ are the von Neumann entropies for the algebras of the regions $\mathcal{N}_a,\mathcal{N}_b$ respectively, in the state $\Phi$, and $F(\mathcal{S}_a\in\mathcal{N}_b)$ is the free energy of the later cut $\mathcal{S}_a$ (considered as a degree of freedom contained within $\mathcal{N}_b$).

\subsection{Entropy inequalities from conditioning on an ordering}

The second law is a statement that is conditioned on the ordering of the cuts. If $S[\mathcal{N}_a]$ is an entropy associated with the region to the exterior of the cut $\mathcal{S}_a$, then the second law should in principle be some statement of the form (ignoring here for conciseness the cut free energy)
\begin{equation}
  \text{``} \,v_a\ge v_b \implies S[\mathcal{N}_b] \ge S[\mathcal{N}_b],\quad v_a \le v_b \implies S[\mathcal{N}_a] \le S[\mathcal{N}_b].\,\text{''}
\end{equation}
With this in mind, we need to condition on an ordering as part of our formulation of the second law; if we don't do this, then there is no fixed sense in which a given dynamical cut is to the future of another, and so no reason to expect a second law to hold.

Let us write $\mathcal{S}_a>\mathcal{S}_b$ to denote $\mathcal{S}_a$ being to the future of $\mathcal{S}_b$. The operator which conditions on $\mathcal{S}_a>\mathcal{S}_b$ is $\theta(\hat{v}_a-\hat{v}_b)$, which is a central element of both $\mathcal{A}_{\mathcal{N}_a}^G$ and $\mathcal{A}_{\mathcal{N}_b}^G$. This means we can condition on $\mathcal{S}_a>\mathcal{S}_b$ by restricting to states whose density operators $\rho_a\in\mathcal{A}_{\mathcal{N}_a}^G$, $\rho_b\in\mathcal{A}_{\mathcal{N}_b}^G$ satisfy
\begin{equation}
  \rho_a = \rho_a\theta(\hat{v}_a-\hat{v}_b), \qquad \rho_b=\rho_b\theta(\hat{v}_a-\hat{v}_b).
  \label{Equation: conditioned density operators}
\end{equation}
Let us define traces on the subalgebras
\begin{equation}
  \mathcal{A}_{\mathcal{N}_a,\mathcal{S}_a>\mathcal{S}_b}^G=\theta(\hat{v}_a-\hat{v}_b)\mathcal{A}_{\mathcal{N}_a}^G\subset \mathcal{A}_{\mathcal{N}_a}^G,
  \qquad
  \mathcal{A}_{\mathcal{N}_b,\mathcal{S}_a>\mathcal{S}_b}^G=\theta(\hat{v}_a-\hat{v}_b)\mathcal{A}_{\mathcal{N}_{\changed[a]{b}}}^G\subset \mathcal{A}_{\mathcal{N}_{\changed[a]{b}}}^G,
\end{equation}
as the restrictions of $\Tr_a$, $\Tr_b$ on $\mathcal{A}_{\mathcal{N}_a}^G$, $\mathcal{A}_{\mathcal{N}_b}^G$ respectively. Although $\rho_a$ and $\rho_b$ were originally defined as the density operators of $\mathcal{A}_{\mathcal{N}_a}^G$ and $\mathcal{A}_{\mathcal{N}_b}^G$, one may then straightforwardly show that $\rho_a$ and $\rho_b$ obeying~\eqref{Equation: conditioned density operators} are also the corresponding density operators of the \emph{subalgebras} $\mathcal{A}_{\mathcal{N}_a,\mathcal{S}_a>\mathcal{S}_b}^G$ and $\mathcal{A}_{\mathcal{N}_b,\mathcal{S}_a>\mathcal{S}_b}^G$ respectively. Moreover, the von Neumann entropies for $\mathcal{A}_{\mathcal{N}_a}^G$, $\mathcal{A}_{\mathcal{N}_b}^G$ agree with the von Neumann entropies for the subalgebras $\mathcal{A}_{\mathcal{N}_a,\mathcal{S}_a>\mathcal{S}_b}^G$, $\mathcal{A}_{\mathcal{N}_b,\mathcal{S}_a>\mathcal{S}_b}^G$ respectively. This fact is one half of what allows us to derive a GSL.

The other half is that these subalgebras obey \emph{isotony}:
\begin{equation}
  \mathcal{A}_{\mathcal{N}_a,\mathcal{S}_a>\mathcal{S}_b}^G \subseteq \mathcal{A}_{\mathcal{N}_b,\mathcal{S}_a>\mathcal{S}_b}^G.
  \label{Equation: isotony}
\end{equation}
It is to some extent intuitively obvious that this should hold. These two algebras contain operators acting on the degrees of freedom in $\mathcal{N}_a$ and $\mathcal{N}_b$ respectively, conditioned on $\mathcal{S}_a>\mathcal{S}_b$. Under this ordering one has $\mathcal{N}_a\subset\mathcal{N}_b$, so clearly everything which can be observed in $\mathcal{N}_a$ ought to also be capable of being observed in $\mathcal{N}_b$. We will explicitly confirm this below. Due to isotony~\eqref{Equation: isotony}, we can apply the monotonicity of relative entropy~\cite{Araki:1976zv}.\footnote{The relevant result in~\cite{Araki:1976zv} is that relative entropy is monotonic for \emph{hyperfinite} algebras. QFT algebras are hyperfinite~\cite{Buchholz:1986bg}, which implies that all the algebras considered in this paper are hyperfinite too.} In particular, for any two weights $\Phi$ and $\Phi'$, the relative entropy for $\mathcal{A}_{\mathcal{N}_b,\mathcal{S}_a>\mathcal{S}_b}^G$ is bounded below by the relative entropy for its subalgebra $\mathcal{A}_{\mathcal{N}_a,\mathcal{S}_a>\mathcal{S}_b}^G$:
\begin{equation}
  S_{\text{rel}}^{\mathcal{A}_{\mathcal{N}_a,\mathcal{S}_a>\mathcal{S}_b}^G}(\Phi||\Phi')
  \le
  S_{\text{rel}}^{\mathcal{A}_{\mathcal{N}_b,\mathcal{S}_a>\mathcal{S}_b}^G}(\Phi||\Phi').
  \label{Equation: monotonicity of relative entropy}
\end{equation}
We will derive the GSL from this inequality.

It suffices to condition on the ordering of just two cuts, $\mathcal{S}_a>\mathcal{S}_b$, in order to use~\eqref{Equation: monotonicity of relative entropy}. However, in the following we will further condition on a complete ordering of all the cuts. Without loss of generality, we can label the cuts according to the ordering, so that
\begin{equation}
  \mathcal{S}_N > \mathcal{S}_{N-1}>\dots>\mathcal{S}_2>\mathcal{S}_1,
  \label{Equation: complete ordering}
\end{equation}
where $N$ is the total number of cuts. Moreover, we will focus on the case where $\mathcal{S}_a=\mathcal{S}_{b+1}$, i.e.\ where $\mathcal{S}_a$ is the cut just after $\mathcal{S}_b$. The reason for these further conditions is that they make the analysis more straightforward, while still allowing one to derive a GSL relating the entropies of subsequent cuts.\footnote{It is quite likely that our methods can be generalised to the case where these extra conditions do not hold, obtaining stronger entropy inequalities than the GSL that we present in this work. But a complete quantum information theoretic interpretation of these inequalities would be more complicated. Indeed, we suspect that relaxing the conditioning on a complete ordering~\eqref{Equation: complete ordering} may lead to subtleties related to indefinite causal order, while relaxing $\mathcal{S}_a=\mathcal{S}_{b+1}$ would require one to account for entanglement among the cuts in between $\mathcal{S}_a$ and $\mathcal{S}_b$. We shall leave further exploration of these interesting possibilities to future work.}

\subsubsection{Isotony}

Let us now explicitly confirm that isotony~\eqref{Equation: isotony} holds. One way to do so is to convince oneself that, before imposing gauge-invariance, the kinematical algebras obey isotony (in the presence of the conditioning $\mathcal{S}_a>\mathcal{S}_b$), which directly implies that the gauge-invariant algebras do too. Here, let us work directly with the gauge-invariant algebras. Using~\eqref{Equation: invariant algebra}, one has
\begin{align}
  \MoveEqLeft e^{-i\hat{v}_b(P+\sum_{c\ne b}\hat{k}_c)}\mathcal{A}^{G}_{\mathcal{N}_b,\mathcal{S}_a>\mathcal{S}_b}e^{i\hat{v}_b(P+\sum_{c\ne b}\hat{k}_c)} \\
  &= \theta(\hat{v}_a)e^{-i\hat{v}_b(P+\sum_{c\ne b}\hat{k}_c)}\mathcal{A}^{G}_{\mathcal{N}_b}e^{i\hat{v}_b(P+\sum_{c\ne b}\hat{k}_c)} \\
  &=e^{i\hat{p}(H+\sum_{c\ne b}\hat{d}_c)}\qty\Big(\mathcal{A}_{\text{QFT}}(0)\otimes\mathcal{B}(\mathcal{H}_a^{>0})\otimes \bigotimes_{c\ne a,b} \mathcal{A}_{\mathcal{S}_c}(0))e^{-i\hat{p}(H+\sum_{c\ne b}\hat{d}_c)}\vee \{\hat{q}\}'',
  \label{Equation: isotony part one}
\end{align}
and similarly,
\begin{align}
  \MoveEqLeft e^{-i\hat{v}_a(P+\sum_{c\ne a}\hat{k}_c)}\mathcal{A}^{G}_{\mathcal{N}_a,\mathcal{S}_a>\mathcal{S}_b}e^{i\hat{v}_a(P+\sum_{c\ne a}\hat{k}_c)} \\
  &= \theta(-\hat{v}_b)e^{-i\hat{v}_a(P+\sum_{c\ne a}\hat{k}_c)}\mathcal{A}^{G}_{\mathcal{N}_a}e^{i\hat{v}_a(P+\sum_{c\ne a}\hat{k}_c)} \\
  &= \theta(-\hat{v}_b) e^{i\hat{p}(H+\sum_{c\ne a,b}\hat{d}_c)}\qty\Big(\mathcal{A}_{\text{QFT}}(0)\otimes\bigotimes_{c\ne a,b} \mathcal{A}_{\mathcal{S}_c}(0))e^{-i\hat{p}(H+\sum_{c\ne a,b}\hat{d}_c)}\vee \{\hat{q}\}''.
\end{align}
So, noting
\begin{equation}
  e^{-i\hat{v}_b(P+\sum_{c\ne b}\hat{k}_c)}e^{i\hat{v}_a(P+\sum_{c\ne a}\hat{k}_c)} = e^{-i\hat{v}_b\hat{k}_a} e^{i(\hat{v}_a-\hat{v}_b)(P+\sum_{c\ne a,b}\hat{k}_c)}e^{i\hat{v}_a\hat{k}_b},
\end{equation}
we have
\begin{align}
  \MoveEqLeft e^{-i\hat{v}_b(P+\sum_{c\ne b}\hat{k}_c)}\mathcal{A}^{G}_{\mathcal{N}_a,\mathcal{S}_a>\mathcal{S}_b}e^{i\hat{v}_b(P+\sum_{c\ne b}\hat{k}_c)} \\
  &=
  \begin{multlined}[t]
    e^{-i\hat{v}_b\hat{k}_a}\theta(\hat{v}_a-\hat{v}_b)e^{i(\hat{v}_a-\hat{v}_b)(P+\sum_{c\ne a,b}\hat{k}_c)}\\
    \qty\Big(e^{i\hat{p}(H+\sum_{c\ne a,b}\hat{d}_c)}\qty\Big(\mathcal{A}_{\text{QFT}}(0)\otimes\bigotimes_{c\ne a,b} \mathcal{A}_{\mathcal{S}_c}(0))e^{-i\hat{p}(H+\sum_{c\ne a,b}\hat{d}_c)}\vee \{\hat{q}\}'')\\
    e^{-i(\hat{v}_a-\hat{v}_b)(P+\sum_{c\ne a,b}\hat{k}_c)}e^{i\hat{v}_b\hat{k}_a}
  \end{multlined}\\
  &\subseteq e^{-i\hat{v}_b\hat{k}_a}\theta(\hat{v}_a-\hat{v}_b)
  \qty\Big(e^{i\hat{p}(H+\sum_{c\ne a,b}\hat{d}_c)}\qty\Big(\mathcal{A}_{\text{QFT}}(0)\otimes\bigotimes_{c\ne a,b} \mathcal{A}_{\mathcal{S}_c}(0))e^{-i\hat{p}(H+\sum_{c\ne a,b}\hat{d}_c)}\vee \{\hat{q}\}'')e^{i\hat{v}_b\hat{k}_a}\\
  &= e^{i\hat{p}(H+\sum_{c\ne a,b}\hat{d}_c)}\qty\Big(\mathcal{A}_{\text{QFT}}(0)\otimes\theta(\hat{v}_a)\otimes\bigotimes_{c\ne a,b} \mathcal{A}_{\mathcal{S}_c}(0))e^{-i\hat{p}(H+\sum_{c\ne a,b}\hat{d}_c)}\vee \{\hat{q}\}'',
  \label{Equation: isotony part two}
\end{align}
where the $\subseteq$ follows from the fact that $\mathcal{A}_{\text{QFT}}(v)\subseteq\mathcal{A}_{\text{QFT}}(0)$ and $\mathcal{A}_{\mathcal{S}_c}(v)\subset\mathcal{A}_{\mathcal{S}_c}(0)$ for positive $v$, and the conjugation by $e^{i(\hat{v}_a-\hat{v}_b)(P+\sum_{c\ne a,b}\hat{k}_c)}$ (with the presence of $\theta(\hat{v}_a-\hat{v}_b)$) transforms the algebras in this way. Comparing~\eqref{Equation: isotony part two} with~\eqref{Equation: isotony part one}, one finds (since $\theta(\hat{v}_a)\in\mathcal{B}(\mathcal{H}_a^{>0})$)
\begin{equation}
  e^{-i\hat{v}_b(P+\sum_{c\ne b}\hat{k}_c)}\mathcal{A}^{G}_{\mathcal{N}_a,\mathcal{S}_a>\mathcal{S}_b}e^{i\hat{v}_b(P+\sum_{c\ne b}\hat{k}_c)} \subseteq e^{-i\hat{v}_b(P+\sum_{c\ne b}\hat{k}_c)}\mathcal{A}^{G}_{\mathcal{N}_b,\mathcal{S}_a>\mathcal{S}_b}e^{i\hat{v}_b(P+\sum_{c\ne b}\hat{k}_c)},
\end{equation}
or equivalently, the claimed isotony~\eqref{Equation: isotony}.

\subsection{Free energy of the later cut}

As explained in the Introduction, an important role in the GSL will be played by the free energy of the later cut $\mathcal{S}_a$. Let us now give more details on this quantity.

We more precisely consider the free energy of $\mathcal{S}_a$ in the perspective of $\mathcal{S}_b$ and the asymptotic boost clock $C$. This is described by the algebra
\begin{equation}
  \mathcal{A}_{\mathcal{S}_a\in\mathcal{N}_b}^G = e^{i\hat{v}_b\hat{k}_a}e^{i\hat{p}\hat{d}_a}\mathcal{B}(\mathcal{H}_a^{>0}) e^{-i\hat{p}\hat{d}_a} e^{-i\hat{v}_b\hat{k}_a}
\end{equation}
consisting of operators acting on the later cut $\mathcal{S}_a$, dressed to $\mathcal{S}_b$ and $C$. This is a subalgebra of $\mathcal{A}_{\mathcal{N}_b,\mathcal{S}_a>\mathcal{S}_b}^G$. It is clearly a Type I algebra, being unitarily equivalent to the algebra of bounded operators acting on the Hilbert space $\mathcal{H}_a^{>0}$. The trace on this algebra may be simply written
\begin{equation}
  \Tr_{\mathcal{S}_a\in\mathcal{N}_b}\qty(e^{i\hat{v}_b\hat{k}_a}e^{i\hat{p}\hat{d}_a} a e^{-i\hat{p}\hat{d}_a} e^{-i\hat{v}_b\hat{k}_a}) = \tr_{\mathcal{H}_a^{>0}}(a).
\end{equation}
Any state $\Phi$ has a density operator $\hat{\rho}_a\in\mathcal{A}_{\mathcal{S}_a\in\mathcal{N}_b}^G$ via this trace,\footnote{\label{Footnote: explicit free energy state}For a state $\Ket{\phi}$ of the form described in Section~\ref{Subsection: physical Hilbert space}, one may confirm that this algebra is faithfully represented on the Hilbert space, and the density operator in the representation is
  \begin{equation}
    r(\hat\rho_a) = r\qty[e^{i\hat{v}_b\hat{k}_a}e^{i\hat{p}\hat{d}_a} \tr_{\bar{a}}(\ket*{\phi_{|C\mathcal{S}_b}}\bra*{\phi_{|C\mathcal{S}_b}}) e^{-i\hat{p}\hat{d}_a} e^{-i\hat{v}_b\hat{k}_a}],
  \end{equation}
where $\tr_{\bar{a}}$ denotes a partial trace over all other Hilbert space tensor factors than $\mathcal{H}_a$.} and hence a von Neumann entropy
\begin{equation}
  S_{\mathcal{S}_a\in\mathcal{N}_b}[\Phi] := S_{\text{vN}}^{\mathcal{A}^G_{\mathcal{S}_a\in\mathcal{N}_b}}(\Phi) = -\Tr_{\mathcal{S}_a\in\mathcal{N}_b}\qty(\hat\rho_a\log\hat\rho_a) = -\Phi(\log\hat\rho_a).
\end{equation}

Note that $\mathcal{A}_{\mathcal{S}_a\in\mathcal{N}_b}^G$ and $\mathcal{A}_{\mathcal{N}_a,\mathcal{S}_a>\mathcal{S}_b}^G$ together generate $\mathcal{A}_{\mathcal{N}_b,\mathcal{S}_a>\mathcal{S}_b}^G$:
\begin{equation}
  \mathcal{A}_{\mathcal{S}_a\in\mathcal{N}_b}^G\vee\mathcal{A}_{\mathcal{N}_a,\mathcal{S}_a>\mathcal{S}_b}^G = \mathcal{A}_{\mathcal{N}_b,\mathcal{S}_a>\mathcal{S}_b}^G.
\end{equation}
This is because we may use operators in $\mathcal{A}_{\mathcal{S}_a\in\mathcal{N}_b}^G$ to move $\mathcal{S}_a$ to where it overlaps with $\mathcal{S}_b$. At that point, any operator acting in $\mathcal{N}_b$ is also an operator acting in $\mathcal{N}_a$, from which the above follows. One might hope to apply subadditivity of entropy to the above equation in order to derive an inequality like the GSL. Unfortunately, subadditivity of the entropy only straightforwardly applies to Type I algebras, so a more involved argument will be required for the present case.

Let us define the operator
\begin{equation}
  H_{\mathcal{S}_a|\mathcal{S}_b} = \hat{d}_a - 2\pi \hat{v}_b\hat{k}_a\in \mathcal{A}_{\mathcal{S}_a\in\mathcal{N}_b}^G.
\end{equation}
This is the generator of a boost of $\mathcal{S}_a$ around $\mathcal{S}_b$. The free energy that will play a role in the GSL is that of $\mathcal{S}_a\in\mathcal{N}_b$, with respect to $H_{\mathcal{S}_a|\mathcal{S}_b}$:
\begin{equation}
  F(\mathcal{S}_a\in\mathcal{N}_b) = \expval*{H_{\mathcal{S}_a|\mathcal{S}_b}}_\Phi - S_{\mathcal{S}_a\in\mathcal{N}_b}[\Phi].
\end{equation}

\subsubsection{Change in area due to the cut}

It is worth at this stage commenting on how $H_{\mathcal{S}_a|\mathcal{S}_b}$ is related to the effect of $\mathcal{S}_a$ on the geometry of the horizon $\mathscr{H}$. Consider the classical theory. At leading order in perturbation theory, Raychaudhuri's equation on $\mathscr{H}$ reduces to
\begin{equation}
  \dv{\Theta}{v} = - 8\pi G_{\text{N}} t_{vv},
  \label{Equation: Raychaudhuri}
\end{equation}
where $\Theta=\Theta(v)$ is the expansion, and $t_{vv}$ is the ${}_{vv}$ component of the stress tensor of everything on the horizon, including gravitons, matter, and the cuts. It may be decomposed as
\begin{equation}
  t_{vv} = T_{vv} + \sum_{c}\delta(v-v_c) t_c,
\end{equation}
where $T_{\mu\nu}$ is the stress tensor of the gravitons and matter, and $t_c$ is the contribution of cut $\mathcal{S}_c$, satisfying $\int_{\mathscr{H}} \eta\,\delta(v-v_c) t_c=k_c$. Using \eqref{Equation: Raychaudhuri}, and the fact that the expansion vanishes at $v\to\infty$, we may write
\begin{align}
  \frac{A_\infty}{4G_{\text{N}}}-\frac{A(\mathcal{S}_a)}{4G_{\text{N}}} &= 2\pi\int_{\mathscr{H}}\eta \theta(v-v_a) (v-v_a) t_{vv} \\*
  &= 2\pi\qty\Big(\int_{\mathscr{H}}\eta \theta(v-v_a)(v-v_a) T_{vv} + \sum_{c=a+1}^N (v_c-v_a)k_c)
  \label{Equation: gravitational constraint area cut}
\end{align}
where $A({\mathcal{S}_a}),A_\infty$ are the area of $\mathcal{S}_a$ and the asymptotic future cut respectively, and similarly for $\mathcal{S}_b$. Thus,
\begin{equation}
  \frac{A(\mathcal{S}_a)}{4G_{\text{N}}}-\frac{A(\mathcal{S}_b)}{4G_{\text{N}}} = 2\pi (v_a - v_b) k_a + \dots,
\end{equation}
where the $\dots$ contains contributions from the gravitons, matter and momenta $k_c$ of cuts $\mathcal{S}_c$ with $c\ne a$. Upon quantisation to a Hermitian operator, the first term on the right-hand side becomes $H_{\mathcal{S}_a|\mathcal{S}_b}$. Thus, as described in the introduction, we may identify
\begin{equation}
  H_{\mathcal{S}_a|\mathcal{S}_b} = \frac{\Delta_a A}{4G_{\text{N}}},
\end{equation}
where $\Delta_a A$ is the contribution to the change in area between the two cuts due to the energy of $\mathcal{S}_a$.

\subsection{Derivation of the GSL}

We finally now come to the main point of the paper: the derivation of the GSL for dynamical cuts.

Let us write~\eqref{Equation: monotonicity of relative entropy} in terms of density operators. If $\rho_a,\rho_a'$ and $\rho_b,\rho_b'$ are the density operators of $\Phi,\Phi'$ respectively, in the algebras $\mathcal{A}_{\mathcal{N}_a,\mathcal{S}_a>\mathcal{S}_b}^G$ and $\mathcal{A}_{\mathcal{N}_b,\mathcal{S}_a>\mathcal{S}_b}^G$ respectively, then we have
\begin{align}
  S_{\text{rel}}^{\mathcal{A}_{\mathcal{N}_a,\mathcal{S}_a>\mathcal{S}_b}^G}(\Phi||\Phi') &= \Tr_a(\rho_a(\log \rho_a - \log \rho'_a)) = -S_{\text{vN}}^{\mathcal{A}_{\mathcal{N}_a,\mathcal{S}_a>\mathcal{S}_b}^G}(\rho_a) - \expval{\log \rho'_a}_\Phi,\\
  S_{\text{rel}}^{\mathcal{A}_{\mathcal{N}_b,\mathcal{S}_a>\mathcal{S}_b}^G}(\Phi||\Phi') &= \Tr_b(\rho_b(\log \rho_b - \log \rho'_b)) = -S_{\text{vN}}^{\mathcal{A}_{\mathcal{N}_b,\mathcal{S}_a>\mathcal{S}_b}^G}(\rho_b) - \expval{\log \rho'_b}_\Phi.
\end{align}
Now, using
\begin{equation}
  S_{\mathcal{N}_a}[\Phi] := S_{\text{vN}}^{\mathcal{A}_{\mathcal{N}_a^G}}(\rho_a) = S_{\text{vN}}^{\mathcal{A}_{\mathcal{N}_a,\mathcal{S}_a>\mathcal{S}_b}^G}(\rho_a)\qq{and}
  S_{\mathcal{N}_b}[\Phi] := S_{\text{vN}}^{\mathcal{A}_{\mathcal{N}_b^G}}(\rho_b) = S_{\text{vN}}^{\mathcal{A}_{\mathcal{N}_b,\mathcal{S}_a>\mathcal{S}_b}^G}(\rho_b)
\end{equation}
with~\eqref{Equation: monotonicity of relative entropy}, we have
\begin{equation}
  S_{\mathcal{N}_a}[\Phi] + \expval{\log \rho_a'}_\Phi \ge
  S_{\mathcal{N}_b}[\Phi] + \expval{\log \rho_b'}_\Phi.
  \label{Equation: monotonicity Phi'}
\end{equation}
This is close to the kind of inequality we want for the GSL. It holds for any $\Phi'$, so the density operators $\rho_a'$ and $\rho_b'$ are free variables -- but they are not independent of each other. Indeed, they must agree when evaluating expectation values of operators in the subalgebra $\mathcal{A}_{\mathcal{N}_a,\mathcal{S}_a>\mathcal{S}_b}^G$:
\begin{equation}
  \Tr_a(\rho'_a a) = \Tr_b(\rho'_b a) \qq{for all} a\in\mathcal{A}_{\mathcal{N}_a,\mathcal{S}_a>\mathcal{S}_b}^G.
\end{equation}
Note that monotonicity of relative entropy applies for \emph{weights} $\Phi'$ as well as states. Thus, the density operators $\rho'_a,\rho'_b$ need not be normalised or even trace-class.

We will choose a certain useful weight $\Phi'$. In particular, consider $\rho'_b = \Pi_b e^{\hat{q}/2}\sigma e^{\hat{q}/2}$, for non-negative
\begin{equation}
  \sigma = e^{i\hat{v}_b\hat{k}_a}e^{i\hat{p}\hat{d}_a}\tilde\sigma e^{-i\hat{p}\hat{d}_a} e^{-i\hat{v}_b\hat{k}_a} \in \mathcal{A}_{\mathcal{S}_a\in\mathcal{N}_b}^G,
\end{equation}
and
\begin{equation}
  \Pi_b = \prod_{c=1}^{b-1} \theta(\hat{v}_b-\hat{v}_c) \prod_{c=b+1}^N \theta(\hat{v}_c-\hat{v}_b)\in Z(\mathcal{A}_{\mathcal{N}_b}^G).
\end{equation}
As we will shortly show, the corresponding density operator for the later cut is $\rho_a' = \Pi_a e^{\hat{q}}\Psi_{\mathcal{S}_a}(\tilde\sigma)$, with $\Pi_a$ defined similarly to $\Pi_b$.
Let us set $\tilde\sigma = e^{\hat{d}_a/2}\tilde\rho e^{\hat{d}_a/2}$ for some $\tilde\rho$ satisfying $\tr_{\mathcal{H}_a^{>0}}(\tilde\rho)=1$. Then
\begin{equation}
  \rho_a' = \Pi_a e^{\hat{q}}, \qquad \rho_b' = \Pi_b e^{i\hat{v}_b\hat{k}_a}e^{(\hat{q}+\hat{d}_a)/2}e^{i\hat{p}\hat{d}_a}\tilde\rho e^{-i\hat{p}\hat{d}_a} e^{(\hat{q}+\hat{d}_a)/2}e^{-i\hat{v}_b\hat{k}_a}.
\end{equation}
Using that $\hat{q}+\hat{d}_a$ commutes with $e^{i\hat{p}\hat{d}_a}\tilde\rho e^{-i\hat{p}\hat{d}_a}$, we have
\begin{equation}
  \Pi_a\log\rho_a' = \Pi_a \hat{q}, \qquad \Pi_b\log\rho_b' = \Pi_b \qty(\hat{q}+\hat{d}_a-2\pi \hat{v}_b\hat{k}_a + \log \hat\rho) = \Pi_b\qty(\hat{q} + H_{\mathcal{S}_a|\mathcal{S}_b} + \log \hat\rho),
\end{equation}
where $\hat\rho = e^{i\hat{v}_b\hat{k}_a}e^{i\hat{p}\hat{d}_a}\tilde\rho e^{-i\hat{p}\hat{d}_a}e^{-i\hat{v}_b\hat{k}_a}$ is some normalised density operator in $\mathcal{A}_{\mathcal{S}_a\in\mathcal{N}_b}^G$.

Applying this to~\eqref{Equation: monotonicity Phi'} yields
\begin{equation}
  S_{\mathcal{N}_a}[\Phi] \ge
  S_{\mathcal{N}_b}[\Phi] + \expval{H_{\mathcal{S}_a|\mathcal{S}_b} + \log \hat\rho}_\Phi.
\end{equation}
This inequality holds for any normalised density operator $\hat\rho\in\mathcal{A}_{\mathcal{S}_a\in\mathcal{N}_b}^G$. We get the strongest possible inequality by maximising the right-hand side, and it turns out that this maximum is attained when $\hat\rho$ is the density operator of $\Phi$, in which case $\expval{\log\hat\rho}_\Phi=-S_{\mathcal{S}_a\in\mathcal{N}_b}[\Phi]$.\footnote{Indeed, this follows from the fact that the relative entropy of two states on the algebra $\mathcal{A}_{\mathcal{S}_a\in\mathcal{N}_b}^G$ is non-negative, and vanishes if and only if the two states are equal.} Therefore, we have
\begin{equation}
  S_{\mathcal{N}_a}[\Phi] \ge
  S_{\mathcal{N}_b}[\Phi] + F(\mathcal{S}_a\in\mathcal{N}_b).
  \label{Equation: proven GSL}
\end{equation}
This is the GSL for dynamical cuts, and we have now proven it. It holds without a semiclassical limit or a UV cutoff, but we will show in Section~\ref{Section: semiclassical} that it reduces to the usual GSL when we do impose these.

Let us now explain why $\rho_a'=\Pi_ae^{\hat{q}}\Psi_{\mathcal{S}_a}(\tilde\sigma)$. Consider the weight on $\mathcal{A}_{\mathcal{N}_b,\mathcal{S}_a>\mathcal{S}_b}^G$ given by
\begin{equation}
  \Phi'(a) = \Tr_b(\rho_b' a) = \Tr_b(\Pi_b e^{\hat{q}/2} \sigma e^{\hat{q}/2}a), \qquad a\in\mathcal{A}_{\mathcal{N}_b,\mathcal{S}_a>\mathcal{S}_b}^G.
\end{equation}
Suppose we restrict to $a\in \mathcal{A}_{\mathcal{N}_a,\mathcal{S}_a>\mathcal{S}_b}^G$. For such $a$ it may be confirmed that
\begin{equation}
  \gamma_a(a) = \tilde{a}\otimes\theta(-\hat{v}_b),\qquad \gamma_b(a) = e^{i\hat{v}_a(P+\sum_{c\ne a,b}\hat{k}_c)}\qty(\tilde{a} \otimes\theta(\hat{v}_a)) e^{-i\hat{v}_a(P+\sum_{c\ne a,b}\hat{k}_c)}
\end{equation}
for some $\tilde{a}$. Using
\begin{equation}
  \theta(\hat{v}_a)\prod_{c=1}^{b-1}\theta(-\hat{v}_c)\gamma_b(a) = \prod_{c=1}^{b-1}\theta(-\hat{v}_c) e^{i\hat{v}_a(P+\sum_{c=a+1}^N\hat{k}_c)}\qty(\tilde{a} \otimes\theta(\hat{v}_a)) e^{-i\hat{v}_a(P+\sum_{c=a+1}^N\hat{k}_c)},
\end{equation}
we then have
\begin{align}
  \Phi'(a) &= \Psi_b\qty\Big(\bra{0}_p\gamma_b(\Pi_b)\gamma_b(\sigma) e^{\hat{q}/2}\gamma_b(a)e^{-\hat{q}/2}\ket{0}_p)\\
  &= \Psi_b\qty\Big(\prod_{c=1}^{b-1}\theta(-\hat{v}_c)\prod_{c=a+1}^N\theta(\hat{v}_c)\tilde\sigma \bra{0}_pe^{\hat{q}/2}e^{-i\hat{v}_a(P+\sum_{c=a+1}^N\hat{k}_c)}\qty(\tilde{a} \otimes\theta(\hat{v}_a)) e^{i\hat{v}_a(P+\sum_{c=a+1}^N\hat{k}_c)}e^{-\hat{q}/2}\ket{0}_p)\label{Equation: partial trace a}\\
  &= \Psi_b\qty\Big(\prod_{c=1}^{b-1}\theta(-\hat{v}_c)\prod_{c=a+1}^N\theta(\hat{v}_c)e^{i\hat{v}_a(P+\sum_{c=a+1}^N\hat{k}_c)}\tilde\sigma e^{-i\hat{v}_a(P+\sum_{c=a+1}^N\hat{k}_c)} \bra{0}_pe^{\hat{q}/2}\qty(\tilde{a} \otimes\theta(\hat{v}_a)) e^{-\hat{q}/2}\ket{0}_p)\label{Equation: partial trace b}\\
  &= \Psi_b\qty\Big(\prod_{c=1}^{b-1}\theta(-\hat{v}_c)\prod_{c=a+1}^N\theta(\hat{v}_c)\tilde\sigma \bra{0}_pe^{\hat{q}/2}\qty(\tilde{a} \otimes\theta(\hat{v}_a)) e^{-\hat{q}/2}\ket{0}_p)\\
  &= \Psi_a\qty\Big(\prod_{c=1}^{b}\theta(-\hat{v}_c)\prod_{c=a+1}^N\theta(\hat{v}_c)\Psi_{\mathcal{S}_a}(\tilde\sigma) \bra{0}_pe^{\hat{q}/2}\qty(\tilde{a} \otimes\theta(-\hat{v}_b)) e^{-\hat{q}/2}\ket{0}_p)\\
  &= \Psi_{\mathcal{S}_a}(\tilde\sigma) \Psi_a(\bra{0}_pe^{\hat{q}/2}\gamma_a(a)\gamma_a(\Pi_a) e^{-\hat{q}/2}\ket{0}_p)\\
  &= \Psi_{\mathcal{S}_a}(\tilde\sigma)\Tr_a(\Pi_ae^{\hat{q}}a),
\end{align}
so $\rho'_a = \Psi_{\mathcal{S}_a}(\tilde\sigma) \Pi_a e^{\hat{q}}$, as claimed.

In going from~\eqref{Equation: partial trace a} to~\eqref{Equation: partial trace b}, we have used that, for any $s\ge 0$ and $x\in\mathcal{B}(\mathcal{H}^{>0}_c)$,
\begin{equation}
  \Psi_{\mathcal{S}_c}(e^{is\hat{k}_c}xe^{-is\hat{k}_c})
  =
  \Psi_{\mathcal{S}_c}(x).
\end{equation}
This follows from the cyclic property of $\tr_{\mathcal{H}^{>0}_c}$, and the fact that $e^{-is\hat{k}_c}e^{-\hat{d}_c}e^{is\hat{k}_c}=e^{-\hat{d}_c}$ (which can be obtained from an analytic continuation of $e^{it\hat{d}_c}e^{is\hat{k}_c}e^{-it\hat{d}_c}=e^{ise^{2\pi t}\hat{k}_c}$ to $t=i$, since $e^{is\hat{k}_c}$ is a bounded analytic function of $\hat{k}_c$).

\subsection{Different trace and entropy normalisation conventions}
\label{Subsection: normalisation ambiguities}
\changed{So far we have been using a particular set of traces on the algebras $\mathcal{A}^G_{\mathcal{N}_a}$. However, as noted at the beginning of Subsection~\ref{Subsection: traces}, these algebras do not have unique traces. Indeed, the most general trace for $\mathcal{A}^G_{\mathcal{N}_a}$ is of the form $\Tr_a^z$ given in~\eqref{Equation: trace ambiguity}, with $z\in Z(\mathcal{A}^G_{\mathcal{N}_a})$ a positive central element. As discussed there, the von Neumann entropy according to the trace $\Tr_a^z$ differs from that according to the trace $\Tr_a$ by the expectation value of $\log z$, as in~\eqref{Equation: entropy ambiguity}. We fixed this ambiguity by simply setting $z=\mathds{1}$.}

\changed{On the other hand, since we derived the GSL from the monotonicity of relative entropy, and relative entropy is not subject to such ambiguities, it should be clear that the physical content of the GSL is also independent of the ambiguities. Instead, using different traces simply organises the terms in the GSL in different ways. It is the purpose of this subsection to explicitly demonstrate this.}

\changed{Having obtained the form of the GSL for the particular choices of trace $\Tr_a$, $\Tr_b$, we can easily write down the form it would take if we had used different traces $\Tr_a^{z_a}$, $\Tr_b^{z_b}$, for $z_a\in Z(\mathcal{A}^G_{\mathcal{N}_a})$, $z_b\in Z(\mathcal{A}^G_{\mathcal{N}_a})$. The von Neumann entropies according to these traces will be}
\begin{equation}
  \changed{S_{\mathcal{N}_a}^{z_a}[\Phi] = S_{\mathcal{N}_a}[\Phi] + \Phi(\log z_a), \qquad S_{\mathcal{N}_b}^{z_b}[\Phi] = S_{\mathcal{N}_b}[\Phi] + \Phi(\log z_b).}
\end{equation}
\changed{Substituting this in to~\eqref{Equation: proven GSL} yields}
\begin{equation}
  \changed{S^{z_a}_{\mathcal{N}_a}[\Phi] \ge
  S^{z_b}_{\mathcal{N}_b}[\Phi] + F(\mathcal{S}_a\in\mathcal{N}_b) + \Phi(\log z_a - \log z_b).}
  \label{Equation: proven GSL different convention}
\end{equation}
\changed{We see that for the entropies defined with more general traces, an extra term appears on the right-hand side, compensating for the relative normalisation of the von Neumann entropies. One convenience of the traces we have been using is that this extra term can be ignored.}

\subsection{Necessity of the cut free energy term}

In~\eqref{Equation: proven GSL}, the most obvious difference with the standard GSL is the free energy term $F(\mathcal{S}_a\in\mathcal{N}_b)$. This term is essential -- it accounts for correlations between $\mathcal{S}_a$ and degrees of freedom outside $\mathcal{N}_a$. If the free energy term were not present, the resulting inequality $S_{\mathcal{N}_a}[\Phi] \ge S_{\mathcal{N}_b}[\Phi]$ would in fact be \emph{violated} in a significant class of states, for which such correlations play a non-trivial role. This would not be a problem in the context of the ordinary semiclassical GSL, where cuts are not dynamical, so such correlations do not need to be accounted for. But for dynamical cuts, it is important.

To demonstrate this, let us restrict our attention to a very simple class of states. For concreteness, we will use states of the form $\Phi(a) = \Bra{\phi}r(a)\Ket{\phi}$, were $\Ket{\phi}$ is a vector in the physical Hilbert space constructed in Section~\ref{Subsection: physical Hilbert space}. We will assume that there are only two dynamical cuts, $\mathcal{S}_a$ and $\mathcal{S}_b$, and consider $\Ket{\phi}\in\mathcal{H}_{\text{phys}}$ of the form
\begin{equation}
  \Ket{\phi} = \zeta\qty\big(\ket{\Omega}\otimes\ket{f}\otimes\ket{\varphi}\otimes\ket{0}_p\otimes\ket{0}_b).
\end{equation}
Here $\ket{f}\in\mathcal{H}_a$, $\ket{\varphi}\in\mathcal{H}_{-\infty}$ are some normalised states of $\mathcal{S}_a$ and the past asymptotic frame at $v=-\infty$ respectively. The overall state $\Ket{\phi}$ corresponds to the state `in the perspective of $C$ and $\mathcal{S}_b$' being
\begin{equation}
  \ket*{\phi_{|C\mathcal{S}_b}} = \mathcal{R}_{C\mathcal{S}_b}\Ket{\phi} = \ket{\Omega}\otimes\ket{f}\otimes\ket{\varphi},
\end{equation}
i.e.\ the tensor product of the QFT vacuum $\ket{\Omega}$, $\ket{f}$ and $\ket{\varphi}$. We will assume that $\Ket{\phi}$ is a state satisfying $\mathcal{S}_a>\mathcal{S}_b$, which here may be written as the condition that $\bra{v_a}_a\ket{f}=0$ for all $v_a<0$.

In this state, the entanglement entropy contribution to the free energy vanishes: $S_{\mathcal{S}_a\in\mathcal{N}_b}[\Phi] =0$.\footnote{Indeed, from footnote~\ref{Footnote: explicit free energy state}, one has that the density operator for this state in the algebra $\mathcal{A}_{\mathcal{S}_a\in\mathcal{N}_b}^G$ is pure:
  \begin{equation}
    r(\hat{\rho}_a) = r\qty\big[e^{i\hat{v}_a\hat{k}_b}e^{i\hat{p}\hat{d}_a}\ket{f}\bra{f}e^{-i\hat{p}\hat{d}_a}e^{-i\hat{v}_a\hat{k}_b}].
\end{equation}} One could also consider examples where there is non-trivial entanglement, but here we are just considering the simplest possible case, to demonstrate the necessity of the cut free energy term. The expectation value of $H_{\mathcal{S}_a|\mathcal{S}_b}$ also simplifies via $\expval*{H_{\mathcal{S}_a|\mathcal{S}_b}}_\Phi= \bra{f}\hat{d}_a\ket{f}$, so overall one has
\begin{equation}
  F(\mathcal{S}_a\in\mathcal{N}_b) = \bra{f}\hat{d}_a\ket{f}.
\end{equation}

Let us now compute the two von Neumann entropies $S_{\mathcal{N}_a}[\Phi]$, $S_{\mathcal{N}_b}[\Phi]$. One may straightforwardly substitute the state $\Ket{\phi}$ into the density operator formula~\eqref{Equation: density operator formula}. One finds that the density operators $\rho_a\in r(\mathcal{A}_{\mathcal{N}_a}^G)$ and $\rho_b\in r(\mathcal{A}_{\mathcal{N}_b}^G)$ take the following forms:
\begin{equation}
  \rho_a = r\qty[2\pi e^{\hat{q}}\theta(\hat{v}_a-\hat{v}_b)\tr_{a\,-\infty}\qty(\ket{\chi}\bra{\chi})],\qquad
  \rho_b = r\qty[2\pi e^{\hat{q}+H_{\mathcal{S}_a|\mathcal{S}_b}}e^{i\hat{v}_b\hat{k}_a}\tr_{-\infty}\qty(\ket{\chi}\bra{\chi})e^{-i\hat{v}_b\hat{k}_a}],
\end{equation}
where
\begin{equation}
  \ket{\chi} = e^{-i\hat{p}(\hat{d}_a+\mathcal{Q})}\qty\big(\ket{f}\otimes\ket{\varphi}\otimes\ket{0}_q),
\end{equation}
$\ket{0}_q = \frac1{\sqrt{2\pi}}\int_{-\infty}^\infty\dd{p}\ket{p}_p$ is the zero eigenstate of $\hat{q}$, and $\tr_{a\,-\infty}$ and $\tr_{-\infty}$ denote ordinary Hilbert space partial traces over $\mathcal{H}_a\otimes\mathcal{H}_{-\infty}$ and $\mathcal{H}_{-\infty}$ respectively. Due to the commuting nature of the operators involved, the logarithms of these density operators obey
\begin{align}
  \log \rho_a &= \log(2\pi) + r\qty[\hat{q} + \log\theta(\hat{v}_a-\hat{v}_b) + \log(\tr_{a\,-\infty}\qty(\ket{\chi}\bra{\chi}))],\\
  \log\rho_b &= \log(2\pi) + r\qty[\hat{q}+H_{\mathcal{S}_a|\mathcal{S}_b} + e^{i\hat{v}_b\hat{k}_a}\log(\tr_{-\infty}\qty(\ket{\chi}\bra{\chi}))e^{-i\hat{v}_b\hat{k}_a}].
\end{align}
Now we may evaluate the corresponding von Neumann entropies with
\begin{align}
  S_{\mathcal{N}_a}[\Phi] &= -\Bra{\phi}\log\rho_a\Ket{\phi} \\
  &= -\log(2\pi) -\Bra{\phi}r(\hat{q})\Ket{\phi} - \Bra{\phi}r\qty[\log(\tr_{a\,-\infty}\qty(\ket{\chi}\bra{\chi}))]\Ket{\phi},\label{Equation: example entropy a}\\
  S_{\mathcal{N}_b}[\Phi] &= -\Bra{\phi}\log\rho_a\Ket{\phi} \\
  &= -\log(2\pi)-\Bra{\phi}r(\hat{q})\Ket{\phi} -\bra{f}\hat{d}_a\ket{f}- \Bra{\phi}r\qty[e^{i\hat{v}_b\hat{k}_a}\log(\tr_{a\,-\infty}\qty(\ket{\chi}\bra{\chi}))e^{-i\hat{v}_b\hat{k}_a}]\Ket{\phi}\label{Equation: example entropy b}
\end{align}
For the terms involving $\ket{\chi}$ it is convenient to use a replica trick, as follows. One may show that
\begin{equation}
  \Bra{\phi}r\qty\big[\tr_{a\,-\infty}(\ket{\chi}\bra{\chi})^n] \Ket{\phi} = \int_{-\infty}^\infty\dd{x} \mathbb{P}_{\hat{d}_a+\mathcal{Q}}(x)^{n+1},
\end{equation}
and
\begin{equation}
  \Bra{\phi}r\qty\big[e^{i\hat{v}_b\hat{k}_a}\tr_{-\infty}(\ket{\chi}\bra{\chi})^ne^{-i\hat{v}_b\hat{k}_a}] \Ket{\phi} = \int_{-\infty}^\infty\dd{x} \mathbb{P}_{\mathcal{Q}}(x)^{n+1},
\end{equation}
where
\begin{equation}
  \mathbb{P}_{\hat{d}_a+\mathcal{Q}}(x) = \bra{f}\otimes\bra{\varphi} \delta(\hat{d}_a+\mathcal{Q}-x) \ket{f}\otimes\ket{\varphi}, \qquad
  \mathbb{P}_{\mathcal{Q}}(x) = \bra{\varphi} \delta(\mathcal{Q}-x) \ket{\varphi}
\end{equation}
are the Born probability distributions associated with $\hat{d}_a+\mathcal{Q}$ and $\mathcal{Q}$ respectively.\footnote{The delta function of a Hermitian operator $a$ is defined via its Fourier transform $\delta(a) = 2\pi\int_{-\infty}^\infty\dd{y} e^{iya}$.} Taking an $n$ derivative and setting $n=0$, one recognises the final terms in~\eqref{Equation: example entropy a} and~\eqref{Equation: example entropy b} as the Shannon entropies of these probability distributions:
\begin{equation}
  - \Bra{\phi}r\qty[\log(\tr_{a\,-\infty}\qty(\ket{\chi}\bra{\chi}))]\Ket{\phi} = H(\hat{d}_a+\mathcal{Q}) = -\int_{-\infty}^\infty\dd{x} \mathbb{P}_{\hat{d}_a+\mathcal{Q}}(x) \log \mathbb{P}_{\hat{d}_a+\mathcal{Q}}(x)
\end{equation}
and
\begin{equation}
  - \Bra{\phi}r\qty[e^{i\hat{v}_b\hat{k}_a}\log(\tr_{-\infty}\qty(\ket{\chi}\bra{\chi}))e^{-i\hat{v}_b\hat{k}_a}]\Ket{\phi} = H(\mathcal{Q}) = -\int_{-\infty}^\infty\dd{x} \mathbb{P}_{\mathcal{Q}}(x) \log \mathbb{P}_{\mathcal{Q}}(x).
\end{equation}
Overall, one finds that the difference of the von Neumann entropies may be written
\begin{equation}
  \mathcal{S}_{\mathcal{N}_a}[\Phi]-\mathcal{S}_{\mathcal{N}_b}[\Phi] = F(\mathcal{S}_a\in\mathcal{N}_b) + H(\hat{d}_a+\mathcal{Q})-H(\mathcal{Q}).
  \label{Equation: simple case entropy diff}
\end{equation}

Suppose we move the free energy to the left-hand side:
\begin{equation}
  \mathcal{S}_{\mathcal{N}_a}[\Phi]-\mathcal{S}_{\mathcal{N}_b}[\Phi] - F(\mathcal{S}_a\in\mathcal{N}_b) = H(\hat{d}_a+\mathcal{Q})-H(\mathcal{Q}).
\end{equation}
The right-hand side of this equation must be non-negative (which is reassuring, because it means the simple states we are considering in this subsection are consistent with the general GSL~\eqref{Equation: proven GSL} we proved earlier). To see this, note that $\hat{d}_a$ and $\mathcal{Q}$ commute with each other, and have independent probability distributions in the state $\ket{f}\otimes\ket{\varphi}$. Thus, they may be treated as independent classical random variables. It is well known that for such random variables the expression $H(\hat{d}_a+\mathcal{Q})-H(\mathcal{Q})$ is strictly non-negative and may be arbitrarily close to zero (consequently, the states we are considering can get arbitrarily close to saturating~\eqref{Equation: proven GSL}).

On the other hand, the right-hand side of~\eqref{Equation: simple case entropy diff} can take arbitrary negative or positive values, so there can be no general inequality of the form $\mathcal{S}_{\mathcal{N}_a}[\Phi]\ge\mathcal{S}_{\mathcal{N}_b}[\Phi]$ (since it would fail for the simple states considered in this subsection). To demonstrate this, first note that the continuous spectra of $\hat{d}_a$ and $\mathcal{Q}$ are both the full real line (even when restricting to $\hat{v}_a>0$), so we can pick the states $\ket{f}$ and $\ket{\varphi}$ such that $\hat{d}_a$ and $\mathcal{Q}$ have normal probability distributions with variances $\sigma^2_{\hat{d}_a},\sigma^2_{\mathcal{Q}}$ and means $\mu_{\hat{d}_a},\mu_{\mathcal{Q}}$. Then one may show that
\begin{equation}
  H(\hat{d}_a+\mathcal{Q})-H(\mathcal{Q}) = \frac12\log(1+\sigma^2_{\hat{d}_a}\Big/\sigma^2_{\mathcal{Q}}).
\end{equation}
By considering states with $\sigma^2_{\hat{d}_a} \ll \sigma^2_{\mathcal{Q}}$, this can be made arbitrarily close to zero. On the other hand, the free energy is just
\begin{equation}
  F(\mathcal{S}_a\in\mathcal{N}_b) = \mu_{\hat{d}_a},
\end{equation}
which can be arbitrarily positive or negative, independent of the variances.

To conclude this subsection, we have exhibited states for which the entropies may be explicitly computed, and $\mathcal{S}_{\mathcal{N}_a}[\Phi]\ge\mathcal{S}_{\mathcal{N}_b}[\Phi]$ fails. This is evidently strong motivation for needing a modified GSL. Indeed, we also showed that these states are consistent with the modified GSL~\eqref{Equation: proven GSL} proposed in this paper. Of course, we already proved that~\eqref{Equation: proven GSL} holds for all states.

\section{Recovering the semiclassical UV-regulated GSL}
\label{Section: semiclassical}

Let us now show how to recover the traditional semiclassical GSL from our inequality~\eqref{Equation: proven GSL}. This happens in three steps:
\begin{enumerate}
  \item Take a classical limit of the boost clock.
  \item Impose a UV cutoff on the fields.
  \item Take a classical limit of the dynamical cuts (and impose some energy conditions).
\end{enumerate}
These could be carried out in any order, but we will do them in the order shown here. The first two steps are similar to what has appeared previously~\cite{Chandrasekaran_2023b,kudlerflam2024generalizedblackholeentropy,Jensen_2023,DeVuyst:2024pop,DEHKlong}, while the third step is new and is now necessary because the cuts are quantum and dynamical.

\subsection{Classical limit of boost clock}

Let us first restrict to states for which the time provided by asymptotic boost clock $C$ may be treated as approximately classical. We may characterise this by considering the expectation value of operators $D_b(a)e^{-i\hat{q}t}\in\mathcal{A}_{\mathcal{N}_b}^G$, where we define for notational simplicity
\begin{equation}
  D_b(a) = e^{i(P+\sum_{c\ne b}\hat{k}_c)\hat{v}_b}e^{i(H+\sum_{c\ne b}\hat{d}_c)\hat{p}}ae^{-i(H+\sum_{c\ne b}\hat{d}_c)\hat{p}}e^{-i(P+\sum_{c\ne b}\hat{k}_c)\hat{v}_b},
\end{equation}
with $a\in\mathcal{A}_{\text{QFT}}(0)\otimes\bigotimes_{c\ne b}\mathcal{A}_{\mathcal{S}_b}(0)$. The operator $D_b(a)e^{-i\hat{q}t}$ acts with $a$ dressed to the boost clock and the cut $\mathcal{S}_a$, and also changes the boost clock time by $t$.

In particular, for a given state $\Phi$ with density operator $\rho_b\in\mathcal{A}_{\mathcal{N}_b}^G$, the boost time is approximately classical if the following two conditions hold. First, the expectation value $\Tr_b(\rho_b D_b(a)e^{-i\hat{q}t})$ is sharply peaked within $\abs{t}<\epsilon$ for some small $\epsilon$. Intuitively, this means if we change the boost time a bit, we get a state basically orthogonal to the one we started with -- which is characteristic of the boost time being a classical degree of freedom. Second, within the peak, the boost clock is approximately uncorrelated from dressed operators, so
\begin{equation}
  \Tr_b(\rho_b D_b(a)e^{-i\hat{q}t}) \approx \Tr_b(\rho D_b(a)) \Tr_b(\rho e^{-i\hat{q}t}), \qquad\text{if }\abs{t}<\epsilon.
\end{equation}
This can be understood as there being no entanglement between the clock and the rest of the degrees of freedom (as in the `classical-quantum' states of~\cite{Witten_2018,Chandrasekaran_2023,Chandrasekaran_2023b,kudlerflam2024generalizedblackholeentropy,Jensen_2023,Faulkner_2016}).
One can generalise this to allow for entanglement (see~\cite{DEHKlong}), which may be interpreted in this classical limit as classical correlations. For simplicity, we consider here the uncorrelated case.

For reasons which we will explain shortly, these conditions suffice to give the following approximate expression for the density operator:
\begin{equation}
  \rho_b \approx 2\pi\,e^{\hat{q}+H+\sum_{c>b}\hat{d}_c-2\pi\hat{v}_b(P+\sum_{c>b}\hat{k}_c)}\, \tilde{f}(\hat{q})\, \Delta_{\Phi|\tilde\Psi}.
  \label{Equation: approximate density operator}
\end{equation}
We are continuing to assume that the complete ordering of the cuts~\eqref{Equation: complete ordering} holds. Here,
\begin{equation}
  \tilde{f}(q) = \Phi(\ket{q}\bra{q}) = \Tr_b(\rho_b\ket{q}\bra{q})
\end{equation}
is the probability distribution over possible values of $\hat{q}$, while $\Delta_{\Phi|\tilde\Psi}$ is the relative modular operator of the algebra
\begin{equation}
  D_b\qty\Big(\mathcal{A}_{\text{QFT}}(0)\otimes \bigotimes_{c>b}\mathcal{A}_{\mathcal{S}_c}(0))
  \label{Equation: dressed algebra without area}
\end{equation}
from $\tilde\Psi$ to $\Phi$, where
\begin{equation}
  \tilde\Psi\circ D_b = \Psi_{\text{QFT}} \otimes\bigotimes_{c>b}\tr_{\mathcal{H}_c^{>0}}.
  \label{Equation: tilde Psi}
\end{equation}
Physically speaking, the algebra~\eqref{Equation: dressed algebra without area} consists of all gauge-invariant operators in $\mathcal{N}_b$, except for those involving the future asymptotic area $\hat{q}$.

We will also show that the conditions above imply that $\tilde{f}(\hat{q})$ and $\Delta_{\Phi|\tilde\Psi}$ approximately commute with each other. This can be understood as meaning that the state of the subsystem corresponding to the algebra~\eqref{Equation: dressed algebra without area} is effectively constant over boost time intervals of size $\epsilon$ (this follows from~\eqref{Equation: slowly varing P(0)} and~\eqref{Equation: P(0) Delta} below). Another way to put this is that the dressed operators in $\mathcal{N}_b$ fluctuate much less than the future asymptotic horizon area $\hat{q}$ -- this is a consistency condition required by the semiclassical regime.

Since the three operator factors in~\eqref{Equation: approximate density operator} all approximately commute with each other, one has
\begin{equation}
  \log\rho_b \approx \log(2\pi)+\hat{q}+H+\sum_{c>b}\hat{d}_c-2\pi\hat{v}_b(P+\sum_{c>b}\hat{k}_c) + \log \tilde{f}(\hat q) + \log \Delta_{\Phi|\tilde\Psi},
\end{equation}
so the entropy may be written
\begin{align}
  S[\mathcal{N}_b] &= -\Phi(\log \rho_b) \\
  &\approx -\log(2\pi)- \Big\langle\hat{q}+H+\sum_{c>b}\hat{d}_c-2\pi\hat{v}_b(P+\sum_{c>b}\hat{k}_c)\Big\rangle_\Phi -\int_{-\infty}^\infty\dd{q}\tilde{f}(q)\log\tilde{f}(q) - \Phi(\log\Delta_{\Phi|\tilde\Psi}).
  \label{Equation: classical boost time entropy}
\end{align}
A similar approximation holds for $S[\mathcal{N}_a]$.

Let us now demonstrate the validity of the above approximations. We may always write the density operator as
\begin{equation}
  \rho_b = \int_{-\infty}^\infty\dd{t'} e^{\hat{q}/2} e^{i\hat{q}t'}D_b(P(t'))e^{\hat{q}/2}.
\end{equation}
Then we have
\begin{equation}
  \Tr_b(\rho_b D_b(a)e^{-i\hat{q}t}) = \Psi_b( P(t) e^{-(H+\sum_{c\ne b}\hat{d}_c)/2} a e^{(H+\sum_{c\ne b}\hat{d}_c)/2}).
\end{equation}
By the first condition above, this is suppressed for $\abs{t}>\epsilon$, and by the faithfulness of $\Psi_b$ this implies that $P(t)$ is suppressed for $\abs{t}>\epsilon$. Moreover, by the second condition we may write
\begin{equation}
  \Tr_b(\rho_b D_b(a)e^{-i\hat{q}t}) \approx \Psi_b( P(0)e^{-(H+\sum_{c\ne b}\hat{d}_c)/2}ae^{(H+\sum_{c\ne b}\hat{d}_c)/2})f(t),
\end{equation}
where $f(t) = \Psi_b(P(t)) = \Tr_b(\rho_be^{-i\hat{q}t})$ is sharply peaked at $\abs{t}<\epsilon$.
Again, by the faithfulness of $\Psi_b$, we can conclude that $P(t) \approx f(t)P(0)$ for $\abs{t}<\epsilon$, and hence we have the approximation
\begin{equation}
  \rho_b \approx \int_{-\infty}^\infty\dd{t'} e^{\hat{q}/2} e^{i\hat{q}t'}f(t')D_b(P(0))e^{\hat{q}/2} = 2\pi e^{\hat{q}/2}\tilde{f}(\hat{q})D_b(P(0))e^{\hat{q}/2},
  \label{Equation: rho deltad}
\end{equation}
where
\begin{equation}
  \tilde{f}(\hat{q}) = \frac1{2\pi} \int \dd{t}e^{i\hat{q}t} f(t) = \frac1{2\pi} \int \dd{t}e^{i\hat{q}t} \Tr_b(\rho_be^{-i\hat{q}t}) = \int_{-\infty}^\infty \dd{q} \Phi(\ket{q}\bra{q}) \ket{q}\bra{q},
\end{equation}
Note that the second condition above also gives, for $\abs{t}<\epsilon$,
\begin{equation}
  \Tr_b(e^{i\hat{q}t}\rho_b e^{-i\hat{q}t}D_b(a)) \approx \Tr_b(\rho_b D_b(a)),
\end{equation}
which may be written
\begin{multline}
  \Psi_b( P(0)e^{-(H+\sum_{c\ne b}\hat{d}_c)/2}ae^{(H+\sum_{c\ne b}\hat{d}_c)/2}) \\
  \approx
  \Psi_b( e^{-i(H+\sum_{c\ne b}\hat{d}_c)t}P(0)e^{i(H+\sum_{c\ne b}\hat{d}_c)t} e^{-(H+\sum_{c\ne b}\hat{d}_c)/2}ae^{(H+\sum_{c\ne b}\hat{d}_c)/2}).
\end{multline}
Again by the faithfulness of $\Psi_b$, this implies
\begin{equation}
  P(0) \approx e^{-i(H+\sum_{c\ne b}\hat{d}_c)t}P(0)e^{i(H+\sum_{c\ne b}\hat{d}_c)t},\qquad \text{if }\abs{t}<\epsilon.
  \label{Equation: slowly varing P(0)}
\end{equation}
Since $f(t)$ is peaked in $\abs{t}<\epsilon$, this in particular implies that $[P(0),\tilde{f}(\hat{q})]\approx 0$. Now, continuing to assume the conditioning $\mathcal{S}_1<\dots<\mathcal{S}_N$ holds, we may write (comparing the definitions of $\Psi_b$ and $\tilde\Psi$)
\begin{align}
  \Tr_b(\rho_b D_b(ab)) &\approx \tilde\Psi\circ D_b(b e^{-(H+\sum_{c>b}\hat{d}_c)/2} P(0)e^{-(H+\sum_{c> b}\hat{d}_c)/2}a)\\
  &= \tilde\Psi\qty(D_b(b)D_b\qty\big(e^{-(H+\sum_{c>b}\hat{d}_c)/2} P(0)e^{-(H+\sum_{c> b}\hat{d}_c)/2})D_b(a)).
\end{align}
From this we may conclude that
\begin{equation}
  D_b\qty\big(e^{-(H+\sum_{c>b}\hat{d}_c)/2} P(0)e^{-(H+\sum_{c> b}\hat{d}_c)/2}) = \Delta_{\Phi|\tilde\Psi}
  \label{Equation: P(0) Delta}
\end{equation}
is the relative modular operator from $\tilde\Psi$ to $\Phi$ for the algebra~\eqref{Equation: dressed algebra without area}. Using
\begin{equation}
  D_b\qty\big(e^{-(H+\sum_{c>b}\hat{d}_c)/2}) = e^{-\qty\big(H+\sum_{c>b}\hat{d}_c-2\pi\hat{v}_b(P+\sum_{c>b}\hat{k}_c))/2},
\end{equation}
we may combine~\eqref{Equation: P(0) Delta} with~\eqref{Equation: rho deltad} (and use that $\hat{q}+H+\sum_{c>b}\hat{d}_c-2\pi\hat{v}_b(P+\sum_{c>b}\hat{k}_c)$ commutes with all operators in~\eqref{Equation: dressed algebra without area}) to obtain~\eqref{Equation: approximate density operator}, as required.

\subsection{Imposing a UV cutoff}

We now impose a UV cutoff, such that all QFT algebras become Type I, following closely the methods used in~\cite{Jensen_2023,Faulkner:2024gst}. We then have that the algebra~\eqref{Equation: dressed algebra without area} is Type I, so we can write the relative modular operator in terms of density operators:
\begin{equation}
  \Delta_{\Phi|\tilde\Psi} = \rho_{\Phi} \qty\big(\rho'_{\tilde\Psi})^{-1}.
\end{equation}
Here $\rho_{\Phi}$ is the density operator in the algebra~\eqref{Equation: dressed algebra without area} for $\Phi$, and $\rho'_{\tilde\Psi}$ is the density operator in the commutant of~\eqref{Equation: dressed algebra without area} for $\tilde\Psi$.

Due to the form~\eqref{Equation: tilde Psi} of $\tilde\Psi$, we have that $\rho'_{\tilde\Psi} = D_b(\rho'_\Omega)$, where $\rho'_\Omega$ is the density operator of the QFT vacuum $\ket{\Omega}$ with respect to the algebra $\mathcal{A}_{\text{QFT}}(0)'$. Here we can use that $\ket{\Omega}$ is thermal with respect to the one-sided QFT modular Hamiltonian
\begin{equation}
  \rho'_\Omega = \frac{e^{-K'}}Z,
\end{equation}
where $Z$ is a normalisation constant, and $K'$ may be written as an integral of the field stress tensor over the horizon:
\begin{equation}
  K' = -2\pi\int_{\mathscr{H}}\eta \theta(-v)v T_{vv}.
\end{equation}
We thus have
\begin{equation}
  \log \Delta_{\Phi|\tilde\Psi} = \log \rho_\Phi - \log D_b(\rho'_{\Omega}) = \log \rho_\Phi + \log Z - 2\pi\int_{\mathscr{H}}\eta\theta(\hat{v}_b-v)(v-\hat{v}_b) T_{vv}.
\end{equation}
Combining this with~\eqref{Equation: classical boost time entropy}, one finds
\begin{multline}
  S[\mathcal{N}_b] \approx -\log(2\pi)- \Big\langle\hat{q}+2\pi\int_{\mathscr{H}}\eta\theta(v-\hat{v}_b)(v-\hat{v}_b)T_{vv}+\sum_{c>b}\qty(\hat{d}_c-2\pi\hat{v}_b\hat{k}_c)\Big\rangle_\Phi \\
  -\log Z -\int_{-\infty}^\infty\dd{q}\tilde{f}(q)\log\tilde{f}(q) - \Phi(\log \rho_\Phi).
  \label{Equation: UV regulated entropy}
\end{multline}

At this stage, we make two identifications. With the first, we use the gravitational constraint~\eqref{Equation: gravitational constraint area cut} to write
\begin{equation}
  \frac{\expval*{A^{(2)}(\mathcal{S}_b)}_\Phi}{4G_{\text{N}}}=-\Big\langle\hat{q}+\int_{\mathscr{H}}\eta\theta(v-\hat{v}_b)(v-\hat{v}_b)T_{vv}+\sum_{c>b}\qty(\hat{d}_c-2\pi\hat{v}_b\hat{k}_c)\Big\rangle_\Phi
\end{equation}
where $A^{(2)}(\mathcal{S}_b)$ is the second order perturbation of the area of $\mathcal{S}_b$. With the second identification, we write
\begin{equation}
  S_{\text{out}}(\mathcal{N}_b) =-\int_{-\infty}^\infty\dd{q}\tilde{f}(q)\log\tilde{f}(q) - \Phi(\log \rho_\Phi).
\end{equation}
Indeed, the first term is the Shannon entropy of the future asymptotic area $\hat{q}$, and the second term is the von Neumann entropy of~\eqref{Equation: dressed algebra without area}, i.e.\ everything else in $\mathcal{N}_b$. These two are uncorrelated under the assumptions we have made, which is why they contribute separately here, and it is sensible to view $S_{\text{out}}(\mathcal{N}_b)$ as the total UV-regulated von Neumann entropy of the degrees of freedom in $\mathcal{N}_b$. For reasons explained in~\cite{DEHKlong}, this continues to be the case even in the presence of correlations between $\hat{q}$ and the rest of the degrees of freedom.

Altogether, one finds
\begin{equation}
  S[\mathcal{N}_b] \approx \frac{\expval*{A^{(2)}(\mathcal{S}_b)}_\Phi}{4G_{\text{N}}} + S_{\text{out}}(\mathcal{N}_b) - \log(2\pi) -\log Z.
  \label{Equation: UV regulated generalised entropy}
\end{equation}
Therefore, the von Neumann entropy in $\mathcal{N}_b$ reduces to the generalised entropy in $\mathcal{N}_b$, with a subtracted state-independent constant. The same formula applies to $\mathcal{N}_a$, with the same constant. The constant therefore cancels out in~\eqref{Equation: dynamical GSL}, which thus now may be written
\begin{equation}
  S_{\text{gen}}(\mathcal{N}_a)\ge S_{\text{gen}}(\mathcal{N}_b) + F(\mathcal{S}_a\in\mathcal{N}_b),
  \label{Equation: half semiclassical GSL}
\end{equation}
which is the form~\eqref{Equation: modified GSL} found with the heuristic thermodynamical derivation in the Introduction. Unlike~\eqref{Equation: dynamical GSL}, which works to all orders beyond the semiclassical regime, this inequality is now merely \emph{approximate}, holding only at leading order.

\subsection{Classical limit of dynamical cuts}

It remains to consider the classical limit of the dynamical cuts. This happens in a somewhat different, but conceptually similar way to the classical limit of the boost clock. The aim will be to constrain the behaviour of the free energy $F(\mathcal{S}_a\in\mathcal{N}_b)$.

The term $S_{\mathcal{S}_a\in\mathcal{N}_b}[\Phi]$ that contributes to the free energy is the entropy of $\mathcal{S}_a$, dressed to $\mathcal{S}_b$ and the boost clock $C$. In the perspective-neutral approach to QRFs, this has the interpretation of the entanglement entropy of $\mathcal{S}_a$ `in the perspective of' $\mathcal{S}_b$ and $C$. For a classical limit, it may be well-motivated to require that $\mathcal{S}_a$ is unentangled, so that $S_{\mathcal{S}_a\in\mathcal{N}_b}[\Phi]=0$. Let us at first assume this is the case.

The free energy then reduces to the expectation value of
\begin{equation}
  H_{\mathcal{S}_a|\mathcal{S}_b} = D_b(\hat{d}_a) = \pi \pb{D_b(\hat{k}_a)}{D_b(\hat{v}_a)}.
\end{equation}
The operator $D_b(\hat{v}_a)$ appearing here measures the location of the later cut $\mathcal{S}_a$, dressed to the earlier cut $\mathcal{S}_b$ and the boost clock $C$. In a classical limit for the cuts, we will assume that the state $\Phi$ is very sharply peaked at a certain value for $D_b(\hat{v}_a)$. One may then approximately decompose the expectation value of $H_{\mathcal{S}_a|\mathcal{S}_b}$ as
\begin{equation}
  \expval*{H_{\mathcal{S}_a|\mathcal{S}_b}}_\Phi \approx 2\pi \expval*{D_b(\hat{k}_a)}_\Phi\expval*{D_b(\hat{v}_a)}_\Phi.
  \label{Equation: classical H}
\end{equation}
Now, since we have conditioned on $\mathcal{S}_a>\mathcal{S}_b$, we have $\expval*{D_b(\hat{v}_a)}_\Phi\ge 0$. If $\Phi$ is a state in which $\mathcal{S}_a$ obeys the null energy condition, then we also have $\expval*{D_b(\hat{k}_a)}_\Phi\ge 0$.

Therefore, if $\mathcal{S}_a$ is unentangled and obeys the null energy condition, its free energy is non-negative:
\begin{equation}
  F(\mathcal{S}_a\in\mathcal{N}_b) = \expval*{H_{\mathcal{S}_a|\mathcal{S}_b}}_\Phi \ge 0.
\end{equation}
Using this in~\eqref{Equation: half semiclassical GSL}, one obtains the traditional GSL
\begin{equation}
  S_{\text{gen}}(\mathcal{N}_a)\ge S_{\text{gen}}(\mathcal{N}_b).
  \label{Equation: semiclassical traditional GSL}
\end{equation}

Even in the case the cut $\mathcal{S}_a$ does have non-vanishing entropy $S_{\mathcal{S}_a\in\mathcal{N}_b}[\Phi]$, a similar result should hold, so long as $\Phi$ is a state in which $\mathcal{S}_a$ obeys the \emph{quantum} null energy condition (QNEC)~\cite{Bousso:2015mna,Bousso:2015wca}, which roughly speaking says that the entropy $S$ obeys $\frac{1}{2\pi A}\dv[2]{v} S\le \expval{t_{vv}}$. If one isolates the contribution from the cut $\mathcal{S}_a$, and uses~\eqref{Equation: classical H}, one finds that this inequality leads to $F(\mathcal{S}_a\in\mathcal{N}_b)\ge 0$, and so~\eqref{Equation: semiclassical traditional GSL} holds.

Throughout the paper we have been using an idealised model for each dynamical cut as a position operator on an auxiliary $L^2(\RR)$. It is important to note that such cuts will not obey the NEC or QNEC for all states. However, a more realistic model would involve forming the dynamical cuts out of pre-existing field degrees of freedom. For a large class of QFTs, the fields \emph{do} obey the QNEC for all states~\cite{Bousso:2015wca}. Thus, in such a model it may be that~\eqref{Equation: half semiclassical GSL} implies~\eqref{Equation: semiclassical traditional GSL} in all states. We leave further investigation of this for future work.

\section{Conclusion}
\label{Section: Conclusion}

In this paper, we have explained how one may account for the dynamics of horizon degrees of freedom, in the presence of boost and null translation gauge symmetries, by making use of \emph{dynamical cuts}. We constructed the algebras of operators supported outside such dynamical cuts, and showed that they admit well-defined von Neumann entropies. We then proved that if we condition on an ordering for the cuts, the entropies obey the inequality
\begin{equation}
  S_{\mathcal{N}_a}[\Phi] \ge
  S_{\mathcal{N}_b}[\Phi] + F(\mathcal{S}_a\in\mathcal{N}_b).
  \label{Equation: conclusion GSL}
\end{equation}
In other words, the entropy of the later region $\mathcal{N}_a$ is bounded below by the entropy of the earlier region $\mathcal{N}_b$ plus the free energy of the later cut $\mathcal{S}_a$ inside $\mathcal{N}_b$. As we have shown, this inequality holds exactly, without needing to go to a semiclassical limit or impose a UV cutoff. To support our proposal that~\eqref{Equation: conclusion GSL} is the appropriate form of the GSL beyond these regimes, we showed that if one does go to a semiclassical limit and imposes a UV cutoff, $S_{\mathcal{N}_a}[\Phi], S_{\mathcal{N}_b}[\Phi]$ reduce to the semiclassical generalised entropies of the respective regions (up to a state-independent constant), and argued that a cut obeying the appropriate energy conditions has positive free energy $F(\mathcal{S}_a\in\mathcal{N}_b)$. Thus, under these conditions, the inequality reduces to the ordinary GSL, $\Delta S_{\text{gen}}\ge 0$.

Before ending, we would like to comment on the various limitations of our results.

First, the dynamical cuts we have employed are idealised, and auxiliary to the system under consideration. They are the simplest possible degrees of freedom one could have used as dynamical cuts. In a more physically realistic model, the dynamical cuts should be made up of the degrees of freedom already present in the system, and will not be as well-behaved as the ones we have used. Nevertheless, we expect that it will be possible to extend our results to other kinds of cuts. We should note that there must be \emph{some} degrees of freedom available that resemble dynamical cuts, if we are to have any hope of establishing a GSL in a theory with diffeomorphism invariance; indeed, as discussed in the Introduction and Section~\ref{Section: fixed cuts}, the need to \emph{relationally} define horizon exterior regions makes dynamical cuts crucial. There are various places these degrees of freedom could come from. They might originate in the asymptotic region, as is the case for the cuts defined by infalling particles shown in Figure~\ref{Figure: dynamical cuts from infalling particles}. Alternatively, they might be degrees of freedom intrinsic to the horizon, which would presumably usually only be explicitly visible at higher orders in perturbation theory. Probably, a full account should allow for all kinds of different dynamical cuts; depending on which cuts we use, we would get different manifestations of the GSL~\cite{DeVuyst:2024pop}. For initial steps in using gravitational QRFs formed out of fields, see~\cite{Chen:2024rpx,Kaplan:2024xyk}. Also, see~\cite{Ciambelli:2024swv} for a promising set of degrees of freedom which can be used as a clock intrinsic to the horizon, emerging from a certain anomaly.

Among the simplifications in the model of dynamical cuts we used is that they are `ideal' QRFs, meaning their location on the horizon may be specified with arbitrary precision, which meant we were able to easily condition on a definite ordering for the cuts, and thus formulate a GSL. In contrast, more physically realistic cuts would usually be `non-ideal', meaning they have some fundamental uncertainty in their location, which might obstruct conditioning on a definite ordering. One way to construct a model of non-ideal cuts is to restrict the spectra of the null energies $\hat{k}_a$ of the ideal cuts considered in this paper. This would allow us to impose the null energy condition on all states by bounding the spectra of $\hat{k}_a$ below -- to still have a good representation of boosts $k\to e^{-2\pi t}k$, the only simple option is to set the spectra of $\hat{k}_a$ to $[0,\infty)$. As explained in~\cite{Hoehn:2019fsy,Hoehn:2020epv,DeVuyst:2024pop,DEHKlong}, each $\hat{v}_a$ eigenstate is then replaced by a coherent state
\begin{equation}
  \ket{v_a}_a = \frac1{\sqrt{2\pi}}\int_0^\infty\dd{k}e^{-iv_ak_a}\ket{k_a}_a.
\end{equation}
This should be thought of the state of the dynamical cut when it is located at $v=v_a$. But now these states are not orthogonal for different values of $v_a$; indeed one has
\begin{equation}
  \bra{v_a}\ket{v_a'} = \frac12\delta(v_a-v_a') + \frac{i}{2\pi}\mathcal{P}\frac1{v_a-v_a'},
\end{equation}
with $\mathcal{P}$ the Cauchy principal value. This overlap is non-vanishing for any value of $\abs{v_a-v_a'}$, so \emph{any} state of the overall system will have non-vanishing contributions from all possible orderings of the cuts.

Of course, there may be other ways to make the cuts non-ideal, while still allowing for states with a definite ordering. But the most general kind of dynamical cut one could consider would not have this property, and in any case there is good reason to believe that indefinite causal ordering should be a feature of quantum gravity~\cite{Hardy_2007,delaHamette:2022cka,Kabel:2024lzr,delaHamette:2024xax}. Without a way to condition on a definite ordering, it is hard to say whether one should expect a GSL to hold in \emph{any} form. However, a surprising amount may be done with non-ideal QRFs that one might have expected to only be possible with ideal QRFs. We anticipate that it may still be possible to formulate a suitable generalisation of the GSL for non-ideal cuts, and would like to explore this further.

We have only accounted for a two-dimensional gauge group. The full gauge group of gravity is much larger; it is an infinite-dimensional diffeomorphism group. As mentioned in the Introduction, the gauge transformations we have accounted for are special, being deeply tied to the modular structure of QFT states. But, for a complete theory, the other gauge symmetries need to be taken into account too. Roughly speaking, the leftover gauge symmetries are given by the quotient $\operatorname{Diff}(\mathcal{M})/G$ of the spacetime diffeomorphism group $\operatorname{Diff}(\mathcal{M})$ by the group $G$ of boosts and null translations (note, this is not a normal subgroup). We have implicitly been assuming that these leftover gauge symmetries have either already been imposed (so, for example, the field algebras $\mathcal{A}_{\text{QFT}}(v)$ are defined in a way that already accounts for them), or that they will be imposed down the line. One possibility is that they do not change the story very much. Another possibility is that the argument given here cannot be extended to the case where we account for all gauge symmetries; we see no qualitative reason why this should be true. In our view, the most likely possibility is that other gauge symmetries will add more chapters to the interesting story that has been unfolding. Indeed, boost invariance regularised the entropy, null translation invariance allowed for an exact GSL -- maybe invariance under other diffeomorphisms will allow us to do even more? In any case, we are interested to see what happens next.

We have ignored most of the degrees of freedom present in the asymptotic future and past of the horizon, and we have strictly speaking restricted our attention to spacetimes for which the horizon is a Cauchy surface. We expect the techniques in~\cite[Section 4]{Faulkner:2024gst} may be generalised to include dynamical cuts, and thus account for spacetimes with more complicated asymptotics. As noted above, realistic dynamical cuts might in fact be constructed out of the asymptotic degrees of freedom. See also~\cite{Rignon-Bret:2024zhj} for a different formulation of black hole thermodynamics relative to asymptotic degrees of freedom.

The last limitation we would like to point out is perhaps the most important: everything we have been doing has been in perturbative gravity. A long term goal should be to extend our results to non-perturbative gravity. It is hard to say how much of this can be done in a model-independent way, but there are certain universal features it may be possible to exploit. For example, can one reproduce the proposals for non-perturbative generalised entropy in~\cite{hollands2024entropydynamicalblackholes,Ciambelli:2023mir,Ciambelli:2024swv}, using the kinds of algebraic methods employed in this paper? It is likely that one would need to involve more subtle techniques such as coarse-graining~\cite{Engelhardt_2019}, and quantum error correction~\cite{Pastawski_2015,QRFsQECC}. Moreover, the correct full treatment would presumably involve Type I algebras, from which the Type II algebras considered in this paper only emerge in the perturbative setting. Of course, non-perturbative effects play an incredibly important role in the quantum information theoretic properties of black holes. They are required, for instance, when using the Euclidean path integral to reproduce the Page curve~\cite{Almheiri:2019psf,Almheiri:2020cfm,Penington:2019kki,Penington:2019npb,Almheiri:2019qdq,Almheiri:2019hni}.

On that note, let us finally comment that our GSL~\eqref{Equation: conclusion GSL} contains a loophole that may admit consistency with the Page curve: the generalised entropy outside a dynamical cut can decrease, if the free energy of a dynamical cut is negative. In principle, this could happen if the cut becomes entangled with enough Hawking radiation, which may be the case after the Page time. Indeed, this is the kind of behaviour we would need for a fine-grained entropy (such as the von Neumann entropy we have constructed here) to be consistent with black hole unitarity~\cite{Almheiri:2020cfm} (see also~\cite{Matsuo:2023cmb}). We would be very keen to understand whether this loophole can be exploited in practice.

%\clearpage
\phantomsection
\addcontentsline{toc}{section}{\numberline{}Acknowledgements}
\section*{Acknowledgements}
I am grateful to Luca Ciambelli, Julian De Vuyst, Stefan Eccles, Laurent Freidel, Philipp H\"ohn, Krishna Jalan, Fabio Mele, Antony Speranza and Aron Wall for helpful and stimulating discussions.
This project/publication was made possible through the support of the ID\# 62312 grant from the John Templeton Foundation, as part of the \href{https://www.templeton.org/grant/the-quantum-information-structure-of-spacetime-qiss-second-phase}{\textit{`The Quantum Information Structure of Spacetime'} Project (QISS)}.~The opinions expressed in this project/publication are those of the author(s) and do not necessarily reflect the views of the John Templeton Foundation. Research at Perimeter Institute is supported in part by the Government of Canada through the Department of Innovation, Science and Economic Development and by the Province of Ontario through the Ministry of Colleges and Universities.

\appendix

\section{Useful lemmas}
\label{Appendix: lemmas}

\begin{lemma} \label{Lemma: crossed product center}
  Let $\mathcal{A}$ be a von Neumann algebra with an action by automorphisms $\alpha:\RR\to\operatorname{Aut}(\mathcal{A})$, $t\mapsto \alpha_t$. Suppose $\alpha_t$ is an outer automorphism of $z\mathcal{A}$ for all non-zero $z\in Z(\mathcal{A})^\alpha$ and all $t\ne 0$. Then the center of the crossed product algebra $\mathcal{A}\rtimes_\alpha \RR$ consists of central elements of $\mathcal{A}$ invariant under $\alpha_t$ for all $t$, i.e.\ $Z(\mathcal{A}\rtimes_\alpha \RR) = Z(\mathcal{A})^\alpha$.
\end{lemma}
\begin{proof}
  Without loss of generality we can assume $\mathcal{A}$ acts on a Hilbert space $\mathcal{H}$ with unitary operators $U(t)\in\mathcal{B}(\mathcal{H})$ such that $U(t)aU(t)^\dagger = \alpha_t(a)$. Let us use the following representation of the crossed product:
  \begin{equation}
    \mathcal{A}\rtimes_\alpha \RR \simeq \{a\otimes \mathds{1}_{L^2(\RR)}, u(t) = U(t)\otimes \lambda(t) \mid a\in \mathcal{A}, t\in \RR\}'',
  \end{equation}
  where $\lambda(t)$ is translation by $t$ on $L^2(\RR)$. Clearly $Z(\mathcal{A})^\alpha\otimes\mathds{1}_{L^2(\RR)}\subset Z(\mathcal{A}\rtimes_\alpha \RR)$.

  A general element of the crossed product may be written
  \begin{equation}
    A = \int_{-\infty}^\infty\dd{t} f(t)\qty(a(t) \otimes\mathds{1}_{L^2(\RR)}) u(t) \in \mathcal{A}\rtimes_\alpha \RR,
    \label{Equation: A integral}
  \end{equation}
  where $f$ is some distribution over $\RR$, and $a(t)\in\mathcal{A}$. Suppose $A\in Z(\mathcal{A}\rtimes_\alpha \RR)$. Then for all $b\in \mathcal{A}$ and $t'\in \RR$ we have
  \begin{equation}
    0 = [A,b\otimes\mathds{1}_{L^2(\RR)}] = \int_{-\infty}^\infty\dd{t} f(t)\qty((a(t)\alpha_t(b) - b a(t))\otimes\mathds{1}_{L^2(\RR)}) u(t)
  \end{equation}
  and
  \begin{equation}
    0 = [A,u(t')] = \int_{-\infty}^\infty\dd{t} f(t)\qty((a(t)-\alpha_{t'}(a(t)))\otimes\mathds{1}_{L^2(\RR)}) u(t+t'),
  \end{equation}
  which imply
  \begin{equation}
    f(t)\qty(a(t)\alpha_t(b)-b a(t)) = 0, \qquad f(t)\qty(a(t) - \alpha_{t'}(a(t))) = 0 \text{ for all }t,t'\in\RR, b\in\mathcal{A}.
  \end{equation}
  Suppose $f(t)\ne 0$ for some $t$. If $t=0$, these equalities give straight away that $a(0)\in Z(\mathcal{A})^\alpha$.
  If $t\ne 0$, the first equality implies that $a(t)^\dagger a(t)$ commutes with $\mathcal{A}$, and the second equality gives that $a(t)^\dagger a(t)$ is invariant under $\alpha$, so we must have $a(t)^\dagger a(t) = \Lambda(t)\in Z(\mathcal{A})^\alpha$. With a polar decomposition we can write $a(t) = V(t)\Lambda(t)^{1/2}$ for some unitary operator $V(t)\in\mathcal{A}$. Then the first equality above gives
  \begin{equation}
    \Lambda(t)\alpha_t(b) = \Lambda(t)V(t)^\dagger b V(t).
  \end{equation}
  But, unless $\Lambda(t)=0$, this contradicts the assumption that $\alpha_t$ is outer on $\Lambda(t)\mathcal{A}$. Therefore, we must have $a(t)=0$, so~\eqref{Equation: A integral} has no contributions away from $t=0$. Therefore, $A\in Z(\mathcal{A})^\alpha$, as required.
\end{proof}

\begin{lemma} \label{Lemma: faithful normal rep}
  Let $\mathcal{A}\subset\mathcal{B}(\mathcal{K})$ be a von Neumann algebra with a normal representation $r:\mathcal{A}\to\mathcal{B}(\mathcal{H})$ such that $r(z)\ne 0$ for all non-zero $z\in Z(\mathcal{A})$. Then $r$ is faithful.
\end{lemma}
\begin{proof}
  Since $r$ is normal, its kernel $\ker(r)$ is ultraweakly closed. It is also a two-sided ideal, since $r(abc)=r(a)r(b)r(c)=0$ if $b\in\ker(r)$.

  Any ultraweakly closed two-sided ideal $\mathcal{I}\subset\mathcal{A}$ is of the form $\mathcal{I} = e \mathcal{A}$ for some $e\in Z(\mathcal{A})$. To see this, let $\mathcal{K'}\subset\mathcal{K}$ be the space of states spanned by the images of operators in $\mathcal{I}$, and let $e$ be the projection onto $\mathcal{K}'$. Since $\mathcal{I}$ is ultraweakly closed, we have $e\in\mathcal{I}$. Moreover, $ex = x$ for any $x\in\mathcal{I}$. Since $\mathcal{I}$ is a two-sided ideal, $ea,ae\in\mathcal{I}$ for any $a\in \mathcal{A}$. Therefore, $ea=eae=ae$ for any $a\in\mathcal{A}$, i.e.\ $e\in Z(\mathcal{A})$ as required.

  By the assumption that $r(z)=0\implies z=0$ for $z\in Z(\mathcal{A})$, we have $e=0$. Therefore, $\ker(r) = \emptyset$.
\end{proof}

\printbibliography

\end{document}